\pdfoutput=1
\RequirePackage{ifpdf}
\ifpdf 
\documentclass[pdftex]{sigma}
\else
\documentclass{sigma}
\fi

\usepackage{bbm,mathrsfs, enumerate}

\usepackage{tikz}

\def\RR{{\mathbb R}}
\def\CC{{\mathbb C}}

\def\SS{{\mathbb S}}

\def\B{{\mathcal B}}
\def\C{{\mathcal C}}

\def\H{{\mathcal H}}

\def\K{{\mathcal K}}
\def\M{{\mathcal M}}

\def\P{{\mathcal P}}
\def\R{{\mathcal R}}

\def\U{{\mathcal U}}
\def\V{{\mathcal V}}

\def\a{\alpha}
\def\b{\beta}

\def\e{\varepsilon}

\def\k{\kappa}

\def\l{\lambda}
\def\L{{\mathrm L}}

\def\R{{\mathrm R}}

\def\x{\xi}
\def\uxi{\underline \xi}

\def\uxi{\underline \xi}
\def\ueta{\underline \eta}

\def\Ad{\operatorname{Ad}}
\def\id{{\rm id}}
\def\1{{\mathbbm 1}}
\def\Dom{{\mathscr D}}
\def\res{\operatorname{Res}}
\def\supp{\operatorname{supp}}
\def\<{\langle}
\def\>{\rangle}
\def\hardytwol{{H^2\big(\SS_{-\frac{\pi}3, 0}\big)}}
\def\hardyinfl{{H^\infty\big(\SS_{-\frac{\pi}3, 0}\big)}}
\def\hardytwou{{H^2\big(\SS_{0,\frac{\pi}3}\big)}}

\def\hardytwolf{{H^2(\SS_{-\pi, 0})}}
\def\hardyinflf{{H^\infty(\SS_{-\pi, 0})}}
\def\hardytwouf{{H^2(\SS_{0,\pi})}}
\def\hardyinfuf{{H^\infty(\SS_{0,\pi})}}

\def\striplf{\SS_{-\pi,0}}

\def\poincare{{\P^\uparrow_+}}

\def\halfshift{{\Delta^\frac1{12}}}

\def\shiftone{{\Delta_1^\frac16}}
\def\shiftione{{\Delta_1^{-\frac16}}}
\def\halfshiftone{{\Delta_1^\frac1{12}}}
\def\halfshiftione{{\Delta_1^{-\frac1{12}}}}
\def\fct{{\widetilde \phi}}

\newcommand{\bl}{\mathrm{Bl}}

\newcommand{\tout}{\mathrm{out}}
\newcommand{\tin}{\mathrm{in}}

\newcommand{\dom}{\operatorname{Dom}}
\newcommand{\im}{\operatorname{Im}}
\newcommand{\re}{\operatorname{Re}}

\numberwithin{equation}{section}

\newtheorem{Theorem}{Theorem}[section]
\newtheorem{Corollary}[Theorem]{Corollary}
\newtheorem{Lemma}[Theorem]{Lemma}
\newtheorem{Proposition}[Theorem]{Proposition}

\begin{document}


\newcommand{\arXivNumber}{1602.04696}

\renewcommand{\PaperNumber}{100}

\FirstPageHeading

\ShortArticleName{Bound State Operators and Wedge-Locality in Integrable Quantum Field Theories}

\ArticleName{Bound State Operators and Wedge-Locality\\ in Integrable Quantum Field Theories}

\Author{Yoh TANIMOTO~$^{\dag\ddag}$}

\AuthorNameForHeading{Y.~Tanimoto}

\Address{$^\dag$~Graduate School of Mathematical Sciences, The University of Tokyo,\\
\hphantom{$^\dag$}~3-8-1 Komaba Meguro-ku Tokyo 153-8914, Japan}
\EmailD{\href{mailto:hoyt@ms.u-tokyo.ac.jp}{hoyt@ms.u-tokyo.ac.jp}}
\URLaddressD{\url{http://www.mat.uniroma2.it/~tanimoto/}}

\Address{$^\ddag$~Institut f\"ur Theoretische Physik, G\"ottingen University,\\
\hphantom{$^\ddag$}~Friedrich-Hund-Platz~1, 37077 G\"ottingen, Germany}

\ArticleDates{Received February 19, 2016, in f\/inal form October 10, 2016; Published online October 19, 2016}

\Abstract{We consider scalar two-dimensional quantum f\/ield theories with a factorizing $S$-matrix which has poles in the physical strip. In our previous work, we introduced the bound state operators and constructed candidate operators for observables in wedges. Under some additional assumptions on the $S$-matrix, we show that, in order to obtain their strong commutativity, it is enough to prove the essential self-adjointness of the sum of the left and right bound state operators. This essential self-adjointness is shown up to the two-particle component.}


\Classification{81T05; 81T40; 81U40}

\usetikzlibrary{arrows}

\section{Introduction}
The goal of this series of papers \cite{CT15-1, CT15-2, CT-sine, Tanimoto15-1} is to construct more interacting two-dimensional integrable quantum f\/ield theories in the Haag--Kastler framework. The novelty of the recent progresses is that one f\/irst constructs observables localized in inf\/initely-extended wedge regions, then obtains compactly localized observables in an indirect way. In the preceding paper~\cite{CT15-1}, we constructed candidates for wedge-local observables, however, they have subtle domain properties and their strong commutativity remained open. In this paper, we provide a possible way to settle this question, which is a crucial step towards construction of Haag--Kastler nets.

The main strategy was explained in \cite{CT15-1} and we only brief\/ly summarize it. See also \cite{Lechner15} for an overview of the program. The operator-algebraic framework for quantum f\/ield theory (QFT), or a Haag--Kastler net \cite{Haag96}, concerns the collection of algebras of local observables. As in any mathematical framework for relativistic QFT, it is a very dif\/f\/icult problem to construct interacting examples of Haag--Kastler net. For certain integrable models with prescribed $S$-matrix, Schroer proposed \cite{Schroer99} to construct certain wedge-local observables, then obtain compactly localized observables as the intersection of two wedge-algebras. Some constructions of wedge-local algebras have been obtained for a class of $S$-matrices \cite{Lechner03, LS14, Tanimoto14-1}, as well as operator-algebraic constructions or deformations \cite{Alazzawi14,BT13,BLS11, DT11, Lechner12,Tanimoto12-2, Tanimoto14-1}. The f\/inal step to show the existence of local observables has been completed in certain cases, by showing modular nuclearity \cite{BL04, Lechner08}, or by split property~\cite{Tanimoto14-1}.

In~\cite{Lechner03}, Lechner constructed wedge-local f\/ields for scalar factorizing $S$-matrices without poles in the physical strip. We attempted to extend this construction to the cases with poles in~\cite{CT15-1}, however, certain spectral properties remained unsolved. More precisely, our candidates for wedge-local f\/ields were a slight modif\/ication of the f\/ields of \cite{Lechner03}, but contained a very singular term, the bound state operator, which has been introduced exactly in order to treat the poles. The strong commutativity of these candidates remained unclear. We studied the one-particle component of these operators in detail in \cite{Tanimoto15-1} and found convenient self-adjoint extensions. In the present work, we propose a strategy to prove the strong commutativity between the candidate operators.

Self-adjointness of unbounded operators is in general a hard problem. There are several criteria (e.g., see~\cite{RSII}), but we saw in \cite{CT15-1, Tanimoto15-1} that most of them do not apply to our situation: the reason is that the domain of self-adjointness depends very much on the observables and one cannot f\/ind either a single domain or a single dominating operator. We try to solve this problem by a front attack, namely, we should f\/ind an appropriate common domain for any pair of wedge-observables and show their strong commutativity. Thanks to the Driessler--Fr\"ohlich theorem, the problem of strong commutativity is reduced to the (essential) self-adjointness of the sum of the left and right bound state operators. We show this self-adjointness up to the two-particle component. It is worth mentioning that, because of all these technicalities, the candidate observable $\fct$ is no longer an operator-valued distribution.

This paper is organized as follows. In Section~\ref{preliminaries} we recall our approach to integrable QFT with bound states. Especially, we extend the domain of the bound state operators introduced in \cite{CT15-1} using the convenient extension found in~\cite{Tanimoto15-1}. Then in Section~\ref{towards} it is shown that the strong commutativity of wedge-observable candidates follows from the essential self-adjointness of the sum bound state operator. We do this by combining the Driessler--Fr\"ohlich theorem. With this choice of domain, the covariance of the candidate operators also follows, which allows one to proceed to Borchers triples. In Section~\ref{sa}, we show the essential self-adjointness of the one and two particle components of the sum bound state operator. We also present some observations on $n$-particle components which demonstrates how the poles and zeros in the $S$-matrix af\/fect the analytic properties of the observables.

\section{Preliminaries}\label{preliminaries}

\subsection{Wedge-local nets and Borchers triples}
Let us summarize the construction of \cite{CT15-1}. Our goal is to construct interacting Haag--Kastler nets in two dimensions. As an intermediate step, we need to construct nets of observables localized in wedges, or wedge-local nets (of von Neumann algebras) for short. This notion is equivalent to that of Borchers triples, which we will use in this work. We refer to \cite[Section~2.1]{CT15-1} for an overview of the construction program through Borchers triples.

A (two-dimensional) \textit{Borchers triple} $(\M,U,\Omega)$ consists of a von Neumann algebra $\M$, a~unitary representation of the proper orthochronous Poincar\'e group $\poincare$ and a vector $\Omega$ which satisfy the following conditions (see \cite[Section~2.1]{CT15-1}):
\begin{itemize}\itemsep=0pt
 \item $\Omega$ is cyclic and separating for $\M$;
 \item the restriction of $U$ to the translation subgroup $\RR^2$ has the joint spectrum included in $\overline{V_+} = \{(a_0,a_1)\colon a_0 \ge |a_1|\}$;
 \item $\M$ is covariant with respect to $U$, namely, $\Ad U(g)\M \subset \M$ for $g\in \poincare$ which preserves the standard right-wedge $W_\R$.
\end{itemize}
In \cite{CT15-1}, we constructed for test functions $f$, $g$ supported in $W_\L$, $W_\R$ respectively, unbounded operators $\fct(f)$, $\fct'(g)$ which commute weakly on a dense domain. We also showed that, if $\fct(f)$ and $\fct'(g)$ commute strongly, then $\Omega$ is cyclic and separating for $\M := \big\{e^{i\fct'(g)}\colon \supp g\big\}''$ (strong commutativity is important for the separating property). Furthermore, if we can show the covariance of $\fct$ with respect to $U$, then the covariance of the Borchers triple follows and the construction is complete. We will investigate strict locality, namely whether this Borchers triple def\/ines a Haag--Kastler net, in a separate paper~\cite{CT15-3} (actually, for this purpose the def\/inition of~$\fct$ must be slightly extended as we will do in Section~\ref{one}). Therefore, our most urgent problem is the strong commutativity.

We have $\fct(f) = \phi(f) + \chi(f)$, where $\phi(f)$ is very well under control, while $\chi(f)$ is a new sort of operator and we investigated its one-particle component $\chi_1(f)$ in detail. We review the construction of these operators below.

\subsection[Factorizing $S$-matrix with poles]{Factorizing $\boldsymbol{S}$-matrix with poles}\label{s-matrix}
As an input, we f\/ix a {\it two-particle scattering function} $S$,
which is at f\/irst a meromorphic function on $\RR + i(0,\pi)$, the region we call the {\it physical strip},
and later extended to $\CC$ which satisf\/ies the following properties:
\begin{enumerate}[{(S}1{)}]\itemsep=0pt
 \item\label{ax:unitarity} {\it Unitarity.} $S(\theta)^{-1} = \overline{S(\theta)}$ for $\theta \in \RR$.
 \item\label{ax:hermitian} {\it Hermitian analyticity.} $S(-\theta) = S(\theta)^{-1}$ for $\theta \in \RR$.
 \item\label{ax:crossing} {\it Crossing symmetry.} $S$ has a meromorphic extension in $\CC$ and satisf\/ies $S(\zeta) = S(\pi i - \zeta)$.
 \item\label{ax:bootstrap} {\it Bootstrap equation.} $S(\zeta) = S\big(\zeta + \frac{\pi i}3\big)S\big(\zeta - \frac{\pi i}3\big)$.
 \item\label{ax:residue} {\it Positive residue.} $S$ has a simple pole at $\frac{2\pi i}3$ and $R := {\underset{\zeta = \frac{2\pi i}3} \res S(\zeta) \in i\RR_+}$. Except this and the pole at $\frac{\pi i}3$, there is no pole in the physical strip $\RR + i(0,\pi)$ and $S(\zeta)$ is bounded in the complement of the union of neighborhoods of these poles in the physical strip.
 \item\label{ax:atzero} {\it Value at zero.} $S(0) = -1$.
 \item\label{ax:re} {\it A bound on the real part.} $0 \le \re S\big(\theta + \frac{2\pi i}3\big)$ $\big({=} \re S\big(\theta + \frac{\pi i}3\big)\big)$. 
 \item \label{ax:bound} \textit{A bound on the modulus.} $\big|S\big(\theta + \frac{\pi i} 6\big)\big| \le 1$.
 \item \label{ax:regularity} \textit{Regularity.} $S$ has only f\/initely many zeros in the physical strip and
 there is $\k > 0$ such that $\|S\|_\k := \sup\{|S(\zeta)|\colon \zeta \in \RR + i(-\k,\k)\} < \infty$.
\end{enumerate}
We showed in~\cite{CT15-1} that~(S\ref{ax:atzero}) follows from (S\ref{ax:unitarity})--(S\ref{ax:residue}). It turns out crucial in the consideration of domains of our operators. We also classif\/ied all functions which satisfy these conditions \cite[Appendix~A]{CT15-1}. We show that (S\ref{ax:bound}) follows from (S\ref{ax:unitarity})--(S\ref{ax:residue}) in Appendix~\ref{bound}.

(S\ref{ax:re}) and (S\ref{ax:regularity}) are new requirements. The property (S\ref{ax:re}) had not appeared either in the literature on form factor program or in our previous works \cite{CT15-1, CT15-2}. Yet, it is necessary in order to apply our methods, for some properties of the sum bound state operators. The last assump\-tion~(S\ref{ax:regularity}) will be necessary also in the proof of modular nuclearity (cf.~\cite{Lechner08}), and we use it here already for wedge-observables, in Lemma~\ref{lm:symmetric-cont}. With the f\/initeness of zeros in the physical strip, the f\/initeness of $\|S\|_\k$ is equivalent to the absence of the singular part in~$S$ (cf.~\cite[Appendix~A]{LW11}). From the condition~(S\ref{ax:atzero}), there must be at least one zero on the interval $i(0,\frac\pi 3)$, hence $\k < \frac\pi 3$, as we saw in~\cite{CT15-1}.

We check in Appendix~\ref{existenceS} that there are concrete examples of $S$ which satisfy all these requirements. In the following, we mark explicitly the points where~(S\ref{ax:re}) is needed.

\subsection{The wedge-local observables}\label{fields}
Some materials here are parallel to those of \cite{Lechner03}. Assume that our model has only a single species of particle with mass $m > 0$. We f\/ix a two-particle scattering function $S$, which satisf\/ies the properties explained in Section~\ref{s-matrix}.

Our Hilbert space is as follows. The one-particle space $\H_1$ is simply $L^2(\RR)$ with the Lebesgue measure. We consider $n$-particle space $\H_1^{\otimes n}$. An element of $\H_1^{\otimes n}$ can be considered as a function in $L^2(\RR^n)$. On each of these spaces, there is the representation of the symmetric group $\mathfrak{S}_n$:
\begin{gather*}
 (D_n(\tau_k)\Psi_n)(\theta_1,\dots,\theta_n) := S(\theta_{k+1}-\theta_k)\Psi_n(\theta_1,\dots, \theta_{k+1}, \theta_k,\dots, \theta_n),
\end{gather*}
where $\tau_k$ is the transposition of $k$ and $k+1$. This indeed extends to a unitary representation of $\mathfrak{S}_n$. We say that $\Psi_n \in \H_1^{\otimes n}$ is {\it $S$-symmetric} if $D_n(\sigma)\Psi_n = \Psi_n$ for $\sigma \in \mathfrak{S}_n$. This notion is equivalent to that for each~$k$
\begin{gather*}
 \Psi_n(\theta_1,\dots, \theta_n) = S(\theta_{k+1} - \theta_k)\Psi_n(\theta_1, \dots, \theta_{k+1},\theta_k,\dots, \theta_n).
\end{gather*}
We denote the orthogonal projection onto the $S$-symmetric functions by $P_n$ and the space of $S$-symmetric functions by $\H_n = P_n\H_1^{\otimes n}$. The physical Fock space is the direct sum
\begin{gather*}
 \H = \bigoplus_n \H_n,
\end{gather*}
where $\H_0 = \CC$. Let us also introduce an auxiliary Hilbert space, the unsymmetrized Fock space, which is $\H^\Sigma = \bigoplus_n \H_1^{\otimes n}$. By putting $P_S = \bigoplus_n P_n$, we have $\H = P_S\H^\Sigma$.

For $\psi \in \H_1$, the unsymmetrized creation and annihilation operators are def\/ined on the algebraic direct sum of $\H_1^{\otimes n}$ by
\begin{gather*}
 \big(a^\dagger(\psi)\Psi\big)_n(\theta_1,\dots, \theta_n) = \sqrt{n} \psi(\theta_1)\Psi_{n-1}(\theta_2,\dots, \theta_n),\qquad
 a(\psi) = \big(a^\dagger(\psi)\big)^*.
\end{gather*}
By our convention $a^\dagger(\psi)$ is linear in $\psi$, while $a(\psi)$ is antilinear. The $S$-symmetrized creation and annihilation operators are $z^\dagger(\psi) = P_Sa^\dagger(\psi)P_S$, $z(\psi) = P_Sa(\psi)P_S$.

The one-particle CPT operator $J_1$ is given by $(J_1\psi)(\theta) = \overline{\psi(\theta)}$, and the $n$-particle compo\-nent~$J_n$ is
\begin{gather*}
 (J_n\Psi_n)(\theta_1,\dots, \theta_n) = \overline{\Psi_n(\theta_n,\dots, \theta_1)}.
\end{gather*}
It is easy to see that $J_n$ commutes with $P_n$, hence its preserves the subspace $\H_n$. The full CPT operator is $J = \bigoplus_n J_n$, where $J_0$ is the complex conjugation on $\H_0 = \CC$.

For a test function $f$ in $\RR^2$, we def\/ine
\begin{gather*}
 f^\pm(\theta) = \int dx\, f(x)e^{\pm ip(\theta)\cdot x},\qquad
 p(\theta) = m\left(\begin{matrix} \cosh \theta \\ \sinh \theta \end{matrix}\right).
\end{gather*}
As $f$ is compactly supported, its Fourier transform decays fast, and so do~$f^\pm$, therefore, they belong to~$\H_1$. On the level of test functions, we def\/ine the CPT transformation by $f_j(x) = \overline{f(-x)}$.

If $S$ had no pole, then Lechner took the f\/ield $\phi(f) = z^\dagger(f^+) + z(J_1 f^-)$ and the ref\/lected f\/ield $\phi'(g) = z'^\dagger(g^+) + z'(J_1g^+)$, where $z'^\dagger(\psi) = Jz^\dagger(J_1\psi)J$ and $z'(\psi) = Jz(J_1\psi)J$ and proved that they were wedge-local~\cite{Lechner03}: namely, if $\supp f \subset W_\L$ and $\supp g \subset W_\R$, then $\phi(f)$ and $\phi'(g)$ strongly commute.

\subsubsection*{Extension to Hardy space vectors}

Actually, it is possible and also natural to consider $\phi$, $z^\dagger$ and~$z$ as a function of $f^+$, and in this sense we can extend them to any $\xi \in \hardytwouf\cap \hardyinfuf$. For $a \le 0 \le b$, recall the def\/inition of the Hardy space:
\begin{gather*}
 H^2(\SS_{a,b}) = \left\{\xi \in \H_1\colon \xi(\theta) \mbox{ has an } L^2\mbox{-bounded analytic continuation to } \SS_{a,b}\right\}.
\end{gather*}
In other words, a function $\xi \in H^2(\SS_{a,b})$ is an analytic function on $\SS_{a,b}$ whose $L^2$-norms on horizontal lines $\xi(\,\cdot + i\l)$ are bounded uniformly for $\l \in (a,b)$. It is well-known that for $\l \to a,b$ the limits $\xi(\,\cdot + i\l)$ exist in the $L^2$-sense \cite[Corollary~III.2.10]{SW71}. As $0\in [a,b]$, we can idenf\/ity $\H_1 = L^2(\RR)$ as a subspace of $H^2(\SS_{a,b})$ by considering the value $\xi(\theta)$, $\theta \in \RR$. We denote the $L^2$-norm by~$\|\xi\|$, while the Hardy space norm is $\|\xi\|_2 := \sup\limits_{\l \in (a,b)}\|\xi(\,\cdot + i\l)\|$.

Similarly, $H^\infty(\SS_{a,b})$ is the space of bounded analytic functions on $\SS_{a,b}$ and we def\/ine the norm $\|\xi\|_\infty := \sup\limits_{\zeta \in \SS_{a,b}}|\xi(\zeta)|$.

The operators $z^\dagger(\xi)$ and $z(\xi)$ are def\/ined for vectors $\xi \in \H_1$, hence with the understanding that $\H_1 = L^2(\RR) \supset \hardytwouf$ by taking the boundary value at $\im \zeta = 0$, $z^\dagger$ and~$z$ are also def\/ined on $\hardytwouf$. The reality condition $f(a) = \overline{f(a)}$ on a test function $f$ translates to an element $\xi \in \hardytwouf$ as $\xi(\zeta + i\pi) = \overline{\xi(\zeta)}$. It is straightforward to see that the proof of wedge-locality for~$S$ without pole~\cite{Lechner03} works for such $\xi$'s: the Cauchy theorem works for $H^2$-functions (see, e.g., \cite[Lemma~B.2]{CT15-1}).

As we are interested in cases where $S$ has poles, we introduce the bound state operator $\chi(\xi)$ (cf.~in~\cite{CT15-1} we considered it as associated with $f^+$, instead of $\xi$). We associate an operator $\chi(\xi)$ to $\xi \in \hardytwouf\cap\hardyinfuf$ which satisf\/ies $\overline{\xi(\theta)} = \xi(\theta + \pi i)$ by
\begin{gather*}
 (\chi_1(\xi)\Psi)(\theta) := \sqrt{2\pi |R|}\xi\left(\theta + \frac{\pi i}3\right)\Psi\left(\theta - \frac{\pi i}3\right), \\
 \chi_n(\xi) := nP_n \left(\chi_1(\xi)\otimes\1\otimes\cdots \otimes\1\right)P_n, \qquad
 \chi(\xi) := \bigoplus_n \chi_n(\xi),
\end{gather*}
where $\chi_0(\xi) = 0$ and $R = \underset{\zeta = \frac{2\pi i}3} \res S(\zeta)$. Of course, $\chi_1(\xi)$ cannot be def\/ined on arbitrary vectors, as it involves an analytic continuation of $\Psi$. In Section~\ref{one}, we will specify a domain of self-adjointness of $\chi_1(\xi)$
and def\/ine~$\chi_n(\xi)$ accordingly. It follows that $\chi(\xi)$ is symmetric as in \cite[Proposition~3.1]{CT15-1}. For $\eta \in \hardytwolf\cap\hardyinflf$, $\overline{\eta(\theta)} = \eta(\theta-\pi i)$ we introduce also
\begin{gather*}
 (\chi'_1(\eta)\Psi)(\theta) := \sqrt{2\pi |R|}\eta\left(\theta - \frac{\pi i}3\right)\Psi\left(\theta + \frac{\pi i}3\right), \\
 \chi_n(\eta) := nP_n \left(\1\otimes\cdots \otimes\1\otimes \chi'_1(\eta)\right)P_n, \qquad
 \chi'(\eta) := \bigoplus_n \chi'_n(\eta).
\end{gather*}
It holds that $\chi'_1(\eta) = J_1\chi_1(J_1\eta)J_1$, $\chi'_n(\eta) = J_n\chi_n(J_1\eta)J_n$ and $\chi'(\eta) = J\chi(J_1\eta)J$.

Our candidate operators for wedge-observables are $\fct(\xi) := \phi(\xi) + \chi(\xi)$ and $\fct'(\eta) := \phi'(\eta) + \chi'(\eta)$. We will show that they weakly commute on the intersection of their domains and propose a way to prove their strong commutativity.

\subsection{The bound state operator}\label{one}
We studied the one-particle component $\chi_1(\xi)$ in detail in \cite{Tanimoto15-1}: actually, we studied the domain $\hardytwolf$ rather than $\hardytwol$, but the adaptation is straightforward. The important lessons obtained there are that $\chi_1(\xi)$, when def\/ined on $\hardytwol$, does not always have a nice self-adjoint extension and even if it does, its domain varies as $\xi$ varies. This is due to the fact that the def\/iciency subspaces of $\chi_1(\xi)$ is closely related with the zeros and the decay rate of $\xi$, and $\chi_1(\xi)$ does not have any self-adjoint extension if~$\xi$ has an odd number of zeros and does not decay suf\/f\/iciently fast as $\re \zeta \to \pm \infty$. And even if $\xi$ has an even number of zeros, it is not obvious which extension is the suited one for our purpose of f\/inding quantum observables. Therefore, we restrict ourselves to the certain subclass of functions.

Recall that one can def\/ine the analytic continuation operator\footnote{We use the symbol $\Delta$, because if we could prove the strong commutativity we will be discussing in this paper, then $\Delta_1$ will turn out to be the one-particle component of the modular operator
of the wedge-algebra with respect to the vacuum~\cite{CT15-3}.}
\begin{gather*}
 \dom\big(\shiftone\big) = \hardytwol, \qquad \big(\shiftone\Psi\big)(\theta) = \Psi\left(\theta - \frac{\pi i}3\right),
\end{gather*}
namely, any vector $\Psi \in \hardytwol$ has an $L^2$-boundary value $\Psi(\theta)$ and $\Psi(\theta - \frac{\pi i}3)$ and it is considered as an element of~$L^2(\RR)$ by taking~$\Psi(\theta)$, while~$\shiftone\Psi$ is def\/ined to be the other boundary value $\Psi(\theta - \frac{\pi i}3)$. This is self-adjoint, indeed unitarily equivalent to the multiplication by the exponential function through Fourier transform \cite[Appendix~A]{Tanimoto15-1}.

A property which we will use repeatedly is the ``Cauchy theorem'' in the following form: since~$\Delta^{\frac t{2\pi}}$ is self-adjoint on $H^2(\SS_{-t,0})$, we have $\<\Phi, \Delta^{\frac t{2\pi}}\Psi\> = \<\Delta^{\frac t{2\pi}}\Phi, \Psi\>$ for $\Phi, \Psi \in H^2(\SS_{-t,0})$. By noting that $\overline{\Phi} \in H^2(\SS_{0,t})$, where $\overline{\Phi}(z) := \overline{\Phi(\overline z)}$, we can write this as
\begin{gather*}
 \<\Phi, \Delta^{\frac t{2\pi}}\Psi\> = \int d\theta\,\overline{\Phi(\theta)}\Psi(\theta - it)
 = \int d\theta\,\overline{\Phi}(\theta)\Psi(\theta - it)
 = \int d\theta\,\overline{\Phi}(\theta + it)\Psi(\theta) \\
\hphantom{\<\Phi, \Delta^{\frac t{2\pi}}\Psi\>}{} = \int d\theta\,\overline{\Phi(\theta - t)}\Psi(\theta) = \<\Delta^{\frac t{2\pi}}\Phi, \Psi\>,
\end{gather*}
namely, we can shift the argument in the integral. An analogous formula is valid even if the functions have more variables, by considering the operator $\Delta^{\frac t{2\pi}}\otimes \1$ with the identity operator~$\1$ on the complementary Hilbert space. To apply this, it is important to check that the involved vectors belong to the domain of the operator $\Delta^{\frac t{2\pi}}\otimes \1$.

We learned in \cite[Section 5.1]{Tanimoto15-1} that, if $\xi$ is the square of another analytic function $\xi = \uxi^2$, $\uxi \in \hardytwouf\cap\hardyinfuf$ and if $\uxi$ satisf\/ies the reality condition $\uxi(\zeta + i\pi) = \overline{\uxi(\zeta)}$ (which implies that $\uxi(\zeta + \frac{\pi i}3) = \overline{\uxi(\zeta + \frac{2\pi i}3)}$), then $\chi_1(\xi)$ has a canonical self-adjoint extension, which we denote by the same symbol for simplicity. More precisely, for $\xi$ as above, we can f\/ind a function $\xi_0 \in H^\infty(\SS_{\frac\pi 3,\frac{2\pi}3})$ such that $|\xi_0(\theta + \frac{2\pi i}3)| = 1$ (almost everywhere) and
\begin{gather}\label{eq:consistency}
 \xi_0\left(\theta + \frac{\pi i}3\right)\overline{\xi_0\left(\theta + \frac{2\pi i}3\right)} = \sqrt{2\pi|R|}\xi\left(\theta + \frac{\pi i}3\right).
\end{gather}
Now, the multiplication operator $M_{\xi_0(\cdot\, + \frac{2\pi i}3)}$, which is unitary, corresponds to the unitary ope\-ra\-tor of \cite[Theorem~5.1]{Tanimoto15-1} (with the notation replacing $\xi$ with $f$ and the scaling between the domains $\hardyinfl$ and $\hardyinflf$).
This unitary operator $M_{\xi_0(\cdot\, + \frac{2\pi i}3)}$ implements the unitary equivalence between $\shiftone$ and our bound state operator
\begin{gather*}
 \dom(\chi_1(\xi)) := M_{\xi_0(\cdot\, + \frac{2\pi i}3)}^*\hardytwol, \\
 (\chi_1(\xi)\Psi)(\theta) := \sqrt{2\pi |R|}\xi\left(\theta + \frac{\pi i}3\right)\Psi\left(\theta - \frac{\pi i}3\right),
\end{gather*}
or in other words, $ \chi_1(\xi) = M_{\xi_0(\cdot\, + \frac{2\pi i}3)}^*\shiftone M_{\xi_0(\cdot\, + \frac{2\pi i}3)}$, which follows straightforwardly from the property~\eqref{eq:consistency}. The point is that $\xi_0(\theta+\frac{2\pi i}3)\Psi(\theta) \in \hardytwol$ and the analytic continuation $\theta \to \theta + \frac{\pi i}3$ does not see any pole or non-$L^2$ behavior on this domain.

The $n$-particle component $\chi_n(\xi)$ is def\/ined by $\chi_n(\xi) := P_n(\chi_1(\xi)\otimes\1\otimes\cdots\otimes\1)P_n$. This means that $\dom(\chi_n(\xi)) = \{\Psi_n \in \H_n\colon \Psi_n \in \dom(\chi_1(\xi)\otimes\1\otimes\cdots\otimes\1)\}$. As the expression of $\chi_1(\eta)$ is same as~\cite{CT15-1} (although the domain is larger), we have also the following alternative expression (we can do this because the continuation
$\theta \to \theta + \frac{\pi i}3$ is always accompanied by the multiplication by $\xi_0(\theta + \frac{2\pi i}3)$, hence one does not encounter non-$L^2$ problem by def\/inition):
\begin{gather}\label{eq:chi-n}
 (\chi_n(f)\Psi_n)(\theta_1,\dots,\theta_n)\\
 = \sqrt{2\pi|R|}\sum_{1\le k \le n}
 \prod_{1\le j \le k-1} S\left(\theta_k-\theta_j + \frac{\pi i}3\right)
 \xi\left(\theta_k + \frac{\pi i}3\right)
 \Psi_n\left(\theta_1,\dots,\theta_k - \frac{\pi i}3,\dots,\theta_n\right).\nonumber
\end{gather}

We can rephrase all this as follows: the domain of $\chi_1(\xi)$ consists of functions $\Psi_1$ such that $\xi_0(\theta + \frac{2\pi i}3)\Psi_1(\theta)$ is the boundary value of an element in $\hardytwol$. The domain of $\chi_n(\xi)$ consists of functions $\Psi_n$ such that $P_n\Psi_n$ belongs to the domain of $\chi_1(\xi)\otimes\1\otimes\cdots\otimes\1$. In particular, $\xi_0(\theta_1 + \frac{2\pi i}3)(P_n\Psi_n)(\theta_1,\dots,\theta_n)$ can be viewed as an $L^2(\RR^{n-1})$-valued analytic function in $\theta_1$. By $S$-symmetry, $\xi_0(\theta_k + \frac{2\pi i}3)(P_n\Psi_n)(\theta_1,\dots,\theta_n)$ is also analytic in other variables $\theta_k$, but as the $S$-factor has poles, hence it belong to $\1\otimes\cdots\otimes\1\otimes\underset{k\text{-th}}{\Delta_1^\frac{\k}{2\pi}}\otimes\1\otimes\cdots\otimes\1$, where the boundedness of $S$ is assured in $\RR + i(-\k,\k)$ as~(\ref{ax:regularity}).

Similarly, for each $\eta = \ueta^2$ where $\ueta\in\hardytwolf\cap\hardyinflf$ such that $\ueta(\theta - i\pi) = \overline{\ueta(\theta)}$, there is $\eta_0 \in H^\infty(\SS_{-\frac{2\pi}3, -\frac{\pi}3})$ such that $|\eta_0(\theta - \frac{2\pi i}3)| = 1$ and
\begin{gather*}
 \eta_0\left(\theta - \frac{\pi i}3\right)\overline{\eta_0\left(\theta - \frac{2\pi i}3\right)} = \sqrt{2\pi|R|}\eta\left(\theta - \frac{\pi i}3\right).
\end{gather*}
We def\/ine
\begin{gather*}
 \dom(\chi'_1(\eta)) = M_{\eta_0(\cdot\, - \frac{2\pi i}3)}^*\hardytwou, \\
 (\chi'_1(\eta)\Psi)(\theta) := \sqrt{2\pi |R|}\eta\left(\theta - \frac{\pi i}3\right)\Psi\left(\theta + \frac{\pi i}3\right),
\end{gather*}
where $M_{\eta_0(\cdot\, - \frac{2\pi i}3)}$ is unitary, and this is equivalent to $\chi'_1(\eta) = M_{\eta_0(\cdot\, - \frac{2\pi i}3)}^*\shiftione M_{\eta_0(\cdot\, - \frac{2\pi i}3)}$. This operator $\chi'(\eta)$ has an alternative expression as $\chi(f)$ does:
\begin{gather*}
 (\chi'_n(\eta)\Psi_n)(\theta_1,\dots,\theta_n) = \sqrt{2\pi|R|}\sum_{1\le k \le n} \prod_{n-k+2 \le j \le n}\! S\left(\theta_j- \theta_{n-k+1} + \frac{\pi i}3\right)
 \eta\left(\theta_{n-k+1} - \frac{\pi i}3\right) \\
\hphantom{(\chi'_n(\eta)\Psi_n)(\theta_1,\dots,\theta_n) =}{} \times\Psi_n\left(\theta_1,\dots,\theta_{n-k+1} + \frac{\pi i}3,\dots,\theta_n\right).
\end{gather*}

\section{Towards strong commutativity}\label{towards}
In this Section we show that, in order to construct the Borchers triple associated with $S$ with poles, it is enough to prove that $\chi(\xi)+\chi'(\eta)$ is essentially self-adjoint.

\subsection{A general criterion}
This Section is technically independent from QFT.

For unbounded operators, strong commutativity is a hard problem. It may fail even for a~two self-adjoint operators which commute on a common invariant core \cite[Section~10]{Nelson59}, \cite[Section~VIII.5, Example~1]{RSII}. We also exhibit more examples of weakly-commuting but not strongly-commuting operators in Appendix~\ref{counterexample}. A suf\/f\/icient condition which is used very widely in the context of QFT (e.g.,~\cite{ BT13, CKLW15, GJ87}) is the Driessler--Fr\"ohlich theorem~\cite{DF77}. The theorem says roughly that, if there is a positive self-adjoint operator~$T$ (``Hamiltonian'') which can nicely bound the weakly commuting symmetric operators~$A$,~$B$, and their commutators with~$T$, then~$A$ and~$B$ are actually self-adjoint and commute strongly (see Appendix~\ref{DF} for the precise statements).

The trouble in our situation is that the physical Hamiltonian is not strong enough to estimate our candidates which contain the bound state operator $\chi(\xi)$: the Hamiltonian is the multiplication operator, while~$\chi(\xi)$ contains an analytic continuation. We have no idea for any other operator which dominates $\chi(\xi)$ ``nicely'', and anyway there cannot be a single operator which bounds all $\chi(\xi)$'s with dif\/ferent $\xi$'s, as $\chi_1(\xi)$'s have already dif\/ferent domains of self-adjointness. In such a case, the question of strong commutativity is considerably dif\/f\/icult (cf.~\cite{CKLW15}).

Yet, the Driessler--Fr\"ohlich theorem can be used to reduce the problem of strong commutativity to self-adjointness of a certain positive operator.
\begin{Proposition}\label{pr:general}
Let $A$, $B$ and $T_0$ be symmetric operators. Assume that $T:=A+B+T_0 \ge \1$ and is essentially self-adjoint on its natural domain $\dom(T) := \dom(A)\cap\dom(B)\cap\dom(T_0)$, and there are $\Dom \subset \dom(T)$ which is a core of $T$ and $c_1,c_2,c_3 > 0$ such that
 \begin{enumerate}[{\rm (1)}]\itemsep=0pt
 \item\label{gn1} $\|A\psi\| \le c_1\|T\psi\|$ and $\|B\psi\| \le c_1\|T\psi\|$ for all $\psi \in \Dom$,
 \item\label{gn2} for all $\varphi,\psi \in \Dom$,
 \begin{itemize}\itemsep=0pt
 \item $|\<T_0\psi,A\varphi\> - \<A\psi, T_0\varphi\>| \le c_2\big\|T^\frac12\psi\big\|\cdot \big\|T^\frac12\varphi\big\|$,
 \item $|\<T_0\psi,B\varphi\> - \<B\psi, T_0\varphi\>| \le c_2\big\|T^\frac12\psi\big\|\cdot \big\|T^\frac12\varphi\big\|$,
 \end{itemize}
 \item\label{gn3} for all $\varphi,\psi \in \Dom$,
 \begin{itemize}\itemsep=0pt
 \item $|\<T_0\psi,A\varphi\> - \<A\psi, T_0\varphi\>| \le c_3 \|\psi\|\cdot \|T\varphi\|$,
 \item $|\<T_0\psi,B\varphi\> - \<B\psi, T_0\varphi\>| \le c_3 \|\psi\|\cdot \|T\varphi\|$,
 \end{itemize}
 \item\label{gn4} $\<A\psi, B\varphi\> = \<B\psi, A\varphi\>$ for all
 $\psi, \varphi \in \Dom$.
 \end{enumerate}
 Then $A$ and $B$ are essentially self-adjoint on any core of $T$ and they strongly commute.
\end{Proposition}
\begin{proof}
 As we have $\<A\psi, B\varphi\> = \<B\psi, A\varphi\>$ by assumption and $T=A+B+T_0$ by def\/inition, it follows that for all $\varphi,\psi \in \Dom$
 \begin{itemize}\itemsep=0pt
 \item $|\<T\psi,A\varphi\> - \<A\psi, T\varphi\>| \le c_2\big\|T^\frac12\psi\big\|\cdot \big\|T^\frac12\varphi\big\|$,
 \item $|\<T\psi,B\varphi\> - \<B\psi, T\varphi\>| \le c_2\big\|T^\frac12\psi\big\|\cdot \big\|T^\frac12\varphi\big\|$,
 \item $|\<T\psi,A\varphi\> - \<A\psi, T\varphi\>| \le c_3 \|\psi\|\cdot \|T\varphi\|$,
 \item $|\<T\psi,B\varphi\> - \<B\psi, T\varphi\>| \le c_3 \|\psi\|\cdot \|T\varphi\|$.
 \end{itemize}
Hence we can apply Theorem \ref{th:commutator}, by taking $\overline{T}$ as the positive self-adjoint operator.
\end{proof}

The point of this proposition is that, instead of using a single ``Hamiltonian'' for all cases, we take the operator $T=A+B+T_0$ which depends on~$A$ and~$B$. Of course, as weak commutativity does not imply strong commutativity, the essential self-adjointness hypothesis of~$T$ and the estimate of~$A$,~$B$ by $T$ are crucial.

We have the following case in mind: $A = \fct(\xi)$, $B = \fct'(\eta)$, $T_0 = c(N+\1)$, where $N$ is the number operator on the $S$-symmetric Fock space and $c \in \RR_+$. Furthermore, we will see in Section~\ref{strong} that $T = A + B + T_0$ is a Kato--Rellich perturbation of the operator $\chi(\xi) + \chi'(\eta) + c(N+\1)$. This last operator preserves each $n$-particle subspace, and the question of (essential) self-adjointness is reduced to each~$n$.

This self-adjointness is far from obvious. By a slight modif\/ication, the self-adjointness of a~similar operator easily fails (see Appendix~\ref{ce-sum}).

\subsection{Proof of weak commutativity on the intersection domain}\label{weak}
We show that $\fct(\xi)$ and $\fct'(\eta)$ commute weakly on $\dom(\fct(\xi))\cap \dom(\fct'(g))$. Thanks to our choice of the extension (Section~\ref{one}), this will be a straightforward modif\/ication of the proof of~\cite{CT15-1}, with the help of some techniques in complex analysis of several variables in Appendix~\ref{S-symmetric}.

\begin{Proposition}\label{pr:weak}
Let $\underline\xi \in \hardyinfuf$, $\underline\xi(\theta + \pi i) = \overline{\underline\xi(\theta)}$ and $\underline\eta \in \hardyinflf$, $\underline\eta(\theta - \pi i) = \overline{\underline\eta(\theta)}$ for $\theta \in \RR$ and assume that $\xi = \underline\xi^2 \in \hardytwouf$, $\eta = \underline\eta^2 \in \hardytwolf$. Then, for each $\Psi,\Phi \in \dom(\fct(\xi))\cap\dom(\fct'(\eta))$, it holds that
 \begin{gather*}
 \big\<\fct(\xi)\Phi, \fct'(\eta)\Psi\big\> = \big\<\fct'(\eta)\Phi, \fct(\xi)\Psi\big\>.
 \end{gather*}
\end{Proposition}
\begin{proof}
We may assume that $\Phi$, $\Psi$ have f\/initely many non-zero components and since we are considering the algebraic direct sum $\underline{\bigoplus}_n \dom(\chi_n(\xi)) \subset \dom(\phi(\xi))$, $\underline{\bigoplus}_n \dom(\chi_n(\eta)) \subset \dom(\phi'(\eta))$, therefore, we can compute the action of the operators as the sum:
 \begin{gather*}
 \fct(\xi) = \phi(\xi) + \chi(\xi) = z^\dagger(\xi) + \chi(\xi) + z(\xi), \\
 \fct'(\eta) = \phi'(\eta) + \chi'(\eta) = z'^\dagger(\eta) + \chi'(\eta) + z'(\eta).
 \end{gather*}
We may also assume that $\Phi$ and $\Psi$ are already $S$-symmetric, although we use the unsymmetrized Fock space $\H^\Sigma$ in the intermediate steps. We are going to follow the steps of our previous result and compute the commutator term by term. Since the most of the computations are same as the previous ones, we will be brief and indicate only the points where care is needed, and refer the reader to the proof of~\cite[Theorem~3.4]{CT15-1}.

{\it The commutator $[\chi(\xi), z'(\eta)]$.} We can compute this commutator as operators as before~\cite{CT15-1}. Although we have an extended domain for $\chi(\xi)$, the expression \eqref{eq:chi-n} is valid as we explained before. By computing the operators term by term, we get
 \begin{gather*}
([\chi(\xi),z'(\eta)]\Psi_n)(\theta_1,\dots,\theta_{n-1}) \\
= -\sqrt{2\pi|R|n}\int d\theta\, \overline{\eta(\theta)}\prod_{1\le j \le n-1} S\left(\theta-\theta_j + \frac{\pi i}3\right)
 \xi\left(\theta + \frac{\pi i}3\right) \Psi_n\left(\theta_1,\dots,\theta_{n-1},\theta - \frac{\pi i}3\right) \\
= -\sqrt{2\pi|R|n}\int d\theta\, \eta(\theta - \pi i)\prod_{1\le j \le n-1} S(\theta - \theta_j)
 \xi\left(\theta + \frac{\pi i}3\right) \Psi_n\left(\theta - \frac{\pi i}3, \theta_1,\dots,\theta_{n-1}\right),
\end{gather*}
where in the last equation we used also the symmetry property $\overline{\eta(\theta)} = \eta(\theta-i\pi)$ and the $S$-symmetry of~$\Psi_n$.

Let us look at the integrand without the $S$-factor and shift the f\/irst variable $\theta$ by $\frac{\pi i}3$: $\eta(\theta - \frac{2\pi i}3)\xi\left(\theta + \frac{2\pi i}3\right) \Psi_n\left(\theta, \theta_1,\dots,\theta_{n-1}\right)$. Recall that:
 \begin{itemize}\itemsep=0pt
 \item $\xi \in \hardyinfuf$, $\eta \in \hardyinflf$,
 \item $\overline{\xi_0\left(\theta + \frac{\pi i}3\right)}\xi_0\left(\theta + \frac{2\pi i}3\right) =
 \sqrt{2\pi|R|}\xi\left(\theta + \frac{2\pi i}3\right)$, $\xi_0 \in H^\infty(\SS_{\frac{\pi}3,\frac{2\pi}3})$,
 \item $\overline{\eta_0\left(\theta - \frac{\pi i}3\right)}\eta_0\left(\theta - \frac{2\pi i}3\right) =
 \sqrt{2\pi|R|}\eta\left(\theta - \frac{2\pi i}3\right)$, $\eta_0 \in H^\infty\big(\SS_{-\frac{2\pi}3,-\frac{\pi}3}\big)$,
 \item $\xi_0\big(\theta + \frac{2\pi i}3\big)\Psi_n(\theta, \theta_1,\dots,\theta_{n-1}) \in \dom\big(\shiftone\otimes\1\otimes\cdots\otimes\1\big)$,
 \item $\eta_0\big(\theta - \frac{2\pi i}3\big)\Psi_n(\theta, \theta_1,\dots,\theta_{n-1}) \in \dom\big(\Delta_1^{-\frac{\k}{2\pi}}\otimes\1\otimes\cdots\otimes\1\big)$.
 \end{itemize}
 We consider the function
 \begin{gather*}
 \eta\left(\theta - \frac{2\pi i}3\right)\xi\left(\theta + \frac{2\pi i}3\right)\prod_{j=1}^{n-1}\xi_0\left(\theta_j + \frac{2\pi i}3\right)
 \cdot \Psi_n\left(\theta, \theta_1,\dots,\theta_{n-1}\right),
 \end{gather*}
 which belongs to
 \begin{gather*}
 \dom\big(\shiftone\otimes\1\otimes\cdots\otimes\1\big)\bigcap\dom\big(\Delta_1^{-\frac{\k}{2\pi}}\otimes\1\otimes\cdots\otimes\1\big)\\
 \qquad {}\bigcap_{j=1}^{n-1} \dom\big(\1\otimes\cdots\otimes\1\otimes\underset{j+1\textrm{-th}}{\Delta_1^{\frac\k{2\pi}}}\otimes\1\otimes\cdots\otimes\1\big)
 \end{gather*}
 by $S$-symmetry of $\Psi$. Furthermore, again from the $S$-symmetry of~$\Psi_n$ and from the proper\-ty~(S\ref{ax:atzero}), it has zeros at any point where two of the variables coincide: $\theta = \theta_j$. By applying Lem\-ma~\ref{pr:analyticity-n}, we obtain that
 \begin{gather*}
 \prod_{1\le j \le n-1} S\left(\theta - \theta_j + \frac{\pi i}3\right) \cdot \eta\left(\theta - \frac{2\pi i}3\right)\xi\left(\theta + \frac{2\pi i}3\right)\\
 \qquad {}\times\prod_{j=2}^n\xi_0\left(\theta_j + \frac{2\pi i}3\right) \cdot \Psi_n\left(\theta, \theta_1,\dots,\theta_{n-1}\right)
 \end{gather*}
still belongs to the same domain. Now we can remove extra $\xi_0$'s with variables $\theta_j$, which does not af\/fect the domain of $\shiftone\otimes\1\otimes\cdots\otimes\1$ and get that
 \begin{gather*}
 \prod_{1\le j \le n-1} S\left(\theta - \theta_j + \frac{\pi i}3\right) \cdot \eta\left(\theta - \frac{2\pi i}3\right)\xi\left(\theta + \frac{2\pi i}3\right) \Psi_n\left(\theta, \theta_1,\dots,\theta_{n-1}\right)
 \end{gather*}
belongs to $\dom\big(\shiftone\otimes\1\otimes\cdots\otimes\1\big)$. Therefore, we can apply the Cauchy theorem (on Hilbert-space valued functions) and
 \begin{gather*}
 ([\chi(\xi),z'(\eta)]\Psi_n)(\theta_1,\dots,\theta_{n-1})
 = -\sqrt{2\pi|R|n}\int d \theta\, \eta\left(\theta - \frac{2\pi i}3\right)\prod_{1\le j \le n-1} S\left(\theta - \theta_j + \frac{\pi i}3\right) \\
\hphantom{([\chi(\xi),z'(\eta)]\Psi_n)(\theta_1,\dots,\theta_{n-1}) =-\sqrt{2\pi|R|n}\int}{} \times\xi\left(\theta + \frac{2\pi i}3\right) \Psi_n\left(\theta, \theta_1,\dots,\theta_{n-1}\right).
 \end{gather*}

{\it The commutator $[z(\xi), \chi'(\eta)]$.} This can be computed in a similar way as the previous one:
 \begin{gather*}
 [(z(\xi), \chi'(\eta)]\Psi_n) (\theta_1,\dots,\theta_{n-1})
 = \sqrt{2\pi|R|n}\int d \theta\, \xi\left(\theta + \frac{2\pi i}3\right)\prod_{1\le j \le n-1} S\left(\theta-\theta_j + \frac{\pi i}3\right) \\
 \hphantom{[(z(\xi), \chi'(\eta)]\Psi_n) (\theta_1,\dots,\theta_{n-1}) =\sqrt{2\pi|R|n}\int}{}
 \times\eta\left(\theta - \frac{2\pi i}3\right) \Psi_n\left(\theta,\theta_1,\dots,\theta_{n-1}\right),
\end{gather*}
which coincides with the result of the commutator $[\chi(f), z'(J_1g^-)]$ up to a sign, therefore they cancel each other.

{\it The commutators $[z^\dagger(\xi), \chi'(\eta)]$ and $[\chi(\xi), z'^\dagger(\eta)]$.} One can show that these two commutators cancel each other by taking the adjoints and repeat the computations as in the commutators before.

{\it The commutator $[\phi(\xi), \phi'(\eta)]$.} This term has been essentially computed in~\cite{Lechner03}. The only dif\/ference is that $\xi$ does not come from the Fourier transform of a test function~$f$ supported in~$W_\L$, but the properties needed in the computations are that $\xi \in \hardytwouf$ (which allows one to apply the Cauchy theorem) and the symmetry property $\overline{\xi(\theta)} = \xi(\theta + i\pi)$ (which assures that the terms appearing in the commutator is the boundary values of the same analytic function), and the corresponding properties of~$\eta$. Here is no ef\/fect of the domain and the result is essentially taken from \cite[before Proposition~2]{Lechner03}:
 \begin{gather*}
 ([\phi(\xi), \phi'(\eta)]\Psi_n)(\theta_1,\dots,\theta_n) \\
 \qquad{}
 =-2\pi i \Bigg(\sum_{k = 1}^n R\eta\left(\theta_k - \frac{\pi i}3\right)\xi\left(\theta_k + \frac{2\pi i}3\right)
 \prod_{j\neq k} S\left(\theta_k-\theta_j + \frac{2\pi i}3\right) \\
\qquad\qquad\qquad{}- \sum_{k = 1}^n R\eta\left(\theta_k - \frac{2\pi i}3\right)\xi\left(\theta_k + \frac{\pi i}3\right)
 \prod_{j\neq k} S\left(\theta_k-\theta_j + \frac{\pi i}3\right)\Bigg) \Psi_n(\theta_1,\dots,\theta_n).
 \end{gather*}

{\it The commutator $[\chi(\xi), \chi'(\eta)]$.} We have the expressions~\eqref{eq:chi-n} and those for $\chi_n'(\eta)$. This allows one to expand the commutator into the sum of $n^2$ terms as in~\cite{CT15-1} by a straightforward computation. Of them, there are~$2n(n-1)$ terms which come from the actions of $\chi(\xi)$, $\chi'(\eta)$ on dif\/ferent variables in the sense of~\eqref{eq:chi-n}, and cancel each other. For that, we only need that~$\chi_n(\xi)$ is the projection onto $P_n\H^{\otimes n}$ of the operator~$\chi_1(\xi)\otimes\1\otimes\cdots\otimes\1$ acting on the f\/irst component (and the corresponding property of~$\chi'_n(\eta)$ and the $S$-symmetry of the vectors). The extended domains do not af\/fect this.

The remaining terms are those where $\chi_1(\xi)$ and $\chi'_1(\eta)$ act on the same variable in the sense of~\eqref{eq:chi-n}:
\begin{gather*}
 \<D_n(\rho'_{n-k+1})(\1\otimes\cdots\otimes\1\otimes\chi'_1(\eta))\Phi, \;D_n(\rho_k)(\chi_1(\xi)\otimes\1\otimes\cdots\otimes\1)\Psi\> \\
 = 2\pi|R|\int d\theta_1\cdots d\theta_n\, \prod_{l=1}^{k-1} S\left(\theta_k - \theta_l + \frac{\pi i}3\right) \xi\left(\theta_k + \frac{\pi i}3\right)\Psi_n\left(\theta_1,\dots,\theta_k-\frac{\pi i}3,\dots,\theta_n\right) \\
\hphantom{= 2\pi|R|\int}{} \times \overline{\prod_{l=k+1}^{n} S\left(\theta_l - \theta_k + \frac{\pi i}3\right) \eta\left(\theta_k - \frac{\pi i}3\right)\Phi_n\left(\theta_1,\dots,\theta_k+\frac{\pi i}3,\dots,\theta_n\right)} \\
 = 2\pi|R|\int d\theta_1\cdots d\theta_n\, \prod_{l\neq k} S\left(\theta_k - \theta_l + \frac{2\pi i}3 - i\k\right) \xi\left(\theta_k + \frac{2\pi i}3 - i\k\right)\overline{\eta\left(\theta_k - \frac{2\pi i}3 + i\k\right)} \\
\hphantom{= 2\pi|R|\int}{}\times \Psi_n\left(\theta_1,\dots,\theta_k-i\k,\dots,\theta_n\right)\overline{\Phi_n\left(\theta_1,\dots,\theta_k+i\k,\dots,\theta_n\right)},
 \end{gather*}
 where $\rho_k$, $\rho_{n-k+1}'$ are the cyclic permutations
 \begin{gather*}
 \rho_k\colon \ (1,2,\dots, n) \mapsto (k,1,2, \dots, k-1, k+1,\dots, n), \\
 \rho_{n-k+1}'\colon \ (1,2,\dots, n) \mapsto (1,2, \dots, k-1, k+1,\dots, n, k),
 \end{gather*}
and in the second equality we used that $S$ has no pole except $\frac{\pi i}3$, $\frac{2\pi i}3$ and applied the (Hilbert space-valued) Cauchy theorem: the pole at $\frac{\pi i}3$ makes no problem by the fact that $\Psi_n \in \dom(\chi_n(\xi))$, $\Phi_n \in \dom(\chi_n'(\eta))$.

Now, the question is the pole of $S$ at $\frac{2\pi i}3$. Let us forget at the moment the $S$-factors and look at the further shifted integrand:
\begin{gather*}
\xi\left(\theta_k + \frac{2\pi i}3\right)\Psi_n\left(\theta_1,\dots,\theta_n\right)
 \overline{\eta\left(\theta_k - \frac{2\pi i}3\right)\Phi_n\left(\theta_1,\dots,\theta_n\right)} \\
=\frac1{\sqrt{2\pi|R|}}\xi\left(\theta_k + \frac{2\pi i}3\right)\eta_0\left(\theta_k - \frac{\pi i}3\right)\Psi_n\left(\theta_1,\dots,\theta_n\right) \cdot \overline{\eta_0\left(\theta_k - \frac{2\pi i}3\right)\Phi_n\left(\theta_1,\dots,\theta_n\right)},
\end{gather*}
where we used that $\overline{\eta_0\left(\theta_k - \frac{\pi i}3\right)}\eta_0\left(\theta_k - \frac{2\pi i}3\right) = \sqrt{2\pi|R|}\eta\left(\theta_k - \frac{2\pi i}3\right)$.

{\samepage Let us observe that
\begin{itemize}\itemsep=0pt
\item $\eta_0\big(\theta_k - \frac{2\pi i}3\big)\Phi_n(\theta_1,\dots,\theta_n)
 \in \dom\big(\1\otimes\cdots \otimes\1\otimes\underset{k\textrm{-th}}{\Delta_1^{-\frac{\k}{2\pi}}}\otimes\1\otimes\cdots\otimes\1\big)$,
\item $\xi\big(\theta_k + \frac{2\pi i}3\big)\eta_0\big(\theta_k - \frac{\pi i}3\big)\Psi_n(\theta_1,\dots,\theta_n)
 \in \dom\big(\1\otimes\cdots \otimes\1\otimes\underset{k\textrm{-th}}{\big(\Delta_1^{\frac{\k}{2\pi}} + \Delta_1^{-\frac{\k}{2\pi}}\big)}\otimes\1\otimes\cdots\otimes\1\big)$,
 since $\xi \in \hardyinfuf$ and $\eta_0(\theta_k - \frac{\pi i}3) = \sqrt{2\pi|R|}\eta\big(\theta_k - \frac{\pi i}3\big)\eta_0\big(\theta_k - \frac{2\pi i}3\big)$, $\eta \in H^\infty(\SS_{-\pi, 0})$.
 \end{itemize}
 We} use the same trick as before: consider the function
 \begin{gather*}
 \xi\left(\theta_k + \frac{2\pi i}3\right)\eta_0\left(\theta_k - \frac{\pi i}3\right)\prod_{l\neq k}\xi_0\left(\theta_l + \frac{2\pi i}3\right)
 \cdot \Psi_n\left(\theta_1,\dots,\theta_n\right),
 \end{gather*}
 which belongs to
 \begin{gather*}
 \dom\big(\1\otimes\cdots \otimes\1\otimes\underset{k\textrm{-th}}{\big(\Delta_1^{\frac{\k}{2\pi}} + \Delta_1^{-\frac{\k}{2\pi}}\big)}\otimes\1\otimes\cdots\otimes\1\big)\\
 \qquad{} \bigcap_{l\neq k} \dom\big(\1\otimes\cdots\otimes\1\otimes\underset{l\textrm{-th}}{\Delta_1^{\frac\k{2\pi}}}\otimes\1\otimes\cdots\otimes\1\big)
 \end{gather*}
 by $S$-symmetry of $\Psi$. Furthermore, again by $S$-symmetry, for $k< l$ we have
 \begin{gather*}
 \Psi(\theta_1,\dots,\theta_n) = \prod_{k<j<l}S(\theta_l-\theta_j)S(\theta_j-\theta_k)\cdot S(\theta_l-\theta_k)\Psi(\theta_1,\dots,\theta_l,\dots,\theta_k,\dots,\theta_n).
 \end{gather*}
 From this expression, it is clear that $\Psi$ has a zero at $\theta_l - \theta_k = 0$, $k < l$. A parallel argument holds for $l < k$.

By applying Lemma \ref{pr:analyticity-n}, we obtain that
 \begin{gather*}
 \prod_{l\neq k} S\left(\theta_k - \theta_l + \frac{2\pi i}3\right) \cdot \xi\left(\theta_k + \frac{2\pi i}3\right)\eta_0\left(\theta_k - \frac{\pi i}3\right)\prod_{l\neq k}\xi_0\left(\theta_l + \frac{2\pi i}3\right) \cdot \Psi_n\left(\theta_1,\dots,\theta_n\right)
 \end{gather*}
still belongs to the same domain. We remove extra $\xi_0$'s with variables $\theta_l$ and get that
 \begin{gather*}
 \prod_{l \neq k} S\left(\theta_k - \theta_l + \frac{2\pi i}3\right) \cdot \xi\left(\theta_k + \frac{2\pi i}3\right)\eta_0\left(\theta_k - \frac{\pi i}3\right) \Psi_n\left(\theta_1,\dots,\theta_n\right)
 \end{gather*}
belongs to $\dom\big(\1\otimes\cdots \otimes\1\otimes\underset{k\textrm{-th}}{\Delta_1^{\frac{\k}{2\pi}}}\otimes\1\otimes\cdots\otimes\1\big)$.

Therefore, we can apply the Cauchy theorem and
 \begin{gather*}
 \<D_n(\rho'_{n-k+1})(\1\otimes\cdots \otimes\1\otimes\chi'_1(\eta))\Phi, D_n(\rho_k)(\chi_1(\xi)\otimes\1\otimes\cdots\otimes\1)\Psi\> \\
\qquad{} = 2\pi |R|\int d\theta_1\cdots d\theta_n \prod_{l\neq k} S\left(\theta_k - \theta_l + \frac{2\pi i}3\right) \xi\left(\theta_k + \frac{2\pi i}3\right)\Psi_n\left(\theta_1,\dots,\theta_n\right) \\
 \qquad\hphantom{= 2\pi |R|\int}{}\times \overline{\eta\left(\theta_k - \frac{2\pi i}3\right)\Phi_n\left(\theta_1,\dots,\theta_n\right)} \\
\qquad{} = 2\pi |R|\int d\theta_1\cdots d\theta_n \prod_{l\neq k} S\left(\theta_k - \theta_l + \frac{2\pi i}3\right) \xi\left(\theta_k + \frac{2\pi i}3\right)\eta\left(\theta_k - \frac{\pi i}3\right) \\
\qquad\hphantom{= 2\pi |R|\int}{} \times \overline{\Phi_n\left(\theta_1,\dots,\theta_n\right)}\Psi_n\left(\theta_1,\dots,\theta_n\right).
 \end{gather*}
 By recalling (S\ref{ax:atzero}) which implies $R = i|R|$, this cancels the f\/irst contribution from the commutator $\<\Phi, [\phi(f),\phi'(g)]\Psi\>$.

 The remaining terms in $\<\chi(\xi) \Phi, \chi'(\eta)\Psi\>$ are, by an analogous argument as above,
 \begin{gather*}
 \<D_n(\rho_k)(\chi_1(\xi)\otimes\1\otimes\cdots \otimes\1)\Phi, D_n(\rho'_{n-k+1})(\1\otimes\cdots\otimes\1\otimes\chi'_1(\eta))\Psi\> \\
 \qquad{}= -2\pi iR\int d\theta_1\cdots d\theta_n \prod_{l\neq k} S\left(\theta_k - \theta_l + \frac{\pi i}3\right) \eta\left(\theta_k - \frac{2\pi i}3\right)\xi\left(\theta_k + \frac{\pi i}3\right) \\
 \qquad\hphantom{= -2\pi iR\int}{}\times \overline{\Phi_n\left(\theta_1,\dots,\theta_n\right)}\Psi_n\left(\theta_1,\dots,\theta_n\right)
 \end{gather*}
which is equal to the second contribution from the commutator $\<\Phi, [\phi(f),\phi'(g)]\Psi\>$ up to the sign, therefore, they cancel each other.
\end{proof}

\subsection{Estimates}\label{strong}

Let us assume that $\chi_n(\xi) + \chi'_n(\eta)$ is already essentially self-adjoint on its natural domain $\Dom_n(\xi,\eta) := \dom(\chi_n(\xi))\cap\dom(\chi'_n(\eta))$. We denote the algebraic direct sum by $\Dom(\xi,\eta) := \underline{\bigoplus}_n \Dom_n(\xi,\eta)$. We are going to show that Proposition~\ref{pr:general} can be applied to show that~$\fct(\xi)$ and~$\fct'(\eta)$ strongly commute.

\begin{Proposition}\label{pr:positive}
There is $c(\xi,\eta) > 0$ such that, for $\Psi \in \Dom_n(\xi,\eta)$, it holds that
\begin{gather*}
 \re\<\chi(\xi)\Psi, \chi'(\eta)\Psi\> \ge -nc(\xi,\eta)\<(\chi(\xi) + \chi'(\eta)+2)\Psi, \Psi\>,
 \end{gather*}
and therefore,
 \begin{gather*}
 \|\chi(\xi)\Psi\| \le \|(\chi(\xi) + \chi'(\eta) +c(\xi,\eta)(N+\1)\Psi\|, \\
 \|\chi'(\eta)\Psi\| \le \|(\chi(\xi) + \chi'(\eta) +c(\xi,\eta)(N+\1)\Psi\|
 \end{gather*}
 for $\Psi \in \Dom(\xi,\eta)$.
\end{Proposition}
\begin{proof}
We may assume that $\Psi \in \Dom_n(\xi,\eta)$. We use the following expressions:
 \begin{gather*}
 \chi_n(\xi) = \sum_{1\le k \le n} D_n(\rho_k)(\chi_1(\xi)\otimes\1\otimes\cdots\otimes\1)P_n, \\
 \chi'_n(\eta) = \sum_{1\le k \le n} D_n(\rho'_k)(\1\otimes\cdots\otimes\1\otimes\chi'_1(\eta))P_n.
 \end{gather*}
 By proceeding as in Section \ref{weak}, we see that $\<\chi(\xi)\Psi, \chi'(\eta)\Psi\>$ consists of $n(n-1)$ times
 \begin{gather*}
 \<(\chi_1(\xi)\otimes\1\otimes\cdots\otimes\1)\Psi, (\1\otimes\cdots\otimes\1\otimes\chi'_1(\eta))\Psi\>
 \end{gather*}
 and the following $n$ terms ($1\le k \le n$):
 \begin{gather}\label{eq:cross}
 \<D_n(\rho_k)(\chi_1(\xi)\otimes\1\otimes\cdots\otimes\1)\Psi, D_n(\rho'_{n-k+1})(\1\otimes\cdots\otimes\1\otimes\chi'_1(\eta))\Psi\>.
 \end{gather}
 We write $D_n(\rho_k) = M_{S_k}F_k$, where
 \begin{gather*}
 S_k\left(\pmb{\theta}\right) := \prod_{j < k} S\left(\theta_k-\theta_j\right)
 \end{gather*}
 and $F_k$ is the corresponding cyclic permutation of the variables. We have
 \begin{gather*}
D_n(\rho_k)(\chi_1(\xi)\otimes\1\otimes\cdots\otimes\1)\Psi \\
\qquad{} = M_{S_k}F_k\big(\chi_1(\xi)^\frac12\otimes\1\otimes\cdots\otimes\1\big)F_k^*F_k(\chi_1(\xi)^\frac12\otimes\1\otimes\cdots\otimes\1)\Psi \\
\qquad{}= \big(\1\otimes\cdots\1\otimes\underset{k\textrm{-th}}{\chi_1(\xi)^\frac12}\otimes\1\otimes\cdots\otimes\1\big)M_{S_k(\pmb{\theta} + \frac{\pi i}6)}F_k\big(\chi_1(\xi)^\frac12\otimes\1\otimes\cdots\otimes\1\big)\Psi,
 \end{gather*}
 where
 \begin{gather*}
 S_k\left(\pmb{\theta} + \frac{\pi i}6\right) = \prod_{j < k} S\left(\theta_k-\theta_j + \frac{\pi i}6\right),
 \end{gather*}
 which is bounded by $1$ by (S\ref{ax:bound}). Similarly,
 \begin{gather*}
 D_n(\rho'_{n-k+1})(\1\otimes\cdots\otimes\1\otimes\chi'_1(\eta))\Psi \\
 \qquad{} = M_{S'_{n-k+1}}F'_{n-k+1}\big(\1\otimes\cdots\otimes\1\otimes\chi'_1(\eta)^\frac12\big)F_{n-k+1}^{\prime *}F'_{n-k+1}\big(\1\otimes\cdots\otimes\1\otimes\chi'_1(\eta)^\frac12\big)\Psi \\
\qquad{} = \big(\1\otimes\cdots\1\otimes\underset{k\textrm{-th}}{\chi'_1(\eta)^\frac12}\otimes\1\cdots\otimes\1\big)M_{S'_{n-k+1}(\pmb{\theta} + \frac{\pi i}6)}F_{n-k+1}\big(\1\otimes\cdots\otimes\1\otimes\chi'_1(\eta)^\frac12\big)\Psi,
 \end{gather*}
 where $S'_{n-k+1}(\pmb{\theta} + \frac{\pi i}6) = \prod\limits_{j > k} S(\theta_k-\theta_j + \frac{\pi i}6)$, which is again bounded by $1$ and $F'_{n-k+1}$ is the corresponding cyclic permutation. We def\/ine $c(\xi,\eta) = \|\chi_1(\xi)^\frac12\cdot \chi'_1(\eta)^\frac12\| + 1$, which is f\/inite. Now we have
 \begin{gather*}
 |\eqref{eq:cross}|
 \le c(\xi,\eta)\big\|(\chi_1(\xi)^\frac12\otimes\1\otimes\cdots\otimes\1)\Psi\big\|^\frac12\cdot
 \big\|(\1\otimes\cdots\otimes\1\otimes\chi'_1(\eta)^\frac12)\Psi\big\|^\frac12 \\
\hphantom{|\eqref{eq:cross}|}{} \le c(\xi,\eta)\big(\big\|(\chi_1(\xi)^\frac12\otimes\1\otimes\cdots\otimes\1)\Psi\big\|+
 \big\|(\1\otimes\cdots\otimes\1\otimes\chi'_1(\eta)^\frac12)\Psi\big\|\big) \\
\hphantom{|\eqref{eq:cross}|}{} \le c(\xi,\eta)\big(\big\|(\chi_1(\xi)^\frac12\otimes\1\otimes\cdots\otimes\1)\Psi\big\|^2+
 \big\|(\1\otimes\cdots\otimes\1\otimes\chi'_1(\eta)^\frac12)\Psi\big\|^2 + 2\big) \\
\hphantom{|\eqref{eq:cross}|}{}= c(\xi,\eta)\big( \<\chi_n(\xi)\Psi, \Psi \>+ \<\chi'_n(\eta)\Psi, \Psi \> + 2\big),
 \end{gather*}
 from which the f\/irst statement follows directly.

The second statement is easy because
 \begin{gather*}
\|(\chi(\xi) + \chi'(\eta) + (n+1)c(\xi,\eta)\Psi\|^2\\
\qquad{} = \|\chi(\xi)\Psi\|^2 + 2\re\<\chi(\xi)\Psi, \chi'(\eta)\Psi\> + \|\chi'(\eta)\Psi\|^2 \\
\qquad\quad {}+ 2(n+1)c(\xi,\eta)\<(\chi(\xi) + \chi'(\eta)\Psi,\Psi\> + (n+1)^2c(\xi,\eta)^2\|\Psi\|^2,
 \end{gather*}
 which extends to $\Psi = (\Psi_n) \in \Dom(\xi,\eta)$ by replacing $n$ by the number operator $N$, because the operators appearing here preserve $\H_n$. Therefore, by dropping the positive terms at each $n$ in the right-hand side, we conclude:
 \begin{gather*}
 \|\chi(\xi)\Psi\|^2 \le \|(\chi(\xi) + \chi'(\eta) + c(\xi,\eta)(N+\1)\Psi\|^2, \\
 \|\chi'(\eta)\Psi\|^2 \le \|(\chi(\xi) + \chi'(\eta) + c(\xi,\eta)(N+\1)\Psi\|^2.\tag*{\qed}
 \end{gather*}
\renewcommand{\qed}{}
\end{proof}

It is important that the estimate does not grow too fast with $n$, thanks to (S\ref{ax:bound}).

\begin{Corollary}\label{co:strong}
Assume that $\chi_n(\xi) + \chi'_n(\eta)$ is essentially self-adjoint. Then $\fct(\xi)$ and $\fct'(\eta)$ commute strongly.
\end{Corollary}
\begin{proof}
From the assumption, it is immediate that $\chi(\xi) + \chi'(\eta) + c(N+\1)$ is essentially self-adjoint, where $N$ is the number operator and $c = \|\xi\| + \|\eta\| + c(\xi,\eta)$, where $c(\xi,\eta)$ comes from Proposition~\ref{pr:positive}. For simplicity we do not indicate the dependence of $c$ on $\xi$ and $\eta$. As~$\chi(\xi)$ and~$\chi'(\eta)$ are def\/ined at each $n$-particle space, they strongly commute with $N$.

Next, let us observe that $\fct(\xi) + \fct'(\eta) + c(N+\1)$ is essentially self-adjoint. Indeed, this is equal to the sum of $\chi(\xi) + \chi'(\eta) + c(N+\1)$ and $\phi(\xi) + \phi'(\eta)$. As $\chi(\xi) + \chi'(\eta)$ and $N+\1$ are both positive and strongly commute, we have for any $\Psi \in \Dom(\xi,\eta)$ that $\|c(N+\1)\Psi\| \le \|(\chi(\xi) + \chi'(\eta) + c(N+\1))\Psi\|$. On the other hand, as $\phi(\xi) + \phi'(\eta)$ are the sum of creation and annihilation operators, it holds by the above choice of $c$ that $\|(\phi(\xi) + \phi'(\eta))\Psi\| \le \|(\|\xi\| + \|\eta\|)(N+\1)\Psi\| < \|c(N+\1)\Psi\|$. Now by Kato--Rellich perturbation theorem \cite[Theorem~X.12]{RSII}\footnote{Here the larger operator is only essentially self-adjoint, but the estimate extends to the whole domain, hence the theorem applies. Accordingly, the conclusion of the theorem is essential self-adjointness.}, $\fct(\xi) + \fct'(\eta) + c(N+\1)$ is essentially self-adjoint on $\Dom(\xi,\eta)$.

In order to show the hypothesis (\ref{gn1}) of Proposition~\ref{pr:general} with $A = \fct(\xi)$, $B = \fct'(\eta)$, $T_0 = c(N+\1)$, note that, by Proposition~\ref{pr:positive} and the def\/inition of $c$, we have
 \begin{gather*}
 \big\|\fct(\xi)\Psi\big\| = \|(\phi(\xi) + \chi(\xi))\Psi\| \le \|c(N+\1)\Psi\| + \|(\chi(\xi) + \chi'(\eta)) + c(\xi,\eta)(N+\1)\Psi\| \\
\hphantom{\big\|\fct(\xi)\Psi\big\| = \|(\phi(\xi) + \chi(\xi))\Psi\|}{} \le 2\|(\chi(\xi) + \chi'(\eta) + c(N+\1))\Psi\|
 \end{gather*}
 and a parallel estimate shows that $\big\|\fct'(\eta)\Psi\big\| \le 2\|(\chi(\xi) + \chi'(\eta) + c(N+\1))\Psi\|$.

As for (\ref{gn2}) and (\ref{gn3}), recall that $N$ commutes with $\chi(\xi)$ and $\chi'(\eta)$, therefore, the (weak) commutator between $(N+\1)$ and $\fct(\xi)$ is reduced to $[N,\phi(\xi)]$, which is the sum of $[N, z^\dagger(\xi)] = z^\dagger(\xi)$ and $[N, z(\xi)] = -z(\xi)$. Namely, we have
 \begin{gather*}
 \big|\<N\Psi, \fct(\xi)\varphi\> - \<\fct(\xi)\Psi, N\Phi\>\big| = |\<\Psi, (z^\dagger(\xi) - z(\xi))\Phi\>|
 \end{gather*}
and it is easy to see that this can be bounded by both $\big\|(c(N+\1))^\frac12\Psi\big\|\cdot \big\|(c(N+\1))^\frac12\Phi\big\|$ and $\|\Psi\|\cdot \|c(N+\1)\Phi\|$.

As we saw that $\fct(\xi)$ and $\fct'(\eta)$ weakly commute on $\Dom(\xi,\eta)$ in Section~\ref{weak}, we can apply Proposition~\ref{pr:general} and they strongly commute.
\end{proof}

\subsection{Construction of Borchers triple}\label{nets}
We now proceed to von Neumann algebras, always assuming that $\chi_n(\xi)+\chi'_n(\eta)$ is essentially self-adjoint on $\Dom_n(\xi,\eta)$. Among the properties of Borchers triples, important are the separating property of the vacuum $\Omega$ and the endomorphic action of the Poincar\'e group. The former is an immediate consequence of Corollary~\ref{co:strong} and the latter follows from the canonical correspondence from~$\xi$ to~$\chi(\xi)$.

{\sloppy For the next lemma, we need to recall the Beurling factorization \cite[Theorems~17.15 and~17.17]{Rudin87}. Any bounded analytic function $f(\zeta)$ on the unit circle can be decomposed in to the product $f(\zeta)=cf_{\bl}(\zeta)f_{\tout}(\zeta)f_{\tin}(\zeta)$, where the Blaschke product $f_{\bl}$ is a possibly inf\/inite product of functions $\frac{|\a|}{\a}\frac{\a-\zeta}{1-\bar\a\zeta}$, where $|\a| < 1$, $f_{\tout}$ is the Poisson integral of the boundary value of~$|f(\zeta)|$ on~$S^1$ and~$f_{\tin}$ is the Poisson integral of a certain singular measure. A corresponding decomposition holds for any region in~$\CC$ which is conformally equivalent to the unit disk, in particular, the strip $\SS_{-\frac{2\pi}3,-\frac\pi 3}$. In~\cite{Tanimoto15-1}, we exploited this decomposition in order to construct the analytic function~$\eta_0$ (see Section~\ref{one}).

}

From here on, we switch to the objects with $'$, because they are by convention those which generate the von Neumann algebra.

\begin{Lemma}\label{lm:covariance}
 $\chi'(\eta)$ and $\fct'(\eta)$ are covariant in the following sense: if $a \in W_\R$, then we have $\chi'(U_1(a,\l)\eta) = \Ad U(a,\l)(\chi'(\eta))$ and $\fct'(U_1(a,\l)\eta) = \Ad U(a,\l)(\fct'(\eta))$ as operators $($inclu\-ding the domains$)$.
\end{Lemma}
\begin{proof}
First note that, if $\eta \in \hardytwolf \cap \hardyinflf$, then it follows from $(U_1(a,\l)\eta)(\theta) = e^{ip(\theta)\cdot a}\eta(\theta - \l)$ that $U_1(a,\l)\eta \in \hardytwolf \cap \hardyinflf$, as $a \in W_\R$ implies that\footnote{Our convention of the Minkowski metric is $a\cdot b = a_0b_0 - a_1b_1$.} $|e^{ip(\zeta)\cdot a}| \le 1$ for $\zeta \in \striplf$. Furthermore, if $\eta = \underline \eta^2$, then $(U_1(a,\l)\eta)(\theta) = \left(e^{ip(\theta)\cdot a/2}\underline\eta(\theta-\l)\right)^2$, hence $\chi'(U_1(a,\l)\eta)$ can be def\/ined.

Next, let us remember how $\chi_1'(\eta) = M_{\eta_0(\cdot \, - \frac{2\pi i}3)}^{\prime *}\shiftione M_{\eta_0(\cdot \, - \frac{2\pi i}3)}'$ is constructed: $M_{\eta_0(\cdot \, - \frac{2\pi i}3)} = M_{\eta_\bl}M_{\eta_\tin}M_{\eta_-}$ (again up to the scaling), where $M_{\eta_\bl}$ is the Blaschke product corresponding to the zeros of $\underline\eta$ in $\SS_{-\frac{2\pi}3,-\frac\pi 3}$, $M_{\eta_\tin}$ is the inner singular part of $\underline\eta$ again in $H^\infty\big(\SS_{-\frac{2\pi}3,-\frac{\pi}3}\big)$ and $\eta_-$ is an outer function on $\SS_{-\frac{2\pi}3,-\frac{\pi}3}$ such that $ \eta_-\big(\theta - \frac{\pi i}3\big)\overline{\eta_-\big(\theta - \frac{2\pi i}3\big)} = \eta_\tout\big(\theta - \frac{\pi i}3\big)$, or equivalently,
 \begin{gather}\label{eq:outer}
 \eta_-\left(\theta - \frac{2\pi i}3\right)\overline{\eta_-\left(\theta - \frac{\pi i}3\right)} = \eta_\tout\left(\theta - \frac{2\pi i}3\right),
 \end{gather}
 where $\eta_\tout$ is the outer part of $\underline\eta$.

First we consider a pure boost $U_1(0,\l)$. As the decomposition $M_{\eta_0(\cdot \,- \frac{2\pi i}3)} = M_{\eta_\bl}M_{\eta_\tin}M_{\eta_-}$ is canonical where ${\eta_\bl}$ is a~Blaschke product, $\eta_\tin$ is inner and~$\eta_-$ is outer with the above equality, we also have the corresponding decomposition for $U_1(0,\l)\eta$ and we obtain $\chi'_1(U_1(0,\l)\eta) = U_1(0,\l)\chi'(\eta)U_1(0,\l)^*$.

Next (and more importantly), take a pure translation $U_1(a,0)$. We have $(U_1(a,0)\eta)(\theta) = e^{ia\cdot p(\theta)}\eta(\theta)$, hence $(U_1(a,0)\eta)\big(\theta-\frac{2\pi i}3\big) = e^{ia\cdot p(\theta-\frac{2\pi i}3)}\eta\big(\theta-\frac{2\pi i}3\big)$. The function $e^{ia\cdot p(\theta)}$ is purely outer\footnote{$e^{ia\cdot p(\theta)}$ is singular as a function on $\SS_{-\pi,0}$, but is outer when restricted to $\SS_{-\frac{2\pi}3,-\frac{\pi}3}$, as it gets scaled.} on $\SS_{-\frac{2\pi}3,-\frac{\pi}3}$. It holds that $-p(\theta)+p\big(\theta-\frac{\pi i}3\big) = p\big(\theta-\frac{2\pi i}3\big)$, hence
 \begin{gather*}
 e^{-ia\cdot p(\theta)} \cdot \overline{e^{-ia\cdot p(\theta+\frac{\pi i}3)}} = e^{ia\cdot p(\theta-\frac{2\pi i}3)},
 \end{gather*}
which, by comparison with \eqref{eq:outer}, implies that $(U_1(a,0)\eta)_-$ is the product $\eta_-$ and $e^{-ia\cdot p(\theta)}$, i.e.,
 \begin{gather*}
 (U_1(a,0)\eta)_-\left(\theta - \frac{\pi i}3\right) = e^{-ia\cdot p(\theta)}\eta_-\left(\theta - \frac{\pi i}3\right)
 \end{gather*}
and $(U(a,0)\eta)_\bl = \eta_\bl$, $(U(a,0)\eta)_\tin = \eta_\tin$. In other words,
\begin{gather*}
 M_{(U_1(a,0)\eta)_-(\cdot \,-\frac{\pi i}3)} = U_1(-a,0) M_{\eta_-(\cdot \,-\frac{\pi i}3)} = U_1(a,0)^* M_{\eta_-(\cdot \,-\frac{\pi i}3)}.
\end{gather*}
Now we have the decomposition
\begin{gather*}
 M_{(U_1(a,0)\eta)_0(\cdot \,- \frac{2\pi i}3)} = M_{\eta_\bl}M_{\eta_\tin}M_{\eta_-}U_1(a,0)^*,
 \end{gather*}
therefore,
 \begin{gather*}
 \chi'_1(U_1(a,0)\eta) = U_1(a,0)\chi'(\eta)U_1(a,0)^*.
 \end{gather*}

The covariance of the full operators $\chi'(\eta) = \bigoplus_n P_n(\1\otimes\cdots\otimes\1\otimes \chi'_1(\eta))P_n$ and $\fct'(\eta) = \phi'(\eta) + \chi'(\eta)$ follows from the covariance of~$\chi'_1(\eta)$.
\end{proof}

\begin{Theorem}
Assume that $\V$ is a subset of $\hardytwouf\cap\hardyinfuf$, whose elements $\xi$ have the form $\xi = \underline\xi^2$, satisfy $\underline\xi(\theta + i\pi) = \overline{\underline\xi(\theta)}$, and $\V$ is invariant under $U_1(a,\l), (a,\l)\in\poincare$ such that $a\in W_\L$ and is total in $\H_1$. Assume further that $\chi(\xi) + \chi'(\eta)$ is essentially self-adjoint on $\Dom(\xi,\eta)$ for any pair of $\xi \in \V$, $\eta \in J_1\V$.

We define
 \begin{gather*}
 \M = \big\{e^{i\fct'(\eta)}\colon \eta \in J_1\V \big\}''.
 \end{gather*}
Then, $(\M,U,\Omega)$ is a Borchers triple.
\end{Theorem}
\begin{proof}
$\Omega$ is cyclic for $\M$ by the argument of \cite[Section~3, Reeh--Schlieder property]{CT15-1} and the density of the linear span of~$\V$ in~$\H_1$. It is separating because~$\Omega$ is cyclic for $\big\{e^{i\fct(\xi)}\colon \xi \in \V \big\}''$, which commutes with~$\M$ by Corollary~\ref{co:strong}. Finally, for $a \in W_\R$, $\Ad U(a,\l)\M \subset \M$ follows immediately from Lemma~\ref{lm:covariance} and the assumption that $\V$ is invariant under $U_1(a,\l)$ for $a \in W_\L$, hence~$J_1\V$ is invariant under $U_1(a,\l)$, $a \in W_\R$.
\end{proof}

Note that, if the assumption of strong commutativity is dropped, the totality of $\V$ is easy, as (cf.~\cite{CT15-1}). We only have to f\/ind suf\/f\/iciently many~$\xi$'s and~$\eta$'s which satisfy strong commutativity.

\section{Self-adjointness of bound state operators}\label{sa}
Here we consider $\chi(\xi) + \chi'(\eta)$. Although our f\/ields is $\fct(\xi) = \phi(\xi) + \chi(\xi)$ and not $\chi(\xi)$ or $\chi'(\eta)$, the question of the domain of $\fct(\xi)$ can be solved by Corollary~\ref{co:strong}, once that of $\chi(\xi) + \chi'(\eta)$ is settled. We show that $\chi_n(\xi) + \chi'_n(\eta)$ is essentially self-adjoint for $n = 1,2$ if $\xi = \underline\xi^2$ and $\eta = \underline\eta^2$ as in Section~\ref{one}, and give some observations for $n \ge 3$.

\subsubsection*{Observations on closure}
\begin{Lemma}\label{lm:sum-closed}
Let $X$, $Y$ be closed operators and assume that there is $c \in \RR$ such that $c\|\Psi\|^2 \le \re \<X\Psi, Y\Psi\>$ for $\Psi \in \dom(X)\cap \dom(Y)$. Then $X+Y$ is closed.
\end{Lemma}
\begin{proof}
 Let $\Psi_n \in \dom(X)\cap\dom(Y)$ such that $\Psi_n \to \Psi$ and $(X+Y)\Psi_n \to \Phi$. Then
 \begin{gather*} \|(X+Y)(\Psi_m-\Psi_n)\|^2 \\
 \qquad{} = \|X(\Psi_m-\Psi_n)\|^2 + \|Y(\Psi_m-\Psi_n)\|^2
 + 2\re \<X(\Psi_m-\Psi_n), Y(\Psi_m-\Psi_n)\> \\
 \qquad{} \ge \|X(\Psi_m-\Psi_n)\|^2 + \|Y(\Psi_m-\Psi_n)\|^2 + 2c\|\Psi_m-\Psi_n\|^2,
 \end{gather*}
from which it follows that $\{X\Psi_m\}$ and $\{Y\Psi_m\}$ are both convergent, and therefore, $\Psi \in \dom(X)\cap\dom(Y)$, namely, $X+Y$ is closed.
\end{proof}

\begin{Lemma}\label{lm:closed}
Let $X$, $Y$ be closable operators and assume that there is $c \in \RR$ such that $c\|\Psi\|^2 \le \re \<X\Psi, Y\Psi\>$ for $\Psi \in \dom(X)\cap \dom(Y)$. Then $X+Y$ is closable and $\overline{X+Y} \subset \overline X + \overline Y$.
\end{Lemma}
\begin{proof}
For $\Psi \in \dom(X)\cap\dom(Y)$, we have $\|(X+Y)\Psi\|^2 = \|X\Psi\|^2+\|Y\Psi\|^2+2\re\<X\Psi,Y\Psi\>$ where the last term is bounded below by $c\|\Psi\|^2$ by assumption. If $\Psi_n\to0$ and $(X+Y)\Psi_n\to \Phi$, then $\|(X+Y)(\Psi_n-\Psi_m)\| \to 0$ and
 \begin{gather*}
\lim_{m,n}\,\inf \left(2\re\<X(\Psi_n-\Psi_m),Y(\Psi_n-\Psi_m)\>\right) \ge 0.
 \end{gather*}
From this it follows that $\|X(\Psi_n-\Psi_m)\| \to 0$ and $\|Y(\Psi_n-\Psi_m)\| \to 0$, namely $X\Psi_n$ and $Y\Psi_n$ are separately convergent. As $X$ and $Y$ are closable, we obtain $X\Psi_n\to 0 $ and $Y\Psi_n \to 0$ and $\Phi = 0$, namely $X+Y$ is closable.

To show $\overline{X+Y} \subset \overline X + \overline Y$, we only have to repeat the argument with $\Psi_n \to \Psi$ to obtain $\Psi \in \dom(\overline X)\cap\dom(\overline Y) = \dom(\overline X+\overline Y)$.
\end{proof}

\begin{Lemma}\label{lm:positive-closed}
Let $X$, $Y$ be closed operators and assume that there is $c > 0$ such that $0 \le c^2\|\Psi\|^2 + 2c(\<\Psi,X\Psi\> + \<\Psi, Y\Psi\>) + 2\re \<X\Psi, Y\Psi\>$ for $\Psi \in \dom(X)\cap \dom(Y)$.
 Then $X+Y$ is closed.
\end{Lemma}
\begin{proof}
 Let $\Psi_n \in \dom(X)\cap\dom(Y)$ such that $\Psi_n \to \Psi$ and $(X+Y+c)\Psi_n \to \Phi$. Hence we have
 \begin{gather*}
 \|(X+Y+c)(\Psi_m-\Psi_n)\|^2 \\
 \qquad{} = \|X(\Psi_m-\Psi_n)\|^2 + \|Y(\Psi_m-\Psi_n)\|^2 + 2\re \<X(\Psi_m-\Psi_n), Y(\Psi_m-\Psi_n)\> \\
 \qquad \quad{} + 2c(\<(\Psi_m-\Psi_n), X(\Psi_m-\Psi_n)\> + \<(\Psi_m-\Psi_n), Y(\Psi_m-\Psi_n)\>) + c^2\|(\Psi_m-\Psi_n)\|^2 \\
 \qquad{} \ge \|X(\Psi_m-\Psi_n)\|^2 + \|Y(\Psi_m-\Psi_n)\|^2,
 \end{gather*}
 from which it follows that $\{X\Psi_m\}$ and $\{Y\Psi_m\}$ are both convergent, and therefore, $\Psi \in \dom(X)\cap\dom(Y)$, namely, $X+Y+c$ is closed. This is equivalent to that $X+Y$ is closed.
\end{proof}

\subsection{One-particle components}\label{sa-one}

Now we show that $\chi_1(\xi) + \chi'_1(\eta)$ is self-adjoint on the intersection $\dom(\chi_1(\xi)) \cap \dom(\chi'_1(\eta))$.
A key observation is that $\chi_1(\xi)$ can be written in the form $X^*X$ where $X = \halfshiftone M_{\xi_0(\cdot \,+ \frac{2\pi i}3)}$. It is clear that $X$ is densely def\/ined, as $\xi_0 \in H^\infty(\SS_{\frac{\pi}3,\frac{2\pi}3})$, and closed as $M_{\xi_0(\cdot \,+ \frac{2\pi i}3)}$ is bounded. Furthermore, it holds that $X^* = M_{\xi_0(\cdot \,+ \frac{2\pi i}3)}^*\halfshiftone$ without closure, since $M_{\xi_0(\cdot \,+ \frac{2\pi i}3)}^*$ is unitary. Similarly, one can write $\chi'_1(g) = Y^*Y$, where $Y = \halfshiftione M_{\eta_0(\cdot \,- \frac{2\pi i}3)}$ and $Y^* = M_{\eta_0(\cdot \,- \frac{2\pi i}3)}^*\halfshiftione$.

Some ideas in the following lemma are due to Henning Bostelmann.
\begin{Lemma}\label{lm:cleave}
 Let $X$, $Y$ be closed, densely defined operators such that
 \begin{itemize}\itemsep=0pt
 \item $X\dom(X)\subset \dom(Y^*)$ and $Y\dom(Y)\subset \dom(X^*)$,
 \item $Y^*X$ and $X^*Y$ are bounded,
 \item $X+Y$ and $X^* + Y^*$ are densely defined,
 \item there is a decomposition of $\H = \H_{X} \oplus \H_{Y}$ such that
 \begin{itemize}\itemsep=0pt
 \item any $\Psi \in \dom({X})$ is the sum $\Psi = \Psi_{X} \oplus \Psi_{Y}$, $\Psi_{X} \in \H_{X}\cap \dom({X})$ and $\Psi_{Y} \in \H_{Y} \cap \dom({X})$,
 \item any $\Phi \in \dom({Y})$ is the sum $\Phi = \Phi_{X} \oplus \Phi_{Y}$, $\Phi_{X} \in \H_{X}\cap \dom({Y})$ and $\Phi_{Y} \in \H_{Y} \cap \dom({Y})$,
 \item ${X}$ is bounded on $\H_{X}$ and ${Y}$ is bounded on $\H_{Y}$.
 \end{itemize}
 \end{itemize}
 Then the sum $X^*X + Y^*Y$ is self-adjoint on $\dom(X^*X)\cap\dom(Y^*Y)$.
\end{Lemma}
\begin{proof}
By def\/inition, $\dom(X+Y) = \dom(X)\cap\dom(Y)$. Under the assumption, $X+Y$ is closed. Indeed, we have $\<X\Psi, Y\Psi\> = \<\Psi, X^*Y\Psi\>$ which is bounded by $\|X^*Y\|\cdot\|\Psi\|^2$ and we can apply Lemma~\ref{lm:sum-closed}.

Next, note that $\dom({X}+{Y}) = \dom({X})\cap\dom({Y})$ is a core of ${X}$. Indeed, any $\Psi \in \dom({X})$ can be decomposed as $\Psi_{X}\oplus\Psi_{Y}$ by assumption. On the other hand, $\dom({Y})$ is dense and admits the decomposition $\dom({Y}) = (\dom({Y})\cap\H_{X}) \oplus (\dom({Y})\cap\H_{Y})$. As $\dom({X}+{Y})$ is dense by assumption, $\dom({X}+{Y})\cap\H_{X}$ is dense in $\H_{X}$. But ${X}$ is bounded on $\H_{X}$, thus $\dom({X})$ includes $\H_{X}$. Namely, $\dom({X}) = \H_{X} \oplus (\dom({X})\cap\H_{Y})$. Similarly, we have $\dom({Y}) = (\dom({Y})\cap \H_{X}) \oplus \H_{Y}$. Therefore, $\dom({X}+{Y}) = (\dom({Y})\cap \H_{X}) \oplus (\dom({X})\cap\H_{Y})$. We saw that $\dom({X}+{Y})\cap\H_{X} = \dom({Y})\cap \H_{X}$ is dense in $\H_{X}$ and ${X}$ is bounded here, hence, the closure of ${X}$ restricted on $\dom({X}+{Y})$ is $\H_{X} \oplus \dom({X})\cap\H_{Y} = \dom({X})$. Similarly, it holds that $\dom({X}+{Y})$ is a core of ${Y}$.

{\allowdisplaybreaks Next we see that $\dom((X+Y)^*) = \dom(X^*)\cap\dom(Y^*)$. For each vector $\Psi \in \dom({X}+{Y}) = \dom({X}) \cap \dom({Y})$, we take the decomposition $\Psi = \Psi_{X} \oplus \Psi_{Y}$ which is given by the assumption. By the def\/inition of the adjoint operator, a vector $\Phi$ is in $\dom(({X}+{Y})^*)$ if and only if the map $\Psi \mapsto \<\Phi, ({X}+{Y})\Psi\>$ is continuous for $\Psi \in \dom({X}+{Y})$ (namely, there is $c$ such that $|\<\Phi, ({X}+{Y})\Psi\>| \le c\|\Psi\|$). Then we have the following chain of equivalences:
 \begin{gather*}
 \Phi \in \dom(({X}+{Y})^*) \\
 \qquad \Longleftrightarrow \ \Psi \mapsto \<\Phi, ({X}+{Y})\Psi\> \textrm{ is continuous for } \Psi \in \dom({X} + {Y}) \\
 \qquad \Longleftrightarrow \ \Psi_{X} \mapsto \<\Phi, ({X}+{Y})\Psi_{X}\> \textrm{ and }
 \Psi_{Y} \mapsto \<\Phi, ({X}+{Y})\Psi_{Y}\> \\
 \qquad{} \hphantom{\Longleftrightarrow}{} \ \ \textrm{are continuous for }
 \Psi_{X} \in \dom({X}+{Y})\cap \H_{X}, \Psi_{Y} \in \dom({X}+{Y}) \cap \H_{Y}.\\
 \qquad \Longleftrightarrow \ \Psi_{X} \mapsto \<\Phi, {Y}\Psi_{X}\> \textrm{ and }
 \Psi_{Y} \mapsto \<\Phi, {X}\Psi_{Y}\> \\
 \qquad{} \hphantom{\Longleftrightarrow}{} \ \ \textrm{are continuous for }
 \Psi_{X} \in \dom({X}+{Y})\cap \H_{X}, \Psi_{Y} \in \dom({X}+{Y}) \cap \H_{Y}.\\
 \qquad \Longleftrightarrow \ \Psi_{X} \oplus \Psi_{Y} \mapsto \<\Phi, {Y}(\Psi_{X} \oplus\Psi_{Y})\> \textrm{ and }
 \Psi_{X} \oplus \Psi_{Y} \mapsto \<\Phi, {X}(\Psi_{X} \oplus \Psi_{Y})\> \\
 \qquad{} \hphantom{\Longleftrightarrow}{} \ \ \textrm{are continuous for } \Psi = \Psi_{X} \oplus \Psi_{Y} \in \dom({X}+{Y}).
 \end{gather*}
We saw that $\dom({X}+{Y})$ is a core both for ${X}$ and ${Y}$, thus, the continuity of $\Phi \mapsto \<\Phi, {X}\Psi\>$ for $\Psi \in \dom({X}+{Y})$ extends to $\Psi \in \dom({X})$, which implies that $\Phi \in \dom(X^{*})$. Analogously, we obtain $\Phi \in \dom(Y^{*})$. This implies that $({X} + {Y})^* = X^*+Y^*$.}

It is well-known \cite[Theorem X.25]{RSII} that, because $X+Y$ is closed, on the domain $\dom((X+Y)^*(X+Y)) = \{\Psi\in\dom(X+Y): (X+Y)\Psi\in\dom((X+Y)^*))\}$ the operator $(X+Y)^*(X+Y)$ is self-adjoint. From the above observations, this domain is actually
 \begin{gather*}
 \dom((X+Y)^*(X+Y)) = \{\Psi\in\dom(X+Y)\colon (X+Y)\Psi\in\dom((X+Y)^*))\} \\
\hphantom{\dom((X+Y)^*(X+Y))}{} = \{\Psi\in\dom(X+Y)\colon (X+Y)\Psi\in\dom(X^*)\cap\dom(Y^*)\} \\
 \hphantom{\dom((X+Y)^*(X+Y))}{} = \{\Psi\in\dom(X+Y)\colon X\Psi\in\dom(X^*), Y\Psi\in\dom(Y^*)\},
 \end{gather*}
where in the last line we used the assumptions that $X\dom(X)\subset \dom(Y^*)$ and $Y\dom(Y)\subset \dom(X^*)$.

Now we have $(X+Y)^*(X+Y) = X^*X + X^*Y + Y^*X + Y^*Y$ and this is dif\/ferent from $X^*X + Y^*Y$ only by a bounded operator.
\end{proof}

\begin{Proposition}\label{pr:sa-1}
The sum of one-particle components $\chi_1(\xi) + \chi_1'(\eta)$ is self-adjoint on the domain $\dom(\chi_1(\xi))\cap\dom(\chi_1'(\eta))$.
\end{Proposition}
\begin{proof}
This is a straightforward consequence of this Lemma~\ref{lm:cleave}.

Indeed, we only check the existence of the decomposition $\H_1 = \H_{X} \oplus \H_{Y}$. We consider $\log \Delta$ and its spectral decomposition: we take $\H_{X}$ as the spectral subspace of $\log \Delta$ corresponding to $(-\infty,0]$, while $\H_{Y}$ corresponds to $(0,\infty)$. Now the required properties are shown as follows.

Let us denote the spectral projections onto these subspaces by $P_X$ and $P_Y$. Let us take an element $\Psi \in \dom(X) = \dom\big(\halfshift M_{\xi_0(\cdot\, + \frac{2\pi i}3)}\big)$. This means, by def\/inition, $\xi_0(\theta+\frac{2\pi i}3)\Psi(\theta)$ is in $\H^2\left(-\frac{\pi}6, 0\right)$. In addition, $\Psi(\theta)$ alone is~$L^2$. While it holds that $\Psi = P_X\Psi \oplus P_Y\Psi$, $P_X\Psi$ is obviously in the domain of $\halfshift$, hence also in the domain of $X = \halfshift M_{\xi_0(\cdot\, + \frac{2\pi i}3)}$. This means that $P_Y\Psi = \Psi - P_X\Psi$ is also in the domain of~$X$. The restriction of~$X$ on~$\H_X$ is bounded, since there it holds that~$X\Psi_X = M_{\xi_0(\cdot\,+\frac{\pi i}2)}\halfshift$, where $\xi_0(\theta + \frac{\pi i}2)$ is bounded.

Similarly, the decomposition $\H_X \oplus \H_Y$ is compatible with $\dom(Y)$ and $Y$ is bounded \linebreak on~$\H_Y$.
\end{proof}

\subsection{Two-particle components}\label{sa-two}
In this Section, under (S\ref{ax:re}), we show that $\chi_2(\xi) + \chi'_2(\eta) = P_2(\chi_1(\xi)\otimes \1 + \1\otimes \chi'_1(\eta))P_2$ is essentially self-adjoint. It should be stressed that this expression does not involve Friedrichs extension. Simply the restriction to $P_2\H_1^{\otimes 2}$ is already essentially self-adjoint. This is important, because we have to be able to compute the weak commutator. In the course we also show that~$\chi_2(\xi)$ is essentially self-adjoint. It is also closed under~(S\ref{ax:re}).

We provide dif\/ferent proofs. The f\/irst one is longer and tricky, but one may hope that with some more ideas the cases $n \ge 3$ could be treated. The other two works only for $\chi_2(\xi)$, but are relatively simple, and give the hints of why the general case requires a better understanding of the poles and zeros of~$S$. It should be noted that operators of the form $\shiftone\otimes\1 + M_f\big(\1\otimes\shiftone\big)M_f^*$ may in general fail to be self-adjoint (see Appendix~\ref{ce-sum}). We have to properly use the properties of~$S$.

\subsubsection{Fubini's theorem}
We f\/irst show that $\chi_2(\xi)$ is self-adjoint. For this purpose, let us note that the unitary operator $M_{\xi_0}\otimes M_{\xi_0}$ obviously commutes with $M_S^*$, hence also with $D_2(\tau_1)$ and $P_2$. As we have
\begin{gather*}
 \chi_2(\xi) = P_2(M_{\xi_0}^*\otimes M_{\xi_0}^*)\big(\shiftone\otimes \1\big)(M_{\xi_0}\otimes M_{\xi_0})P_2 \\
\hphantom{\chi_2(\xi)}{} = (M_{\xi_0}^*\otimes M_{\xi_0}^*)P_2\big(\shiftone\otimes \1\big)P_2 (M_{\xi_0}\otimes M_{\xi_0}),
\end{gather*}
the question is now reduced to the self-adjointness of $P_2\big(\shiftone\otimes \1\big)P_2$.

We start with some observations on $L^2$-functions.
\begin{Lemma}\label{lm:stripl2}
If $\Psi(\theta_1,\theta_2)$ and $\Phi(\theta_1,\theta_2)$ belong to $L^2(\RR^2)$, and $e^{\theta_1}\Psi(\theta_1,\theta_2) + e^{\theta_2}\Phi(\theta_1,\theta_2)$ is $L^2$, then on any region bounded with respect to $\theta_2$, $e^{\theta_1}\Psi(\theta_1,\theta_2)$
is~$L^2$. Similarly, on any region bounded with respect to $\theta_1$, $e^{\theta_2}\Phi(\theta_1,\theta_2)$ is $L^2$. On the region $\theta_1+\theta_2 < N$, both $e^{\theta_1}\Psi(\theta_1,\theta_2)$ and $e^{\theta_2}\Phi(\theta_1,\theta_2)$ are~$L^2$.
\end{Lemma}
\begin{proof}
Consider the region $\theta_2 < N$. Here, $e^{\theta_2}\Phi(\theta_1,\theta_2)$ is obviously $L^2$, and $e^{\theta_1}\Psi(\theta_1,\theta_2) + e^{\theta_2}\Phi(\theta_1,\theta_2)$ is by assumption $L^2$, therefore, the dif\/ference
\begin{gather*}
 e^{\theta_1}\Psi(\theta_1,\theta_2) = \big( e^{\theta_1}\Psi(\theta_1,\theta_2) + e^{\theta_2}\Phi(\theta_1,\theta_2)\big) - e^{\theta_2}\Phi(\theta_1,\theta_2)
\end{gather*}
is again $L^2$. The other case is similar.

The region $\theta_1+\theta_2 < N$ is a subset of the union of $\theta_1 < \frac{N}2$ and $\theta_2 < \frac{N}2$. The function $e^{\theta_1}\Psi(\theta_1,\theta_2)$ is obviously $L^2$ on the former and also on the latter by the f\/irst paragraph. The other case is similar.
\end{proof}

Coming back to our operators, let us summarize the situation.
\begin{itemize}\itemsep=0pt
\item $M_S$: multiplication operator on $\H^1\otimes\H^1 = L^2(\RR^2)$ by $S(\theta_2-\theta_1)$. Similarly, we denote by $M_{S(\cdot \,+ i\l)}$ the multiplication operator by $S(\theta_2-\theta_1 +i\l)$.
\item $\big(\shiftone\Psi\big)(\theta) = \Psi\big(\theta - \frac{\pi i}3\big)$ on the domain $H^2\big({-}\frac\pi 3,0\big)$.
\item $\Delta_1\otimes \Delta_1$ and $M_S$ strongly commute. It holds that $M_S\big(\shiftone\otimes \1\big) \subset \big(\shiftone\otimes \1\big)M_{S(\cdot \,+ \frac{\pi i}6)}$, since $S(\zeta)$ is bounded on $\RR + i\big(0,\frac{\pi i}6\big)$.
\end{itemize}
Let us denote the f\/lip operator by $F\colon \Psi\otimes \Phi\mapsto \Phi \otimes \Psi$ where $\Psi,\Phi \in \H_1$. Then we can write the $S$-symmetrization as follows: $P_2 = \frac12(\1+M_SF)$.

\begin{Proposition}\label{pr:sa-2}
$P_2\big(\shiftone\otimes \1\big)P_2$ is essentially self-adjoint. Under {\rm (S\ref{ax:re})}, it is even self-adjoint.
\end{Proposition}

\begin{proof}
It is well-known that $A^*A$ is self-adjoint for a closed operator $A$ \cite[Theorem~X.25]{RSII}. The operator $\big(\halfshiftone\otimes \1\big)P_2$ is closed (because $\halfshiftone\otimes\1$ is closed (self-adjoint) and $P_2$ is bounded) and densely def\/ined, hence $P_2\big(\halfshiftone\otimes \1\big)$ is closable and $\overline{P_2\big(\halfshiftone\otimes \1\big)} = \big(\big(\halfshiftone\otimes \1\big)P_2\big)^*$ \cite[Theorem~13.2]{Rudin91}. Therefore, the product $\overline{P_2\big(\halfshiftone\otimes\1\big)}\cdot \big(\halfshiftone\otimes \1\big)P_2$ is self-adjoint. The proof will be complete once we show that its domain coincides with that of $P_2\big(\shiftone\otimes \1\big)P_2$. The domain of $\overline{P_2\big(\halfshiftone\otimes \1\big)}\cdot \big(\halfshiftone\otimes \1\big)P_2$ consists of vectors $\underline{\Psi} \in \dom\big(\big(\halfshiftone\otimes \1\big)P_2\big)$ such that $\big(\halfshiftone\otimes \1\big)P_2\Psi \in \dom\Big(\,\overline{P_2\big(\halfshiftone\otimes \1\big)}\,\Big)$. We put $\Psi = \big(\halfshiftone\otimes I\big)P_2 \underline{\Psi}$. Now, what we have to prove is that $\Psi \in \dom\big(\halfshiftone\otimes I\big)$.

If $\Psi$ belongs to the domain of $\overline{P_2\big(\halfshiftone\otimes \1\big)}$, then there is a convergent sequence $\Psi_n \in \dom\big(\halfshiftone\otimes \1\big)$ such that $P_2\big(\halfshiftone\otimes \1\big)\Psi_n$ is again convergent. By the remark before this Proposition, the latter sequence is equal to $\frac12\big(\big(\halfshiftone\otimes \1\big)\Psi_n + \big(\1\otimes \halfshiftone\big)M_{S(\cdot \,+ \frac{\pi i}6)}F\Psi_n\big)$.

Let us consider the joint spectral decomposition with respect to $\log \Delta_1\otimes \1$ and $\1\otimes \log \Delta_1$. Any vector in $\H_1\otimes\H_1$ can be considered as a two-variable $L^2$-function on the spectral space of $\log \Delta_1$. Now, note that $M_{S(\cdot \,+ \frac{\pi i}6)}F$ is a bounded operator. If we put $\Phi_n = M_{S(\cdot \,+ \frac{\pi i}6)}F\Psi_n$, then~$\Psi_n$ and~$\Phi_n$ are convergent to $\Psi$ and $\Phi$ in the $L^2$-sense, respectively.

Now, by considering $\halfshiftone\otimes\1$ as a multiplication operator on the function $\Psi$ in the spectral space considered above, $\big(\halfshiftone\otimes \1\big)\Psi$ is a priori not a vector in $\H^1\otimes\H^1$: this would mean that $\Psi \in \dom\big(\halfshiftone\otimes\1\big)$, which we have to prove. Yet, as $\Psi$ is a function and $\Delta_1\otimes\1$ is a multiplication operator on the above-mentioned spectral space, it is well-def\/ined as a function (almost everywhere), a priori not $L^2$. Similarly, we can treat $\big(\1\otimes \halfshiftone\big)\Phi$ as a function. In this interpretation, the def\/ining property of $\Psi$ says that $\big(\halfshiftone\otimes \1\big)\Psi_n + \big(\1\otimes \halfshiftone\big)\Phi_n$ is $L^2$-convergent, in particular, by taking a subsequence, we may assume that the sequence is almost everywhere convergent and the pointwise limit is again $L^2$. Namely, $\big(\halfshiftone\otimes \1\big)\Psi + \big(\1\otimes \halfshiftone\big)\Phi$ is~$L^2$.

Let $Q_N$ the spectral projection of $\Delta_1\otimes\Delta_1$ corresponding to~$(e^{N-1},e^N]$. Note that on the joint spectral space of $\log\Delta_1\otimes\1$ and $\1\otimes\log\Delta_1$, this corresponds to the region $N-1 < t_1+t_2 \le N$. It commutes strongly with $\overline{P_2\big(\halfshiftone\otimes \1\big)}\cdot \big(\halfshiftone\otimes \1\big)P_2$. We f\/irst prove that, if $\underline\Psi \in \dom\Big(\,\overline{P_2\big(\halfshiftone\otimes \1\big)}\cdot \big(\halfshiftone\otimes \1\big)P_2\Big)$ and $\underline\Psi = P_2Q_n\underline\Psi,$ then $\underline\Psi \in \dom\big(P_2\big(\shiftone\otimes\1\big)P_2\big)$. Under this situation, $\big(\halfshiftone\otimes \1\big)\Psi + \big(\1\otimes \halfshiftone\big)\Phi$ is $L^2$, hence Lemma~\ref{lm:stripl2} tells that $\Psi \in \dom\big(\halfshiftone\otimes\1\big)$ or equivalently, $\underline\Psi \in \dom\big(\shiftone\otimes\1\big)$.

In other words, any vector $\underline\Psi = Q_N\underline\Psi$ in the domain of $\overline{P_2\big(\halfshiftone\otimes \1\big)}\cdot \big(\halfshiftone\otimes \1\big)P_2$ belongs to the domain of $\shiftone\otimes \1$, which means that we do not need the closure on the operator above for such $\underline\Psi$. As $Q_N$ commutes with $\overline{P_2\big(\halfshiftone\otimes \1\big)}\cdot \big(\halfshiftone\otimes \1\big)P_2$, any vector in its domain is the direct sum of such $Q_N\underline\Psi$'s, hence we obtain $\overline{P_2\big(\halfshiftone\otimes \1\big)}\cdot \big(\halfshiftone\otimes \1\big)P_2 = \overline{P_2\big(\shiftone\otimes \1\big)P_2}$ and $P_n\big(\shiftone\otimes \1\big)P_2$ is essentially self-adjoint.

We now prove that $P_2\big(\shiftone\otimes\1\big)P_2$ is closed under~(S\ref{ax:re}). Indeed, For $\underline\Psi \in \dom\big(P_2(\shiftone\otimes\1)P_2\big)$, we may assume that $P_2\underline\Psi = \underline\Psi$. We claim that $\big\<\big(\shiftone\otimes \1\big)\underline\Psi, M_SF\big(\shiftone\otimes \1\big)\underline\Psi\big\> \ge 0$. Note that we are not claiming that~$M_SF$ is a~positive operator (it is not), but the positivity holds for $\big(\shiftone\otimes \1\big)\underline\Psi$, where $\underline\Psi = P_2\underline\Psi$. We have
\begin{gather*}
 M_SF\big(\shiftone\otimes \1\big)\underline\Psi = M_S\big(\1\otimes \shiftone\big)F\underline\Psi = M_S\big(\1\otimes \shiftone\big)M_S^*\underline\Psi,
\end{gather*}
since $\underline\Psi$ is $S$-symmetric: $\underline\Psi = M_SF\underline\Psi$. Now, the inner product in question is
\begin{gather}
 \big\<\big(\shiftone\otimes \1\big)\underline\Psi, M_S\big(\1\otimes \halfshiftone\big)\big(\1\otimes \halfshiftone\big)M_S^*\underline\Psi\big\> \nonumber\\
\qquad{} = \big\<M_S^*\big(\shiftone\otimes \1\big)\underline\Psi,\big(\1\otimes \halfshiftone\big)\big(\1\otimes \halfshiftone\big)M_S^*\underline\Psi\big\> \nonumber\\
 \qquad{} = \big\<\big(\halfshiftone\otimes \1\big)M_{S(\cdot\, + \frac{\pi}6)}^*\big(\halfshiftone\otimes \1\big)\underline\Psi,\big(\1\otimes \halfshiftone\big)\big(\1\otimes \halfshiftone\big)M_S^*\underline\Psi\big\>.\label{eq:cross1}
\end{gather}
By Lemma \ref{lm:symmetric-cont2}, $M_{S(\cdot\, + \frac{\pi}6)}^*\big(\halfshiftone\otimes \1\big)\underline\Psi \in \dom\big(\1\otimes \halfshiftone\big)$ and $\big(\1\otimes \halfshiftone\big)M_S^*\underline\Psi \in \dom\big(\halfshiftone\otimes \1\big)$. Then by \cite[Lemma~B.1]{CT15-1}, we obtain
\begin{gather*}
 \eqref{eq:cross1} = \big\<\big(\1\otimes \halfshiftone\big)M_{S(\cdot\, + \frac{\pi}6)}^*\big(\halfshiftone\otimes \1\big)\underline\Psi, \big(\halfshiftone\otimes \1\big)\big(\1\otimes \halfshiftone\big)M_S^*\underline\Psi\big\> \\
\hphantom{\eqref{eq:cross1}}{} = \big\<M_{S(\cdot\, + \frac{\pi}3)}^*\big(\halfshiftone\otimes \halfshiftone\big)\underline\Psi, M_S^*\big(\halfshiftone\otimes \halfshiftone\big)\underline\Psi\big\> \\
\hphantom{\eqref{eq:cross1}}{} = \big\<\big(\halfshiftone\otimes \halfshiftone\big)\underline\Psi, M_{S(\cdot\, + \frac{\pi}3)}M_S^*\big(\halfshiftone\otimes \halfshiftone\big)\underline\Psi\big\>,
\end{gather*}
where in the second step we used Lemma~\ref{lm:symmetric-cont2}. Now, we assume~(S\ref{ax:re}) that the real part of $S\big(\theta + \frac{\pi i}3\big)S(-\theta) = S\big(\theta + \frac{2\pi i}3\big)$ is positive. Therefore, this scalar product has the positive real part. Now the closedness of $P_2\big(\shiftone\otimes\1\big)P_2 = \frac12(\1 + M_SF)\big(\shiftone\otimes\1\big)P_2$ follows from Lemma~\ref{lm:sum-closed}.
\end{proof}

\begin{Proposition}
$\chi'_2(\eta) = P_2(\1\otimes \chi'_1(\eta))P_2$ is essentially self-adjoint, and under~{\rm (S\ref{ax:re})} it is self-adjoint.
\end{Proposition}
\begin{proof}
Parallel to the case of $\chi_2(\xi)$. One can also prove it by considering the operator $J_2P_2(\chi_1(\xi)\otimes \1)P_2J_2$.
\end{proof}

\begin{Proposition}\label{pr:sa-sum}
Under {\rm (S\ref{ax:re})}, $\chi_2(\xi) + \chi'_2(\eta) = P_2(\chi_1(\xi)\otimes \1 + \1\otimes \chi'_1(\eta))P_2$ is essentially self-adjoint.
\end{Proposition}
\begin{proof}
For simplicity of notation we write $A = \chi_1(\xi)^\frac12$, $B = \chi'_1(\eta)^\frac12$. Recall that $A = M_{\xi_0(\cdot\,+\frac{2\pi i}3)}^*\halfshiftone M_{\xi_0(\cdot\,+\frac{2\pi i}3)}$ and $B = M_{\eta_0(\cdot\,-\frac{2\pi i}3)}^*\halfshiftione M_{\eta_0(\cdot\,-\frac{2\pi i}3)}$, where we have $\xi_0\in\H^\infty(\SS_{\frac{\pi i}3, \frac{2\pi i}3})$, $\eta_0\in\H^\infty(\SS_{-\frac{2\pi i}3, -\frac{\pi i}3})$. It is straightforward to observe that $AB$ and $BA$ are bounded. Our goal is the essential self-adjointness of $P_2(A^2\otimes\1 + \1\otimes B^2)P_2$.

First we claim that $\overline{P_2(A\otimes \1 + \1\otimes B)} \subset \overline{P_2(A\otimes \1)} + \overline{P_2(\1\otimes B)}$. In order to apply Lem\-ma~\ref{lm:closed}, it is enough to check that $c\|\Psi\|^2 \le \re \<P_2(A\otimes \1)\Psi, P_2(\1\otimes B)\Psi\>$. This holds indeed:
 \begin{gather*}
 \<P_2(A\otimes \1)\Psi, P_2(\1\otimes B)\Psi\> = \frac14 \<(\1 + D_2(\tau_1))(A\otimes \1)\Psi,(\1 + D_2(\tau_1))(\1\otimes B)\Psi\> \\
\hphantom{\<P_2(A\otimes \1)\Psi, P_2(\1\otimes B)\Psi\>}{} = \frac12 (\<(A\otimes B)\Psi, \Psi\> + \<(A\otimes \1)\Psi, M_S(B\otimes\1) M_S^*\Psi\>),
 \end{gather*}
and by noting that $A\otimes B$ is positive, while $(A\otimes \1)M_S(B\otimes\1) = M_{S(\cdot\, +\frac{\pi i}6)}(AB\otimes\1)$ is bounded.

We show that, if $(A\otimes \1 + \1\otimes B)\Psi \in \dom(\overline{P_2(A\otimes \1)})$, then it is in $\dom(A\otimes \1)$. This follows the idea of Proposition~\ref{pr:sa-2} (see the third and fourth paragraphs).
Namely, we consider the following formal expression with $\tilde\Psi =
(A\otimes \1 + \1\otimes B)\Psi$,
\begin{gather*}
     \big\<\big(A\otimes \1 + (\1\otimes A) M_{S(\cdot\,+\frac{\pi i}6)}F\big)\tilde\Psi,  \big(A\otimes \1 + (\1\otimes A) M_{S(\cdot\,+\frac{\pi
i}6)}F\big)\tilde\Psi\big\>.
\end{gather*}
This is formal, yet the whole expression has the meaning as an $L^1$ integral over the spectral space with respect to $\log A\otimes\1$ and $\1\otimes \log A$: it is $L^1$ because
 \begin{gather*}
 \big(A\otimes \1 + (\1\otimes A) M_{S(\cdot\,+\frac{\pi i}6)}F\big)(A\otimes \1 + \1\otimes B)\Psi
 \end{gather*}
has the meaning as an $L^2$-function. As $\tilde\Psi = (A\otimes \1 + \1\otimes B)\Psi$, this inner product can be formally decomposed as
 \begin{gather*}
 \big\|(A\otimes \1)\tilde\Psi\big\|^2 + \big\|(\1\otimes A)M_{S(\cdot\, + \frac{\pi i}6)}F\tilde\Psi\big\|^2
 + 2\re \big\<(A\otimes \1)\tilde\Psi, (\1\otimes A)M_{S(\cdot\, + \frac{\pi i}6)}F\tilde\Psi\big\>,
 \end{gather*}
where each of these terms can be understood as an integral over the above-mentioned spectral space. If the f\/irst term is f\/inite, it means that $\tilde\Psi$ is in the domain of $A\otimes \1$ and we are done. If not, namely the f\/irst term is not f\/inite, it must be positive inf\/inite in the sense of Lebesgue integral. Yet the whole expression is~$L^1$ by the def\/ining property of~$\Psi$ (that $(A\otimes \1 + \1\otimes B)\Psi \in \dom(\overline{P_2(A\otimes \1)})$) and the second term is positive (actually in this case it must be positive inf\/inite as well), thus the last term must be negative inf\/inite again in the sense of Lebesgue integral: a~mixed inf\/inity is impossible, because then the whole integral would have a measurable set on which it is positive inf\/inite, which contradicts the assumed $L^1$-property.

On the spectral space of $\log A$, we can write $A$ as the multiplication operator by $e^t$. The crossing term in question $\<(A\otimes \1)\tilde\Psi, (\1\otimes A)M_{S(\cdot\, + \frac{\pi i}6)}F\tilde\Psi\>$ is the following form:
 \begin{gather*}
 \int dt_1\,dt_2\, e^{t_1+t_2}\tilde\Phi(t_1,t_2)\overline{\tilde\Psi(t_1,t_2)},
 \end{gather*}
where $\tilde\Phi = M_{S(\cdot\, + \frac{\pi i}6)}F\tilde\Psi$.

When the real part of an integral is negative inf\/inite, one can apply the Fubini theorem. We insert the spectral projection $Q_N$ of $(A\otimes A)$ corresponding to $\big(e^{N-1},e^N\big]$, which amounts to consider a strip $N-1 < t_1 + t_2 \le N$. This is equivalent to the change of variable and integrating f\/irst with respect to $t_2-t_1$, which is legitimate by Fubini's theorem. On this strip, the integral is $L^1$, because $\tilde\Psi$ and $\tilde\Phi$ are~$L^2$.

Now recall that $\tilde\Psi = (A\otimes\1 + \1\otimes B)\Psi$, therefore, the above integral consists of four contributions. Even though the expression above is formal, note that $Q_N$ commutes with $A\otimes \1$ and $\1\otimes A$, while $AB$, $BA$ are bounded. Therefore, if we take two of the contributions
 \begin{gather*}
\big\<Q_N(A\otimes A)(\1\otimes B)\Psi, M_{S(\cdot\,+\frac{\pi i}6)}F(A\otimes \1 + \1\otimes B)\Psi\big\> \\
\qquad{} = \big\<Q_N(\1\otimes AB)(A\otimes \1)\Psi, M_{S(\cdot\,+\frac{\pi i}6)}F(A\otimes \1 + \1\otimes B)\Psi\big\>,
 \end{gather*}
 where $(A\otimes \1)\Psi$ and $(A\otimes \1 + \1\otimes B)\Psi$ are vectors in $\H_1\otimes\H_1$ and the operators appearing there are all bounded, hence it is f\/inite in the sense of Lebesgue integral. The same argument applies to the term
 \begin{gather*}
 \big\<Q_N(A\otimes A)(A\otimes \1)\Psi, M_{S(\cdot\,+\frac{\pi i}6)}F(\1\otimes B)\Psi\big\> \\
 \qquad{} = \big\<(A\otimes \1)\Psi, M_{S(\cdot\,+\frac{\pi i}6)}F Q_N(\1\otimes AB)(A\otimes \1)\Psi\big\>,
 \end{gather*}
 which implies that the only remaining term
 \begin{gather*}
 \big\<Q_N(A\otimes A)(A\otimes \1)\Psi, M_{S(\cdot\,+\frac{\pi i}6)}F(A\otimes \1)\Psi\big\>
 \end{gather*}
 must have the real part which is negative inf\/inite in the sense of Lebesgue integral. However, we have computed this and it is positive under (S\ref{ax:re}). Hence the whole integral is bounded below, even after removing~$Q_N$. This contradicts the assumption that the whole integral was negative inf\/inite, hence we obtain that each term of the formal expression is f\/inite, namely, $(A\otimes \1 + \1\otimes B)\Psi \in \dom(A\otimes \1)$.

Furthermore, we note that $A\otimes \1$ and $(A\otimes \1 + \1\otimes B)$ commute strongly. By considering the joint spectral decomposition, it is immediate to see that
 \begin{gather*}
 \Psi \in \dom(A\otimes \1 + \1\otimes B) \textrm{ and } (A\otimes \1 + \1\otimes B)\Psi \in \dom(A\otimes \1) \\
 \qquad \Longleftrightarrow \ \Psi \in \dom(A\otimes \1) \textrm{ and }(A\otimes \1)\Psi \in \dom(A\otimes \1 + \1\otimes B).
 \end{gather*}
 To summarize, we proved
 \begin{gather*}
 \overline{P_2(A\otimes \1)}(A\otimes \1 + \1\otimes B)P_2= P_2(A\otimes \1)(A\otimes \1 + \1\otimes B)P_2 \\
\hphantom{ \overline{P_2(A\otimes \1)}(A\otimes \1 + \1\otimes B)P_2}{} = P_2(A\otimes \1 + \1\otimes B)(A\otimes\1)P_2.
 \end{gather*}
 Similarly, we obtain $\overline{P_2(\1\otimes B)}(A\otimes \1 + \1\otimes B)P_2 = P_2(\1\otimes B)(A\otimes \1 + \1\otimes B)P_2 = P_2(A\otimes \1 + \1\otimes B)(\1\otimes B)P_2$.

By the above consideration, we obtain the equality
 \begin{gather*}
 \overline{P_2(A\otimes \1 + \1\otimes B)}(A\otimes \1 + \1\otimes B)P_2
 \subset \big(\overline{P_2(A\otimes \1)} + \overline{P_2(\1\otimes B)}\big)(A\otimes \1 + \1\otimes B)P_2 \\
 \hphantom{\overline{P_2(A\otimes \1 + \1\otimes B)}(A\otimes \1 + \1\otimes B)P_2}{}
 = \left(P_2(A\otimes \1) + P_2(\1\otimes B)\right)(A\otimes \1 + \1\otimes B)P_2 \\
 \hphantom{\overline{P_2(A\otimes \1 + \1\otimes B)}(A\otimes \1 + \1\otimes B)P_2}{}
 = P_2\big(A^2\otimes \1 + 2(A\otimes B) + \1\otimes B^2\big)P_2,
 \end{gather*}
 which is self-adjoint, because of the f\/irst expression.

We show that $P_2(A^2\otimes \1 + \1\otimes B^2)P_2$ is essentially self-adjoint, using the commutator Theo\-rem~\ref{th:commutator}, by taking $P_2\left(A^2\otimes \1 + 2(A\otimes B) + \1\otimes B^2 + c\1\right)P_2$ as the reference operator with some $c > 0$. Let us check the required estimates
 \begin{gather*}
 \big\|P_2\big(A^2\otimes \1 + 2(A\otimes B) + \1\otimes B^2 + c\1\big)\Psi\big\|^2 \\
\qquad{}\ge\big\|P_2\big(A^2\otimes \1 + \1\otimes B^2\big)\Psi\big\|^2 + 4\|P_2(A\otimes B)\Psi\|^2 \\
\quad\qquad{} + 4\re \big\<P_2\big(A^2\otimes \1 + \1\otimes B^2\big)\Psi, P_2(A\otimes B)\Psi\big\> \\
\quad\qquad{} + 2c\big\<\big(A^2\otimes \1 + 2(A\otimes B) + \1\otimes B^2 \big)\Psi,\Psi\big\>.
 \end{gather*}
In order to show $\|P_2(A^2\otimes \1+ \1\otimes B^2)\Psi\|^2 \le c_1\|P_2\big(A^2\otimes \1 + 2(A\otimes B) + \1\otimes B^2 + c\1\big)\Psi\|^2$,
it is enough to have a lower bound on the following term:
\begin{gather*}
4 \big\<P_2\big(A^2\otimes \1 + \1\otimes B^2\big)\Psi, P_2(A\otimes B)\Psi\big\>\\
 \qquad{} = 2\big\<\big(A^2\otimes \1 + \1\otimes B^2\big)\Psi, (A\otimes B)\Psi\big\> \\
\quad\qquad{}+2\big\<\big(A^2\otimes \1 + \1\otimes B^2\big)\Psi, M_S(B\otimes \1)M_S^*D_2(\tau_1)(A\otimes\1)\Psi\big\>,
 \end{gather*}
where the f\/irst term is positive, while we have $\big\<\big(A^2\otimes \1\big)\Psi, M_S(B\otimes \1)M_S^*D_2(\tau_1)(A\otimes\1)\Psi\big\> = \<x(A\otimes \1)\Psi, (A\otimes \1)\Psi\>$, where $x$ is a bounded operator. This is bounded by $\|x\|\cdot \|(A\otimes \1)\Psi\|^2 = \|x\|\cdot \<(A^2\otimes \1)\Psi,\Psi\>$, therefore, we take $c > 2\|x\|$ and $c_1 = 1$. We can estimate analogously the term containing $\1\otimes B^2$.

A bound of the weak commutator
\begin{gather*}
 \big[P_2\big(A^2\otimes \1 + 2(A\otimes B) + \1\otimes B^2\big)P_2, P_2\big(A^2\otimes \1 + \1\otimes B^2\big)P_2\big]
\end{gather*}
reduces to the weak commutator $ [P_2 (A\otimes B )P_2, P_2\big(A^2\otimes \1 + \1\otimes B^2\big)P_2]$, and its bound by
 \begin{gather*}
\big\|\big(P_2\big(A^2\otimes \1 + 2(A\otimes B) + \1\otimes B^2+ c\1\big)P_2\big)^\frac12\Psi\big\|^2\\
\qquad {} = \big\<\big(A^2\otimes \1 + 2(A\otimes B) + \1\otimes B^2+ c\1\big)\Psi, \Psi\big\>
 \end{gather*}
follows from the estimate above of the term $\big\<P_2\big(A^2\otimes \1 + \1\otimes B^2\big)\Psi, P_2(A\otimes B)\Psi\big\>$, up to the constant $c$ above, by noting that $A\otimes B$ commutes with $A^2\otimes \1$ and $\1\otimes B^2$.
\end{proof}

In Proposition \ref{pr:n-sum-closed+} we show that $\chi_2(\xi) + \chi'_2(\eta) + 4\sqrt{c(\xi,\eta)}$ is closed, hence self-adjoint.

\subsubsection[A shorter proof for $\chi_2(\xi)$]{A shorter proof for $\boldsymbol{\chi_2(\xi)}$}\label{shorter}

We here present a proof of the essential self-adjointness of $\chi_2(\xi)$. Although it is not of our direct interest as we need the self-adjointness of $\chi_2(\xi)+\chi'_2(\eta)$, we believe it clarif\/ies the reason why our method does not work for $n\ge 3$.

We observe that it is enough to show that $\chi_1(\xi)\otimes \1 + M_S(\1\otimes \chi_1(\xi))M_S^*$ is self-adjoint, because this operator and $P_2 = \frac12 (\1 + D_2(\tau_1))$ (strongly) commute. We denote again $\chi_1(\xi) = A^2$.

\begin{Proposition}\label{pr:sa-2-cleave}
 $A^2\otimes \1 + M_S\big(\1\otimes A^2\big)M_S^*$ is essentially self-adjoint.
\end{Proposition}
\begin{proof}
We consider the operator $A\otimes A^{-1} + M_S\big(A^{-1}\otimes A\big)M_S^*$. Let $X = A^\frac12\otimes A^{-\frac12}$, $Y = A^{-\frac12}\otimes A^\frac12 M_S^*$. Then $Y^* = M_S\big(A^{-\frac12}\otimes A^\frac12\big)$ and $\dom(Y^*) = \dom\big(A^{-\frac12}\otimes A^\frac12\big)$. One can easily check the assumptions of Lemma~\ref{lm:cleave}: as for the decomposition of the Hilbert space, we change the variables $t_+ := \theta_2 + \theta_1$, $t_- := \theta_2 - \theta_1$. Then the problem reduces to one variable $t_-$ and one can take the same decomposition as Proposition~\ref{pr:sa-1}, by noting that $S$ is bounded in $\RR + i\big(0,\frac\pi 6\big)$.

Finally, both $A\otimes A^{-1}$ and $M_S\big(A^{-1}\otimes A\big)M_S^*$ commute with $A\otimes A$ strongly. Therefore, the bounded operator $Q_N(A\otimes A)$ commute with both of them, where $Q_N$ is the spectral projection of $A\otimes A$ onto $\big(e^{N-1},e^N\big]$ and the product is
 \begin{gather*}
 Q_N(A\otimes A)\big(A\otimes A^{-1} + M_S\big(A^{-1}\otimes A\big)M_S^*\big) = Q_N\big(A^2\otimes \1 + M_S\big(\1\otimes A^2\big)M_S^*\big),
 \end{gather*}
which is still self-adjoint. Our operator $A\otimes A^{-1} + M_S\big(A^{-1}\otimes A\big)M_S^*$ is an extension of the direct sum of these components, therefore, it is essentially self-adjoint.
\end{proof}

It appears dif\/f\/icult to apply the same idea to $P_2\big(A^2\otimes\1 + \1\otimes B^2\big)P_2$, as $A$ and $B$ do not commute.

If $n\ge 3$, $A\otimes A\otimes A$ cannot reduce the problem to one variable. We know no other convenient operator.

\subsubsection{Konrady's trick}\label{konrady}
We give a less simpler proof which however gives an insight of why the problem is complicated, cf.\ Section~\ref{sa-many}. For it, we need an additional assumption on~$S$.
\begin{gather}\label{ineq:K}
 \re S\left(\theta + \frac{\pi i}6\right)S(-\theta) \ge 0.
\end{gather}
We indicate some supporting evidence in Appendix~\ref{conditionK} that there should exist examples of~$S$ which satisfy this condition. We do not pursue this condition further, as our main result (Proposition~\ref{pr:sa-sum}) has been proved without it.

The general idea used here is called Konrady's trick \cite{Konrady71}, \cite[Section~X.2]{RSII}.
\begin{Lemma}\label{lm:konrady}
Let $X$ and $Y$ be symmetric operators on $\dom(X)$ and $\dom(Y)$, respectively. Assume that $X + Y$ is self-adjoint on $\dom(X+Y) = \dom(X) \cap \dom(Y)$ and it holds that $\re \<X\Psi, Y\Psi\> \ge 0$ for $\Psi \in \dom(X + Y)$. Then $X$ and $Y$ are essentially self-adjoint on \mbox{$\dom(X+Y)$}.
\end{Lemma}
\begin{proof}
 Compute:
 \begin{gather*}
 \|(X+Y)\Psi\|^2 = \|X\Psi\|^2 + \|Y\Psi\|^2 + 2\re\<X\Psi,Y\Psi\> \ge \max\big\{\|X\Psi\|^2, \|Y\Psi\|^2\big\},
 \end{gather*}
 now one can use W\"ust's theorem \cite[Theorem~X.14]{RSII} to see that $X = (X+Y)-Y$ and $Y = (X+Y)-X$ are essentially self-adjoint on $\dom(X+Y)$.
\end{proof}

\begin{Proposition}\label{pr:sa-2-k}
For $S$ satisfying \eqref{ineq:K}, $\shiftone\otimes \1+ M_S\big(\1\otimes \shiftone\big)M_S^*$ is essentially self-adjoint.
\end{Proposition}
\begin{proof}
First, $\halfshiftone\otimes \1 + c\big(\1 \otimes \halfshiftone\big)$ is self-adjoint by the spectral theorem for a constant $c > 0$. We choose $c > |S(\theta +\frac{\pi i}6)|$, $\theta \in \RR$.

Then, we claim that $\halfshiftone\otimes \1 + c\big(\1\otimes \halfshiftone\big) + M_S\big(\1\otimes \halfshiftone\big)M_S^*$ is self-adjoint. Indeed, for $\Psi \in \dom\big(\1\otimes \halfshiftone\big)$, we have $\big\|M_S\big(\1\otimes \halfshiftone\big)M_S^*\Psi\big\| = \big\|M_{S(\cdot\, + \frac{\pi i}6)}\big(\1\otimes \halfshiftone\big)\Psi\big\|$ which is bounded by~$c\big\|\big(\1\otimes \halfshiftone\big)\Psi\big\|$. Therefore, one can apply Kato--Rellich theorem \cite[Theorem~X.12]{RSII}.

Next, we show that $\halfshiftone\otimes \1 + M_S\big(\1\otimes \halfshiftone\big)M_S^*$ is essentially self-adjoint on the domain $\dom\big(\halfshiftone\otimes \1 + c\big(\1\otimes \halfshiftone\big) + M_S\big(\1\otimes \halfshiftone\big)M_S^*\big)$. In order to apply Lemma~\ref{lm:konrady}, we only have to check that $\big\<\big(\halfshiftone\otimes\1\big)\Psi, \big(\1\otimes \halfshiftone\big)\Psi\big\> \ge 0$ (trivial) and $\big\<\big(\1\otimes \halfshiftone\big)\Psi, M_S\big(\1\otimes \halfshiftone\big)M_S^*\Psi\big\> \ge 0$. As $\Psi \in \dom\big(\1\otimes \halfshiftone\big)$, it holds that $M_S\big(\1\otimes \halfshiftone\big)M_S^*\Psi = M_SM_{\overline{S(\cdot\, + \frac{\pi i}6)}}\big(\1\otimes \halfshiftone\big)\Psi$. The inequality above follows from~\eqref{ineq:K} and hermitian analyticity~(S\ref{ax:hermitian}). Actually, by the same assumption~\eqref{ineq:K}, it follows from Lemma~\ref{lm:sum-closed} and the argument of Proposition~\ref{pr:sa-2}, $\halfshiftone\otimes \1 + M_S\big(\1\otimes \halfshiftone\big)M_S^*$ is closed on its natural domain, hence it is self-adjoint there.

Note that $\halfshiftone\otimes \1 + M_S\big(\1\otimes \halfshiftone\big)M_S^*$ commutes strongly with $\Delta_1\otimes \Delta_1$. Consider its spectral projection $Q_{[0,N]}$ corresponding to the interval $[0,N]$.

As $S(\theta_2 - \theta_1 + i\l)$ is uniformly bounded for $\l \in [0,\frac\pi 6]$, we have $M_S\big(\1\otimes\halfshiftone\big) \subset \big(\1\otimes\halfshiftone\big)M_{S(\cdot \,+ \frac{\pi i}6)}$ and
 \begin{gather*}
 \big(\halfshiftone\otimes \1\big) \cdot M_S\big(\1\otimes\halfshiftone\big)M_S^*Q_{[0,N]} \subset \big(\halfshiftone\otimes \halfshiftone\big) \cdot M_{S(\cdot \,+ \frac{\pi i}6)}M_S^*Q_{[0,N]}
 \end{gather*}
 is bounded. Similarly,
 \begin{gather*}
 M_S\big(\1\otimes\halfshiftone\big)M_S^*\cdot \big(\halfshiftone\otimes \1\big)Q_{[0,N]}
 \subset M_S \big(\halfshiftone\otimes \halfshiftone\big)M_{S(\cdot\,+\frac{\pi i}6)}^*Q_{[0,N]}.
 \end{gather*}
 Now $\big(\halfshiftone\otimes \1 + M_S\big(\1\otimes \halfshiftone\big)M_S^*)^2Q_{[0,N]}$ is self-adjoint. From the boundedness above, we can expand the square:
 \begin{gather*}
\big(\halfshiftone\otimes \1 + M_S\big(\1\otimes \halfshiftone\big)M_S^*\big)^2Q_{[0,N]}\\
\qquad{}
 \subset \Big(\shiftone\otimes \1 + \big(\halfshiftone\otimes \halfshiftone\big)M_{S(\cdot \,+ \frac{\pi i}6)}M_S^*\\
 \qquad \hphantom{\subset \Big(}{} + M_S \big(\halfshiftone\otimes \halfshiftone\big)M_{S(\cdot\,+\frac{\pi i}6)}^*
 + M_S\big(\1\otimes\shiftone\big)M_S^*\Big) Q_{[0,N]}.
 \end{gather*}
Here, the left-hand side is self-adjoint, and the right-hand side is symmetric, hence it must be an equality and the right-hand side is self-adjoint. The second and the third terms are bounded, therefore, they have no ef\/fect on domains and $\big(\shiftone\otimes \1 + M_S\big(\1\otimes \shiftone\big)M_S^*\big)Q_{[0,N]}$ is self-adjoint. As $N$ is arbitrary, we obtain the claim.
\end{proof}

\subsection{Many-particle components}\label{sa-many}
\begin{Proposition}\label{pr:n-sum-closed}
Under {\rm (S\ref{ax:re})}, $\chi_n(\xi)$ is closed.
\end{Proposition}
\begin{proof}
 We show the closedness of $\sum D_n(\rho_k)\big(\shiftone\otimes\1\otimes\cdots\otimes\1\big)D_n(\rho_k)^*$ on $P_n\H_n$.

To prove it by induction, we may assume that $\chi_{n-1}(\xi)$, hence $\chi_{n-1}(\xi)\otimes\1$ is closed. To apply Lemma~\ref{lm:sum-closed}, we only have to check that the following is positive:
 \begin{gather*}
\re\<D_n(\rho_k)\big(\shiftone\otimes\1\otimes\cdots\otimes\1\big)D_n(\rho_k)^*\Psi_n,
 D_n(\rho_n)\big(\shiftone\otimes\1\otimes\cdots\otimes\1\big)D_n(\rho_n)^*\Psi_n\> \\
 \qquad{} = \re\<D_n(\rho_k)\big(\shiftone\otimes\1\otimes\cdots\otimes\1\big)\Psi_n,
 D_n(\rho_n)\big(\shiftone\otimes\1\otimes\cdots\otimes\1\big)\Psi_n\> \\
 \qquad{} = \re\<(\shiftone\otimes\1\otimes\cdots\otimes\1)\Psi_n,
 D_n(\tau_{1,n})(\shiftone\otimes\1\otimes\cdots\otimes\1)\Psi_n\> \\
 \qquad{} = \re\<\big(\shiftone\otimes\1\otimes\cdots\otimes\1\big)\Psi_n,
 D_n(\tau_{1,2})\big(\shiftone\otimes\1\otimes\cdots\otimes\1\big)\Psi_n\>,
 \end{gather*}
where $\tau_{k,l}$ is the transposition between $k$ and $l$, we used the fact that $\rho_k$ and $\rho_n$ may be replaced by $\tau_{1,k}$, $\tau_{1,n}$, respectively (see \cite[proof of Theorem~3.4]{CT15-1}) and we can further rewrite $\tau_{1,2} = \tau_{2,3}\cdots\tau_{n-1,n}\tau_{1,n}\tau_{n-1,n}\cdots\tau_{2,3}$, and $D_n(\tau_{k,k+1})\big(\shiftone\otimes\1\otimes\cdots\otimes\1\big)\Psi_n
= \big(\shiftone\otimes\1\otimes\cdots\otimes\1\big)\Psi_n$ for $k\ge 2$. From here, we can use Lemma~\ref{lm:symmetric-cont2} and the same computation as at the end of Proposition~\ref{pr:sa-2} under~(S\ref{ax:re}) to see that the real part of the above inner product is positive.
\end{proof}

\begin{Proposition}\label{pr:n-sum-closed+}
 Under {\rm (S\ref{ax:re})}, $\chi_n(\xi) + \chi'_n(\eta)$ is closed.
\end{Proposition}
\begin{proof}
This can be done by the estimate of Proposition \ref{pr:positive}
\begin{gather*}
 \re\<\chi_n(\xi)\Psi, \chi'_n(\eta)\Psi\> \ge -nc(\xi,\eta)\<(\chi(\xi) + \chi'(\eta)+2)\Psi, \Psi\>
\end{gather*}
and by Lemma~\ref{lm:positive-closed} applied to $X = \chi_n(\xi)$, $Y = \chi'_n(\eta)$ (here $X$ and $Y$ are positive, hence we can take the constant $c$ in Lemma~\ref{lm:positive-closed} larger if necessary).
\end{proof}

\begin{Proposition}
 We assume {\rm (S\ref{ax:regularity})}. For a fixed $n$, there is $\epsilon$ such that $P_n(\Delta_1^\epsilon\otimes \1\otimes\cdots\otimes\1)P_n$
 is self-adjoint.
\end{Proposition}
\begin{proof}
 We follow the idea of Proposition~\ref{pr:sa-2-k}. We set $S_k(\pmb{\theta}) := \prod\limits_{j < k} S(\theta_k-\theta_j)$ as before. There is $\epsilon > 0$ such that
 \begin{gather*}
 \re \overline{S_k(\theta_1+2\pi \epsilon i, \dots, \theta_k+2\pi \epsilon i)}S_k(\theta_1,\dots, \theta_k) \ge 0,
 \end{gather*}
because $S_k$ is uniformly continuous in a small neighborhood of~$\RR^n$ as we assume~(S\ref{ax:regularity}) and $\overline{S_k(\pmb{\theta})} S_k(\pmb{\theta}) = 1$.

First consider $\Delta_1^\epsilon\otimes\cdots\otimes \1 + \cdots + c\1\otimes\cdots\otimes \Delta_1^\epsilon\otimes\cdots \otimes \1 + \cdots + c\1\otimes\cdots\otimes \1\otimes \Delta_1^\epsilon$, which is self-adjoint, where $c \ge 0$.

Next,
\begin{gather*}
\Delta_1^\epsilon\otimes\cdots\otimes \1 + \cdots + c\1\otimes\cdots\otimes \1\otimes \Delta_1^\epsilon + \sum_k M_{S_k}(\1\otimes \cdots \otimes \1\otimes \underset{k\mathrm{-th}}{\Delta_1^\epsilon}\otimes \1\cdots\otimes \1)M_{S_k}^*
\end{gather*}
is self-adjoint if $c$ is suf\/f\/iciently large by Kato--Rellich perturbation, as in Proposition~\ref{pr:sa-2-k}.

Now in order to prove that $\sum_k M_{S_k}(\1\otimes \cdots \otimes \1\otimes \underset{k\mathrm{-th}}{\Delta_1^\epsilon}\otimes \1\cdots\otimes \1)M_{S_k}^*$ is essentially self-adjoint and to apply Lemma~\ref{lm:konrady}, it is enough to check that the real part of the following is positive:
 \begin{gather*}
 \Big\<(\1\otimes \cdots \otimes\underset{j\mathrm{-th}}{\Delta_1^\epsilon}\otimes \cdots \otimes \1)\Psi,
 M_{S_k}(\1\otimes \cdots \otimes\underset{k\mathrm{-th}}{\Delta_1^\epsilon}\otimes \cdots \otimes \1)M_{S_k}^*\Psi\Big\>.
 \end{gather*}
 If $j > k$, the operators in the product commute and it is positive. If $j = k$, it reduces to
 \begin{gather*}
 \Big\<(\1\otimes \cdots \otimes\underset{k\mathrm{-th}}{\Delta_1^\epsilon}\otimes \cdots \otimes \1)\Psi,
 M_{S_k\overline{S_k(\cdot\, + 2\pi \epsilon i)}}(\1\otimes \cdots \otimes\underset{k\mathrm{-th}}{\Delta_1^\epsilon}\otimes \cdots \otimes \1)\Psi\Big\>,
 \end{gather*}
 and by the choice of $\epsilon$, this is positive. If $j < k$, we get
 \begin{gather*}
 \Big\<(\1\otimes \cdots \otimes\underset{j\mathrm{-th}}{\Delta_1^\epsilon}\otimes \cdots \otimes \1)M_{\check S_{k,j}}^*\Psi,\;
 M_{S_{k,j}\overline{S_{k,j}(\cdot\, + 2\pi \epsilon i)}}(\1\otimes \cdots \otimes\underset{k\mathrm{-th}}{\Delta_1^\epsilon}\otimes \cdots \otimes \1)M_{\check S_{k,j}}^*\Psi\Big\> \\
 \qquad{}= \Big\<\Delta_{1,k,j}^{\frac\epsilon 2} M_{\check S_{k,j}}^*\Psi,\;
 M_{S_{k,j}(\cdot\, + 2\pi \epsilon i)\overline{S_{k,j}(\,\cdot)}}\Delta_{1,k,j}^{\frac \epsilon 2}M_{\check S_{k,j}}^*\Psi\Big\>,
 \end{gather*}
where $\Delta_{1,k,j} = \1\otimes \cdots \otimes\underset{j\mathrm{-th}}{\Delta_1}\otimes \cdots \otimes\underset{k\mathrm{-th}}{\Delta_1}\cdots \otimes \1$, which is positive again by the choice of $\epsilon$, where we introduced $S_{k,j}(\underline \theta) = S(\theta_k-\theta_j)$ and $\check S_{k,j}(\underline \theta) = S_k(\underline \theta) S(\theta_k-\theta_j)^{-1}$.

The closedness for suf\/f\/iciently small $\epsilon$ follows by arguing as in Proposition~\ref{pr:n-sum-closed}.
\end{proof}

The method of this proof is invalid for $\epsilon = \frac16$, or any f\/ixed $\epsilon$. Indeed, as $k$ increases, $\re S_k(\pmb{\theta} + \frac\pi 6)S_k(-\pmb{\theta})$ may take negative values and the Konrady's trick fails.

In order to have the right domain of self-adjointness of $P_n\big(\shiftone\otimes\1\cdots\otimes\1\big)P_n$, we may have to correctly incorporate the zeros and poles of $S$ to the domain (cf.~\cite{Tanimoto15-1}).

\section{Outlook}
In order to complete the construction of Borchers triples, we have to prove that $\chi_n(\xi) + \chi'_n(\eta)$ is (essentially) self-adjoint.
We also hope to clarify the role played by the assumption~(S\ref{ax:re}), and to drop it if possible. With the same assumption, one should be able to prove the Bisognano--Wichmann property and modular nuclearity~\cite{BL04, Lechner08} for an inclusion with a minimal distance~\cite{Alazzawi14},
hence it should be possible to construct Haag--Kastler nets with minimal size, which will be published elsewhere \cite{CT15-3}.

On the other hand, if the strong commutativity of our candidates fails, then the construction of the Haag--Kastler nets for a given $S$-matrix with poles will be really a hard problem: if the net exists at all, the polarization-free generators are canonically constructed \cite{BBS01}, while we checked that our $\fct(f)$'s are formally compatible with the form factor program \cite[Section~4.1]{CT15-1}. If our candidates are not the right polarization-free generators, the formal computations must fail, which means that the right polarization-free generators should have even subtler domains. A~related problem is whether the $S$-matrix is a complete invariant for asymptotically complete nets. This is in general open even for the simplest case where $S=1$, if the temperateness of the polarization-free generators is not assumed~\cite{Mund12}.

\appendix
\section[Properties of the $S$-matrix]{Properties of the $\boldsymbol{S}$-matrix}\label{existenceS}
Here we show the existence of two-particle $S$-matrix which satisf\/ies some additional properties. Recall that the most general $S$-matrix satisfying (S\ref{ax:unitarity})--(S\ref{ax:atzero}) and (S\ref{ax:regularity}) takes the form\footnote{In \cite[Appendix A]{CT15-1}, we used a parameter $B = \frac12 + \frac{3\e}{\pi i}$, $0 < B < 1$, and had the equivalent expression $S(\theta) = \frac{\tanh\frac12\left(\theta + \frac{2\pi i}3\right)}{\tanh\frac12\left(\theta - \frac{2\pi i}3\right)} \cdot \frac{\tanh\frac12\left(\theta + \frac{(B-2)\pi i}3\right)}{\tanh\frac12\left(\theta - \frac{(B-2)\pi i}3\right)} \frac{\tanh\frac12\left(\theta - \frac{B\pi i}3\right)}{\tanh\frac12\left(\theta + \frac{B\pi i}3\right)} \cdot S_{\mathrm{Blaschke}}(\theta)$.}
\begin{gather*}
 S(\theta) = \frac{\tanh\frac12\left(\theta + \frac{2\pi i}3\right)}{\tanh\frac12\left(\theta - \frac{2\pi i}3\right)} \cdot
\frac{\tanh\frac12\left(\theta - \frac{\pi i}2-i\e\right)}{\tanh\frac12\left(\theta - \frac{\pi i}2+i\e\right)}
 \frac{\tanh\frac12\left(\theta - \frac{\pi i}6-i\e\right)}{\tanh\frac12\left(\theta + \frac{\pi i}6+i\e\right)}
 \cdot S_{\mathrm{Blaschke}}(\theta),
\end{gather*}
where $-\frac\pi6 < \e < \frac\pi 6$ and $S_{\mathrm{Blaschke}}(\theta)$ is a f\/inite Blaschke product which satisf\/ies the conditions of \cite[Appendix~A]{CT15-1}\footnote{A Blaschke product is a bounded analytic function on $\RR + i(0,\pi)$ which can be written
as a possibly inf\/inite product of functions $c_\a\frac{e^\theta + \a}{e^\theta - \bar \a}$, where $|c_\a| = 1$
\cite[Theorem~15.21]{Rudin87}, \cite[Appendix~A]{LW11}. In particular, such a Blaschke product has modulus $1$ on $\RR$ and $\RR + i\pi$.}.

\subsection{Examples with property (S\ref{ax:re})}
Let us consider the simplest case where $S_{\mathrm{Blaschke}}(\theta) = 1$. We show that for $-\frac\pi6 < \e < \frac\pi 6$ the property~(S\ref{ax:re}) holds.

We set
\begin{gather}\label{eq:se}
 S(\theta) = \underline{S}(\theta)\cdot S_\e(\theta), \\
 \underline{S}(\theta) = \frac{\tanh\frac12\left(\theta + \frac{2\pi i}3\right)}{\tanh\frac12\left(\theta - \frac{2\pi i}3\right)}
 = -\frac{\sinh\frac12\left(\theta + \frac{\pi i}3\right)\sinh\frac12\left(\theta + \frac{2\pi i}3\right)}{\sinh\frac12\left(\theta - \frac{\pi i}3\right)\sinh\frac12\left(\theta - \frac{2\pi i}3\right)}, \nonumber\\
 S_\e(\theta)
 = \frac{\tanh\frac12\left(\theta - \frac{\pi i}2-i\e\right)}{\tanh\frac12\left(\theta - \frac{\pi i}2+i\e\right)}
 \frac{\tanh\frac12\left(\theta - \frac{\pi i}6-i\e\right)}{\tanh\frac12\left(\theta + \frac{\pi i}6+i\e\right)}\nonumber\\
\hphantom{S_\e(\theta)}{}
=\frac{\sinh\frac12\!\left(\theta - \frac{\pi i}6 - i\e\right)\sinh\frac12\!\left(\theta - \frac{3\pi i}6+ i\e\right)\sinh\frac12\!\left(\theta - \frac{3\pi i}6-i\e\right)\sinh\frac12\left(\theta - \frac{5\pi i}6+i\e\right)}{\sinh\frac12\!\left(\theta + \frac{\pi i}6+i\e\right)\sinh\frac12\left(\theta + \frac{3\pi i}6-i\e\right)\sinh\frac12\!\left(\theta + \frac{3\pi i}6+i\e\right)\sinh\frac12\!\left(\theta + \frac{5\pi i}6 -i\e\right)}.\nonumber
\end{gather}

We use freely $\sinh\a \cdot \sinh\b = \frac{\cosh(\a + \b) - \cosh(\a - \b)}2$, in particular,
\begin{gather*}
 \sinh\frac12(\theta +i\a) \cdot \sinh\frac12(\theta + i\b) = \frac{\cosh\big(\theta + \frac{(\a + \b)i}2\big) - \cos\big(\frac{(\a-\b)}2\big)}2.
\end{gather*}

For $\underline{S}$, we have
\begin{gather*}
 \underline{S}\left(\theta+\frac{\pi i}3\right) =
 -\frac{\sinh\frac12\left(\theta + \frac{2\pi i}3\right)}{\sinh\frac12\theta} \cdot
 \frac{\sinh\frac12\left(\theta + \pi i\right)}{\sinh\frac12\left(\theta - \frac{\pi i}3\right)} \\
 \hphantom{\underline{S}\left(\theta+\frac{\pi i}3\right)}{} = -
 \frac{\sinh\frac12\left(\theta + \frac{2\pi i}3\right)}{\sinh\frac12\left(\theta - \frac{\pi i}3\right)} \cdot
 \frac{i\cosh\frac12\theta}{\sinh\frac12\theta} \cdot
 \frac{\sinh\frac12\left(\theta + \frac{\pi i}3\right)}{\sinh\frac12\left(\theta + \frac{\pi i}3\right)} \\
\hphantom{\underline{S}\left(\theta+\frac{\pi i}3\right)}{} = -
 \frac{\cosh\left(\theta + \frac{\pi i}2\right) - \cos\frac\pi 6}{\cosh\theta - \cos\frac{\pi}3}\cdot
 \frac{i\cosh\frac12\theta}{\sinh\frac12\theta}\\
\hphantom{\underline{S}\left(\theta+\frac{\pi i}3\right)}{}
= \frac{\sinh\theta +i \cos\frac\pi 6}{\cosh\theta - \cos\frac{\pi}3}\cdot \frac{\cosh\frac12\theta}{\sinh\frac12\theta},
\end{gather*}
and then $S_\e$:
\begin{gather*}
 S_\e\left(\theta+\frac{\pi i}3\right) \\
= \frac{\sinh\frac12\left(\theta + \frac{\pi i}6-i\e\right)
 \sinh\frac12\left(\theta - \frac{\pi i}6+i\e\right)
 \sinh\frac12\left(\theta - \frac{\pi i}6-i\e\right)
 \sinh\frac12\left(\theta - \frac{\pi i}2+i\e\right)}
 {\sinh\frac12\left(\theta + \frac{\pi i}2+i\e\right)
 \sinh\frac12\left(\theta + \frac{5\pi i}6-i\e\right)
 \sinh\frac12\left(\theta + \frac{5\pi i}6+i\e\right)
 \sinh\frac12\left(\theta + \frac{7\pi i}6-i\e\right)} \\
= \frac{\sinh\frac12\left(\theta + \frac{\pi i}6-i\e\right)
 \sinh\frac12\left(\theta - \frac{\pi i}6+i\e\right)
 \sinh\frac12\left(\theta - \frac{\pi i}6-i\e\right)
 \sinh\frac12\left(\theta - \frac{\pi i}2+i\e\right)}
 {\sinh\frac12\left(\theta + \frac{\pi i}2+i\e\right)
 \sinh\frac12\left(\theta - \frac{7\pi i}6-i\e\right)
 \sinh\frac12\left(\theta + \frac{5\pi i}6+i\e\right)
 \sinh\frac12\left(\theta - \frac{5\pi i}6-i\e\right)} \\
 = \frac{\cosh\theta - \cos\left(\frac\pi 6-\e\right)}
 {\cosh\theta - \cos\left(\frac{5\pi} 6 + \e\right)}
 \cdot
 \frac{\cosh\left(\theta - \frac{\pi i}3\right) - \cos\left(\frac\pi 6 - \e\right)}
 {\cosh\left(\theta - \frac{\pi i}3\right) - \cos\left(\frac{5\pi} 6 + \e\right)}.
\end{gather*}
The f\/irst factor in the last expression is positive. Let us look at the second factor:
\begin{gather*}
 \frac{\cosh\left(\theta - \frac{\pi i}3\right) - \cos\left(\frac\pi 6 - \e\right)}
 {\cosh\left(\theta - \frac{\pi i}3\right) - \cos\left(\frac{5\pi} 6 + \e\right)}\\
 \qquad{}
 = \frac{\left(\cosh\left(\theta - \frac{\pi i}3\right) - \cos\left(\frac\pi 6 - \e\right)\right)\left(\cosh\left(\theta + \frac{\pi i}3\right) - \cos\left(\frac{5\pi} 6 + \e\right)\right)}{\left(\cosh\left(\theta - \frac{\pi i}3\right) - \cos\left(\frac{5\pi} 6 + \e\right)\right)\left(\cosh\left(\theta + \frac{\pi i}3\right) - \cos\left(\frac{5\pi} 6 + \e\right)\right)} \\
\qquad{} = \frac{\left(\cosh\left(\theta - \frac{\pi i}3\right) - \cos\left(\frac\pi 6 - \e\right)\right)\left(\cosh\left(\theta + \frac{\pi i}3\right) - \cos\left(\frac{5\pi} 6 + \e\right)\right)}{\frac{\cosh 2\theta}2 + \frac14 - \cos\left(\frac{5\pi} 6 + \e\right)\cosh \theta + \cos\left(\frac{5\pi} 6 + \e\right)^2}
\end{gather*}
and the denominator is positive, as it is the square of the modulus of a complex number. Using $\cosh(\a + i\b) = \cosh\a\cos\b + i\sinh\a\sin\b$ and $\cos\big(\frac{5\pi} 6+\e\big) = - \cos\big(\frac{\pi} 6-\e\big)$, we rewrite the numerator:
\begin{gather*}
 \left(\cosh\left(\theta - \frac{\pi i}3\right) - \cos\left(\frac\pi 6-\e\right)\right)
 \left(\cosh\left(\theta + \frac{\pi i}3\right) - \cos\left(\frac{5\pi} 6+\e\right)\right) \\
\qquad{} = \frac{\cosh 2\theta + \cos\frac{2\pi}3}2 - \cos\left(\frac\pi 6-\e\right)^2 - i2\cos\left(\frac\pi 6-\e\right)\cdot \sinh\theta \cdot\sin\frac\pi 3 \\
\qquad{} = \frac{\cosh 2\theta}2 - i\sqrt 3\cos\left(\frac\pi 6-\e\right)\cdot \sinh\theta - \frac14 - \cos\left(\frac\pi 6-\e\right)^2.
\end{gather*}
The full expression is now:
\begin{gather*}
 S\left(\theta + \frac{\pi i}3\right)
 = \frac{\sinh\theta +i \cos\frac\pi 6}{\cosh\theta - \cos\frac{\pi}3} \cdot\frac{\cosh\frac12\theta}{\sinh\frac12\theta}
 \times\frac{\cosh\theta - \cos\left(\frac\pi 6-\e\right)}
 {\cosh\theta - \cos\left(\frac{5\pi} 6 + \e\right)} \\
\hphantom{S\left(\theta + \frac{\pi i}3\right)=}{} \times
 \frac{\frac{\cosh 2\theta}2 - i\sqrt 3\cos\left(\frac\pi 6-\e\right)\cdot \sinh\theta - \frac14 - \cos\left(\frac\pi 6-\e\right)^2}
 {\frac{\cosh 2\theta}2 + \frac14 - \cos\left(\frac{5\pi} 6 + \e\right)\cosh \theta + \cos\left(\frac{5\pi} 6 + \e\right)^2},
\end{gather*}

To check that $\re S(\theta) \ge 0$, we see that it is proportional (i.e., up to a positive number) to
\begin{gather*}
\re \left[\frac{\sinh\theta + i\cos\frac\pi 6}{\sinh\frac12\theta}\left(\frac{\cosh 2\theta}2 - i\sqrt 3\cos\left(\frac\pi 6-\e\right)\cdot \sinh\theta - \frac14 - \cos\left(\frac\pi 6-\e\right)^2\right)\right] \\
\qquad{} = \frac{\sinh\theta}{\sinh\frac12\theta}\left(\frac{\cosh 2\theta}2 - \frac14 - \cos\left(\frac\pi 6-\e\right)^2 + \frac{3}2\cos\left(\frac\pi 6-\e\right)\right),
\end{gather*}
which is positive for $- \frac\pi 6 < \e < \frac\pi 6$.

The minimal $S$-matrix corresponds to the Bullough--Dodd model \cite{CT15-1}. We expect that it holds also with some Blaschke product. Yet, it is clear that some Blaschke products violate the condition. Indeed, if a Blaschke product takes a value with nontrivial imaginary part at $\theta + \frac\pi 3$, then the product of its $n$-th power and the minimal $S$ may have the negative real part, then the product $S$-matrix violates the condition.

\subsection{Property (S\ref{ax:bound})}\label{bound}
We show that any $S$-matrix of the form \eqref{eq:se} satisf\/ies (S\ref{ax:bound}). It suf\/f\/ices to show it for this minimal case, as the general case follows because $\big|S_{\mathrm{Blaschke}}\big(\theta + \frac{\pi i}6\big)\big| < 1$.

Look at
\begin{gather*}
 \left|S\left(\theta + \frac{\pi i}6\right)\right|\\
\quad{}
 = \left|\frac{\sinh\frac12\left(\theta + \frac{\pi i}2\right)\sinh\frac12\left(\theta + \frac{5\pi i}6\right)}{\sinh\frac12\left(\theta - \frac{\pi i}6\right)\sinh\frac12\left(\theta - \frac{\pi i}2\right)} \right| \\
 \qquad{}\times\left|\frac{\sinh\frac12\left(\theta - i\e\right)\sinh\frac12\left(\theta - \frac{\pi i}3+ i\e\right)\sinh\frac12\left(\theta - \frac{\pi i}3-i\e\right)\sinh\frac12\left(\theta - \frac{2\pi i}3+i\e\right)}{\sinh\frac12\left(\theta + \frac{\pi i}3+i\e\right)\sinh\frac12\left(\theta + \frac{2\pi i}3-i\e\right)\sinh\frac12\left(\theta + \frac{2\pi i}3+i\e\right)\sinh\frac12\left(\theta + \pi i -i\e\right)} \right|\\
\quad{} = \left|\frac{\sinh\frac12\left(\theta + \frac{5\pi i}6\right)}{\sinh\frac12\left(\theta - \frac{\pi i}6\right)}\right| \times\left|\frac{\sinh\frac12\left(\theta - i\e\right)\sinh\frac12\left(\theta - \frac{\pi i}3+ i\e\right)}{\sinh\frac12\left(\theta + \frac{2\pi i}3+i\e\right)\sinh\frac12\left(\theta + \pi i -i\e\right)} \right| \\
\quad {}= \left|\frac{\sinh\frac12\left(\theta + \frac{5\pi i}6\right)}{\sinh\frac12\left(\theta + \pi i -i\e\right)}\right| \times\left|\frac{\sinh\frac12\left(\theta - i\e\right)}{\sinh\frac12\left(\theta - \frac{\pi i}6\right)}\right|\times\left|\frac{\sinh\frac12\left(\theta - \frac{\pi i}3+ i\e\right)}{\sinh\frac12\left(\theta + \frac{2\pi i}3+i\e\right)} \right|,
\end{gather*}
where in the last expression each factor is less than $1$, as $-\frac\pi 6 < \e < \frac\pi 6$, therefore, $\big|S\big(\theta + \frac{\pi i}6\big)\big| < 1$.

\subsection{A further inequality}\label{conditionK}
In Section~\ref{konrady} we used the inequality $\re S\big(\theta + \frac{\pi i}6\big)S(-\theta) \ge 0$. We expect that this holds for some values of~$\e$ in~\eqref{eq:se}. Indeed, we checked numerically that this holds for $\e = \frac\pi 6$.

Although this value is outside the region of our consideration (as the poles get cancelled by the zeros of the other factors), $\re S\big(\theta + \frac{\pi i}6\big)S(-\theta)$ depends uniformly on $\e$, hence for $\e$ suf\/f\/iciently close to $\frac\pi 6$ we expect that the inequality should hold.

\section{Weakly but not strongly commuting operators}\label{counterexample}
Let~$A$, $B$ be self-adjoint operators with domains $\dom(A)$, $\dom(B)$ with the dense intersection $\dom(A)\cap\dom(B)$. It is well-known that, even if~$A$ and~$B$ weakly commute, namely it holds that $\<A\Psi_1, B\Psi_2\> = \<B\Psi_1, A\Psi_2\>$ for $\Psi_k \in \dom(A)\cap\dom(B)$, it does not imply that the spectral projections of~$A$ and~$B$ commute.

A classical (counter)example is due to Nelson \cite[Section~10]{Nelson59}: one takes a Riemann surface with a nontrivial topology and considers it as a real manifold and the $L^2$-space on it. If one chooses two dif\/ferent translations, their generators on the $L^2$-space commute weakly (actually as operators on a common core), as the derivations in these directions commute. However, the two translations do not commute, as the Riemann surface
has a nontrivial topology.

The example above shows clearly why two operators weakly commute but not strongly: weak commutation is a local property, while strong commutation is global. The question of strong commutativity is not just technical and, if it fails, there is often a~good reason for it. See~\cite{Schmuedgen84}, \cite[Example~5.5, Exercise~6.16 ]{Schmuedgen12} for some other examples. We present here two more families of counterexamples.

\subsection{From canonical commutation relation}\label{cef}
This example is essentially due to Faddeev and Volkov \cite{FV08}. Consider a CCR pair $X = M_{\id}$, where $\id(t) = t$ is the identity function and $P = -i\frac{d}{dt}$ on $L^2(\RR)$. Note that $(e^{s_1X}\xi)(t) = e^{s_1t}\xi(t)$, $(e^{s_2P}\xi)(t) = \xi(t-is_2)$ on suitable domains. If we take $s_2 = \frac{2\pi}{s_1}$, then $(e^{s_1X}e^{s_2P}\xi)(t) = (e^{s_2P}e^{s_1X}\xi)(t) = e^{s_1t}\xi(t-\frac{2\pi}{s_2})$. It is not dif\/f\/icult to f\/ind a common dense domain of~$e^{s_1X}$ and~$e^{s_2P}$, and they weakly commute. Yet they do not strongly commute. Indeed we know that $\{e^{s_1X}, e^{s_2P}\}'' = \B(L^2(\RR))$, while strong commutativity would imply $\{e^{s_1X}, e^{s_2P}\}''$ to be abelian, which is not true.

\subsection{From bound state operators}\label{ceb}
Take two Blaschke products $f_1$, $f_2$ on $\RR + i(-\pi,0)$ with the symmetry condition $f_j(\theta - \pi i) = \overline{f_j(\theta)}$. They extend meromorphically to $\CC$ and satisfy $f_j(\theta - 2\pi i) = f_j(\theta)$. Form $A_1 := M_{\overline {f_1}}\Delta M_{f_1}$, $A_2 := M_{\overline {f_2}}\Delta M_{f_2}$, where $(\Delta\xi)(\theta) = \xi(\theta - 2\pi i)$. They are manifestly self-adjoint and, by a similar consideration as \cite[Section~5.2]{Tanimoto15-1}, one can conclude that their domains are determined by the zeros and poles of $f_j$. Especially one can see that the domains $\dom(A_1A_2)$ and $\dom(A_2A_1)$ are dense: indeed, there are analytic functions with specif\/ied zeros at the poles of $f_1$, $f_2$ in $\RR + i(-4\pi,0)$. On such a function $\xi$ it is clear that $A_1A_2\xi = A_2A_1\xi$, because on this domain it holds that $M_{f_j}\Delta\xi = \Delta M_{f_j}\xi$, as $f_j$ are $2\pi i$-periodic. Yet, $A_1$ and $A_2$ do not strongly commute. Indeed, the intersection of their domains is not a core for any of them if $f_1$ and $f_2$ have dif\/ferent zeros or poles, which also follows from the ideas of \cite[Section~3]{Tanimoto15-1}.

\section{The Driessler--Fr\"ohlich theorem}\label{DF}
First let us recall the commutator theorem of Glimm--Jaf\/fe--Nelson. The theorem roughly says the following: if $T$ is a positive self-adjoint operator, $A$ is a symmetric operator and if $A$ and $[T,A]$ can be estimated by $T$, then $A$ is essentially self-adjoint. Yet, depending on how to make such estimates, there are certain variations in the hypothesis of the theorem.

In one of such variations, one estimates $A$ and $[T,A]$ as bilinear forms, def\/ined on $\dom\big(T^\frac12\big)$ and $\dom\big(T^\frac32\big)$ (or on their cores), respectively \cite[Theorem~X.36$'$]{RSII}. When $T$ is a complicated operator, it is dif\/f\/icult to have control over $T^\frac12$ and $T^\frac32$.
Another variation requires to estimate~$A$ as an operator, namely $\|A\psi\| \le c\|T\psi\|$ for $\psi \in \dom(T)$, or for a core of~$T$.
For our situation, this latter version has a better chance to apply (see \cite[Theorem~X.37]{RSII}):
\begin{Theorem}[the commutator theorem]\label{th:commutator}
Let $T$ be a self-adjoint operator with $T \ge \1$. Suppose that $A$ is a symmetric operator defined on a core $\Dom$ of $T$ so that
 \begin{enumerate}[\rm (1)]\itemsep=0pt
 \item\label{comm1} $\|A\psi\| \le c_1\|T\psi\|$ for all $\psi \in \Dom$,
 \item\label{comm2} $|\<T\psi,A\psi\> - \<A\psi, T\psi\>| \le c_2\|T^\frac12\psi\|^2$ for all $\psi \in \Dom$.
 \end{enumerate}
Then, $A$ is essentially self-adjoint on any core of $T$.
\end{Theorem}

The Driessler--Fr\"ohlich theorem provides a suf\/f\/icient condition for two symmetric opera\-tors~$A$,~$B$ to strongly commute. As the hypothesis, one should have a positive self-adjoint operator~$T$ which can estimate $A$, $B$, $[T,A]$, $[T,B]$. Again, there can be two variations: estimates either as bilinear forms or operators. The original proof \cite[Theorem 3.1]{DF77} is written for bilinear forms. In our application, the operator $T$ is complicated and we need a variation.

Furthermore, let us note the following subtle point. The original paper assumes certain estimates of $A$, $B$, $[T,A]$, $[T,B]$, $[T,[T,A]]$, $[T,[T,B]]$ ``on some form core of $T$'', which does not seem to suf\/f\/ice: a form core of $T$ might not be a form core of~$T^\frac32$ (see also \cite[remark after Theorem~X.37]{RSII}).

Although the essential idea is the same as the original~\cite{DF77}, we here present the adapted proof for the sake of clarity. Note that the domain of the weak commutativity is also modif\/ied: the original proof requires only the weak commutativity on a core of $T^2$, while here we need it on a~core of~$T$. By the same reason as explained in \cite[remark after Theorem~X.37]{RSII}, the results might fail by a tiny change in assumptions, therefore, one should not underestimate the importance of these careful examinations.

\begin{Theorem}[Driessler--Fr\"ohlich]\label{th:DF}
Let $T\ge \1$ be a positive self-adjoint operator, $A$ and $B$ symmetric operator on a core $\Dom$ of $T$. Assume that there are positive numbers $c_1$, $c_2$, $c_3$ such that
 \begin{enumerate}[\rm (1)]\itemsep=0pt
 \item\label{df1} $\|A\psi\| \le c_1\|T\psi\|$ and $\|B\psi\| \le c_1\|T\psi\|$ for all $\psi \in \Dom$,
 \item\label{df2} for all $\varphi,\psi \in \Dom$,
 \begin{itemize}\itemsep=0pt
 \item $|\<T\psi,A\varphi\> - \<A\psi, T\varphi\>| \le c_2\big\|T^\frac12\psi\big\|\cdot \big\|T^\frac12\varphi\big\|$,
 \item $|\<T\psi,B\varphi\> - \<B\psi, T\varphi\>| \le c_2\big\|T^\frac12\psi\big\|\cdot \big\|T^\frac12\varphi\big\|$,
 \end{itemize}
 \item\label{df3} for all $\varphi,\psi \in \Dom$,
 \begin{itemize}\itemsep=0pt
 \item $|\<T\psi,A\varphi\> - \<A\psi, T\varphi\>| \le c_3 \|\psi\|\cdot \|T\varphi\|$,
 \item $|\<T\psi,B\varphi\> - \<B\psi, T\varphi\>| \le c_3 \|\psi\|\cdot \|T\varphi\|$,
 \end{itemize}
 \item\label{df4} $\<A\psi, B\varphi\> = \<B\psi, A\varphi\>$ for all $\psi, \varphi \in \Dom$.
 \end{enumerate}
 Then $A$, $B$ are essentially self-adjoint on any core of $T$ and they strongly commute.
\end{Theorem}

\begin{proof}
Note that all the assumptions hold on $\dom(T)$. Indeed, for any pair $\psi, \varphi \in \dom(T)$, there are $\{\psi_n\}, \{\varphi_n\} \subset \Dom$ such that $\psi_n \to \psi$, $T\psi_n \to T\psi$, therefore, $T^\frac12\psi_n \to T^\frac12\psi$ and $\varphi_n \to \varphi$, $T\varphi_n \to T\varphi$, therefore, $T^\frac12\varphi_n \to T^\frac12\varphi$. From this,~$A$ and $B$ can be naturally extended to $\dom(T)$ and all the statements and the weak commutativity hold there.

Now, the claim that $A$ and $B$ are essentially self-adjoint of any core of $T$ follows from the assumptions (\ref{df1}), (\ref{df2}) and Theorem~\ref{th:commutator}.

Next, we observe that, from assumption (\ref{df3}), the commutators $[T,A]$, $[T,B]$ are well-def\/ined on $\dom(T)$: for a f\/ixed $\varphi \in \dom(T)$, $\<T\psi,B\varphi\> - \<B\psi, T\varphi\>$ is an antilinear form in $\psi$ and bounded by $c_3 \|\psi\|\cdot \|T\varphi\|$. By Riesz' representation theorem, there is a unique vector $\eta$ such that $\<T\psi,B\varphi\> - \<B\psi, T\varphi\> = \<\psi, \eta\>$ and $\|\eta\| \le c_3 \|T\varphi\|$. We denote the correspondence $\varphi \mapsto \eta$ by $[T,B]$. It follows by this def\/inition that $\|[T,B]T^{-1}\| \le c_3$.

Let us show next that $(B+z)T^{-1}\Dom$ (which is well-def\/ined since $B$ is extended to $\dom(T)$) is in the domain of $T$ and $T(B+z)T^{-1}\Dom$ is dense in $\H$, where $z$ is a complex number such that $|\im z| > c_3$. For $\psi \in \dom(T)$ and $\varphi \in T^{-1}\Dom$, the following is meaningful:
 \begin{gather*}
 \<T\psi, (B+z)\varphi\> = \<(B+\bar z)\psi, T\varphi\> + (\<T\psi,B\varphi\> - \<B\psi, T\varphi\>).
 \end{gather*}
As $\varphi \in T^{-1}\Dom$, $T\varphi \in \Dom \subset \dom(T) \subset \dom(B)$
 and we have
 \begin{gather*}
 \<(B+\bar z)\psi, T\varphi\> = \<\psi, (B+z)T\varphi\>,
 \end{gather*}
which implies that $|\<(B+\bar z)\psi, T\varphi\>| \le \|\psi\|\cdot \|(B+z)T\varphi\|$. On the other hand, by assumption~(\ref{df3}) we have $|\<T\psi,B\varphi\> - \<B\psi, T\varphi\>| \le c_3 \|\psi\|\cdot \|T\varphi\|$. Altogether, we have the following estimate
 \begin{gather*}
 |\<T\psi, (B+z)\varphi\>| \le \|\psi\|\cdot(\|(B+z)T\varphi\| + c_3 \|T\varphi\|)
 \end{gather*}
 for any $\psi \in \dom(T)$, and therefore, $(B+z)\varphi \in \dom(T^*) = \dom(T)$, and
 \begin{gather*}
 \<\psi, T(B+z)\varphi\> = \<\psi, (B+z)T\varphi\> + \<T\psi,B\varphi\> - \<B\psi, T\varphi\>,
 \end{gather*}
 or equivalently by using $[T,B]$ def\/ined as in the previous paragraph,
\begin{gather*}
 T(B+z)\varphi = ((B+z)T + [T,B])\varphi = \big((B+z) + [T,B]T^{-1}\big)T\varphi.
\end{gather*}

Now, let us assume that $\<\psi, T(B+z)\varphi\> = 0$ for arbitrary $\varphi \in T^{-1}\Dom$. By the above expression of~$T(B+z)\varphi$, this is equivalent to
 \begin{gather*}
 0 = \<\psi, ((B+z)T + [T,B])\varphi\> = \big\<\psi, \big((B+z) + [T,B]T^{-1}\big)\tilde\varphi\big\>,
 \end{gather*}
where $\tilde\varphi \in \Dom$ is arbitrary. As $\Dom$ is a core of $T$, it is a core of $B$ as well. Since $B$ is essentially self-adjoint on $\Dom$ and $[T,B]T^{-1}$ is bounded, it follows that $(B+\bar z + ([T,B]T^{-1})^*)\psi = 0$. If $|\im z| > c_3$, this implies that $\psi = 0$, namely, $T(B+z)\varphi$ form a dense subspace.

From this and the assumption that $T \ge \1$, it follows that $(B+z)\Dom$ is a core of $T$ \cite[Theo\-rem~X.26]{RSII}, and therefore, a core of~$A$. It follows from this that $(A+z_1)(B+z_2)\Dom$ is dense in $\H$ for $z_1\in \CC\setminus \RR$, $|\im z_2| > c_3$.

 Our goal is to show that the resolvents $R_A(z_1) = (A+z_1)^{-1}$ and $R_B(z_2)$ commute (for~$z_1$,~$z_2$ in an open set, which implies the strong commutativity. Let us take $\psi,\varphi \in \Dom$. We compute
 \begin{gather*}
 \<(A+\bar z_1)(B+\bar z_2)\psi, R_A(z_1)R_B(z_2)(A+z_1)(B+z_2)\varphi\> \\
 \qquad{} = \<(A+\bar z_1)\psi, (B+z_2)\varphi\> \\
 \qquad{} = \<(B+\bar z_2)\psi, (A+z_1)\varphi\> \\
\qquad {}= \<(A+\bar z_1)(B+\bar z_2)\psi, \;R_B(z_2)R_A(z_1)(A+z_1)(B+z_2)\varphi\>,
 \end{gather*}
 where the second equality follows from the weak commutativity. We saw that the sets $(A+\bar z_1)(B+\bar z_2)\Dom$ and $(A+z_1)(B+z_2)\Dom$ are both dense in $\H$, which completes the proof.
\end{proof}

\section{Nontrivial extensions of sum operators}\label{ce-sum}
We show here that an operator which looks close to $\chi_2(\xi)$ fails to be self-adjoint.

\begin{Proposition}
Let $f$ be a bounded analytic function in $\RR + i(-\frac\pi 3, 0)$, $|f(\theta)| = 1$ and $a\in\RR$. Then $a\shiftone + M_f^*\shiftone M_f$ has nontrivial self-adjoint extensions if the function $a + \overline{f(\bar \zeta - \frac{\pi i}3)}f(\zeta)$ has zeros in the strip or has a nontrivial singular inner factor.
\end{Proposition}
\begin{proof}
We have
\begin{gather*}
M_{a+\overline{f(\cdot)} f(\cdot\, - \frac{\pi i}3)} \shiftone = a\shiftone + M_f^*M_{f(\cdot\, - \frac{\pi i}3)}\shiftone = a\shiftone + M_f^*\shiftone M_f,
\end{gather*}
because the domain of the right-hand side is that of $\shiftone$. And we know what extensions the left-hand side has~\cite{Tanimoto15-1}. Namely, if the function $a + \overline{f(\bar \zeta - \frac{\pi i}3)}f(\zeta)$ has a Blaschke factor which has zeros in $\RR + i(-\frac\pi 3, 0)$, or if it has a singular inner factor, then the operator $M_{a+\overline{f(\cdot)} f(\cdot\, - \frac{\pi i}3)} \shiftone$ has nontrivial extensions.
\end{proof}

As a corollary, one can show that $\shiftone\otimes \1 + M_f\big(\1\otimes \shiftone\big)M_f^*$ is \textit{not always} self-adjoint, if $f_2(\theta_1,\theta_2) := f(\theta_1+\theta_2) = \frac{\theta_1 + \theta_2 +i\a}{\theta_1 + \theta_2 - i\a}$, where $\a < -\frac\pi 3$. Indeed, each term commutes strongly with $\halfshiftone\otimes\halfshiftione$. By considering the sum on the spectral subspace with respect to $\halfshiftone\otimes\halfshiftione$, where it can be regarded as a constant $a_0$, the sum can be reduced to an operator of the form $\frac1{a_0}\halfshiftone\otimes\halfshiftone + a_0M_f\big(\halfshiftone\otimes\halfshiftone\big) M_f^*$, which is further reduced to $a\shiftone + M_f^*\shiftone M_f$. We saw that this is not always self-adjoint, depending on $f$ and $a$.

\section[Observations on $S$-symmetric functions]{Observations on $\boldsymbol{S}$-symmetric functions}\label{S-symmetric}
For a Hilbert space $\K$, let $L^2(\RR,\K)$ denote the space of $\K$-valued $L^2$-functions, and for an open subset $\U$ of $\RR^d$, $H^2(\RR^d + i\U, \K)$ be the Hardy space of $\K$-valued functions $\Psi$ which have an analytic continuation in $\RR^d + i\U$ such that $\Psi(\,\cdot + i\pmb{\l})$, $\pmb{\l} \in \U$, belongs to $L^2(\RR^d)$ and their $L^2$-norms are uniformly bounded with respect to $\l$.
\begin{Lemma}\label{lm:riesz}
Let $\Psi(\theta) \in L^2(\RR,\K)$ and $\psi \in L^2(\RR)$. Then the integral $\int d\theta\,\overline{\psi(\theta)} \Psi(\theta)$ defines a~vector in $\K$.
\end{Lemma}
\begin{proof}\looseness=1\sloppy
$\Psi(\theta) \in L^2(\RR,\K)$ means that $\|\Psi(\theta)\|$ is an $L^2$-function, hence $\|\overline{\psi(\theta)}\Psi(\theta)\|$ is in~$L^1(\RR)$.
\end{proof}

\begin{Lemma}\label{lm:l2convergence}
Let $\{\Psi_n\}$ be a sequence in $L^2(\RR, \K)$ and assume that there is an $\psi \in L^2(\RR)$ which satisfies $\|\Psi_n(\theta)\| \le \psi(\theta)$ $($almost everywhere$)$. If $\Psi_n$ is almost everywhere pointwise convergent to $\Psi$, then $\Psi$ is again $L^2$ and $\Psi_n \to \Psi$ in the $L^2$-sense.
\end{Lemma}
\begin{proof}
We have the estimate $\|\Psi_n(\theta)\| \le \psi(\theta)$, therefore, the pointwise limit is again bounded: $\|\Psi(\theta)\| \le \psi(\theta)$, especially, $\Psi$ is~$L^2$. In addition, we have the estimate $\|\Psi_n(\theta)-\Psi_m(\theta)\|^2 \le 4 \psi(\theta)$, thus by Lebesgue's dominated convergence,
 \begin{gather*}
 \int d\theta\, \|\Psi_n(\theta)-\Psi_m(\theta)\|^2
 \end{gather*}
converges to $0$. Namely, $\{\Psi_n\}$ is convergent in $L^2(\RR, \K)$.
\end{proof}

\begin{Lemma}\label{lm:h2k}
Let $\K$ be a Hilbert space. Consider the operator $\Delta_1\otimes\1$ on $L^2(\RR)\otimes\K$. By identifying $L^2(\RR)\otimes\K \cong L^2(\RR, \K)$, we have $\dom(\Delta_1\otimes\1) = H^2(\SS_{-\pi,0}, \K)$ $(= H^2(\RR + i(-\pi,0),\K))$.
\end{Lemma}
\begin{proof}
This can be reduced to the case where $\K = \CC$ (see, e.g., \cite[Proposition~A.1]{Tanimoto15-1}) and the standard facts about Banach-space valued holomorphic functions (e.g., \cite[Chapter~3, holomorphic functions]{Rudin91}).
\end{proof}

\begin{Lemma}\label{lm:zero}
Let $\K$ be a Hilbert space and $\Psi \in H^2\big(\SS_{-\frac{\epsilon}{2\pi},\frac{\epsilon}{2\pi}},\K\big)$, namely, $\Psi(\theta) \in L^2(\RR,\K)$ is a $\K$-valued function in $\dom(\Delta^\epsilon\otimes \1)\cap\dom(\Delta^{-\epsilon}\otimes \1)$, for some $\e > 0$ $($see Lemma~{\rm \ref{lm:h2k})}. We assume further that $\Psi(0) = 0$. Let $f(\zeta)$ be a $\CC$-valued meromorphic function on $\RR + i(-\frac{\epsilon}{2\pi}, \frac{\epsilon}{2\pi})$ with the only simple pole at $\zeta = 0$ and bounded outside a neighborhood $\U$ of $0$. Then $\Psi(\zeta) f(\zeta) \in H^2\big(\SS_{-\frac{\epsilon}{2\pi},\frac{\epsilon}{2\pi}},\K\big)$.
\end{Lemma}
\begin{proof}
We may assume that $\e = 1$, as the argument is parallel. Under the identif\/ication $L^2(\RR,\K) \cong L^2(\RR)\otimes \K$, we Fourier-transform the f\/irst factor and get $\hat\Psi(k)$, with our convention $\hat \psi(k) = \frac1{\sqrt{2\pi}}\int dt\, e^{-itk}\psi(t)$. The assumption says that $\hat\Psi$ and $\hat\Psi(k)(e^{k}+e^{-k})$ are $L^2$. Consider the following integral for $\zeta$ with $-1 < \im \zeta < 1$ and $0 < \alpha < 1-|\im \zeta|$:
 \begin{gather*}
 \int dk\, \hat\Psi(k) e^{ik\zeta} = \int dk\, \hat\Psi(k) e^{ik\zeta} \big(e^{\alpha k}+e^{-\alpha k}\big) \cdot \frac1{e^{\alpha k}+e^{-\alpha k}}.
 \end{gather*}
In the latter expression, the f\/irst factor $\hat\Psi(k) e^{ik\zeta} (e^{\alpha k}+e^{-\alpha k})$ is $L^2$ by assumption, and the second factor is an~$L^2$-function in~$k$. Therefore, the integral gives a vector in $\K$ for any $\zeta$ in the above interval by Lemma~\ref{lm:riesz}. This shows that the former expression is $L^1$ and it does not depend on~$\alpha$. Let us denote it by $\Psi(\zeta)$. This notation coincides with $\Psi$ when $\zeta$ is real.

We claim that the map $\zeta \mapsto \Psi(\zeta) \in L^2(\RR)$ is analytic in norm. First, it is continuous in norm. Indeed, if $|\zeta-\zeta'| < \frac\a 3$, where $\a$ is a positive number as above, then the integrand of
 \begin{gather*}
 \int dk\, \hat\Psi(k) \big(e^{ik\zeta} - e^{ik\zeta'}\big)\big(e^{\frac\a 3 k}+e^{-\frac\a 3 k}\big) \cdot \frac1{e^{\frac\a 3 k}+e^{-\frac\a 3 k}}
 \end{gather*}
can be bounded by $\big[e^{-\im \zeta k}\big(1+\big(e^{\frac\a 3 k}+e^{-\frac\a 3 k}\big)\big)\|\hat\Psi(k)\|\big(e^{\frac\a 3 k}+e^{-\frac\a 3 k}\big)\big] \cdot \frac1{e^{\frac\a 3 k}+e^{-\frac\a 3 k}}$, whose f\/irst factor is in $L^2(\RR, \K)$. As $\zeta'\to \zeta$, the above integrand converges pointwise, therefore, by Lemma~\ref{lm:l2convergence}, it is $L^2$-convergent, and the integral is convergent as well. As for analyticity in norm, we can use the Morera theorem. For a closed path~$\Gamma$ around $\zeta$ contained in the region $-1 + \frac\a 3 < \im \zeta' < 1- \frac\a 3$, and a vector $\psi \in \K$, we consider
 \begin{gather*}
 \int_\Gamma d\zeta'\int dk\, \<\psi, \hat\Psi(k)\> e^{ik\zeta'} \big(e^{\frac\a 3 k}+e^{-\frac\a 3 k}\big) \cdot \frac1{e^{\frac\a 3 k}+e^{-\frac\a 3 k}}.
 \end{gather*}
The integrand is $L^1(\RR^2 \times \Gamma)$, therefore, we can perform the $d\zeta'$-integration f\/irst and obtain $0$. By the Morera theorem, this $\K$-valued function $\Psi(\zeta)$ is weakly analytic. Then analyticity in norm (strong analyticity) follows by \cite[Theorem~3.31]{Rudin91}.

Now, again for a f\/ixed $\psi$, $\<\psi, \Psi(\zeta)\>$ is an analytic function in one variable $\zeta$. By assumption, this function has a zero at $\zeta = 0$.

As $f(\zeta)$ has only a simple pole at $\zeta = 0$, the product $\Psi(\zeta)f(\zeta)$ remains analytic in norm. Indeed, by def\/inition of analyticity in norm, $\lim\limits_{\zeta \to 0}\frac1{\zeta-0}(\Psi(\zeta)-0) = \lim\limits_{\zeta \to 0}\frac1{\zeta}\Psi(\zeta)$ exists. Therefore, so does the limit $\frac1{\zeta}\Psi(\zeta)\cdot \zeta f(\zeta)$, since $\zeta f(\zeta)$ is analytic by assumption. Especially, the $\K$-valued function $\Psi(\zeta)f(\zeta)$ is bounded around $\zeta = 0$.

Away from $0$, $f(\zeta)$ is bounded by assumption, therefore, as $\Psi(\theta + i\l) \in L^2(\RR\setminus \U, \K)$ for $\k \in (-\epsilon,\epsilon)$, the product is again $L^2$. Altogether, $\Psi(\theta + i\l)f(\theta)$ is uniformly $L^2$.
\end{proof}

\begin{Lemma}\label{lm:symmetric-cont}
Let $\Psi(\theta_1,\theta_2)$ be a $\K$-valued function, $S$-symmetric in $\theta_1$, $\theta_2$ and assume that $\Psi \in \dom\big(\shiftone\otimes\1\otimes\1\big)$, where $L^2(\RR^2,\K) \cong L^2(\RR)\otimes L^2(\RR)\otimes\K$. Then, under the assumption~{\rm (S\ref{ax:regularity})}, $\Psi(\theta_1,\theta_2)$ has actually an analytic continuation in $H^2(\RR^2 + i\C, \K)$,
 where
 \begin{gather*}
 \C = \left\{(\l_1,\l_2)\,|\, -\frac{\pi}3 < \l_1 < 0,\; \l_2 < 0, \;\frac{\pi}3 < \l_1 + \l_2,\;
 -\k < \l_2 - \l_1 \right\}
 \end{gather*}
Especially, $\Psi(\theta_1- \frac{\pi i}6, \theta_2) \in \dom\big(\1 \otimes\halfshiftone\otimes\1\big)$ $($see Fig.~{\rm \ref{fig:analyticity})}.
\begin{figure}[ht]
 \centering
\begin{tikzpicture}[line cap=round,line join=round,>=triangle 45,x=1.0cm,y=1.0cm]
\clip(-6.29,-5.65) rectangle (2.37,1.48);
\fill[line width=0pt,fill=black,fill opacity=0.1] (-5,0) -- (-2.2,-2.8) -- (0,-0.6) -- (0,0) -- cycle;
\draw [domain=-6.29:2.37] plot(\x,{(-0-0*\x)/-5});
\draw (0,-5.65) -- (0,1.48);
\draw (-5,0)-- (0,-5);
\draw (0,0)-- (-2.5,-2.5);
\draw (-2.2,-2.8)-- (0,-0.6);
\draw [->, line width=1.2pt] (-2.5,0)-- (-2.5,-2.5);
\draw (-5.59,0.85) node[anchor=north west] {$-\frac{\pi}3$};
\draw (0.07,-4.51) node[anchor=north west] {$-\frac{\pi}3$};
\draw (1.2,0.58) node[anchor=north west] {$\mathrm{Im}\,\zeta_1$};
\draw (0.07,1.48) node[anchor=north west] {$\mathrm{Im}\,\zeta_2$};
\draw (0.06,0.52) node[anchor=north west] {$0$};
\draw (0.06,-0.33) node[anchor=north west] {$-\kappa$};
\end{tikzpicture}
\caption{The imaginary parts of the domain of analyticity $\C$ (the shaded area). The arrow corresponds to the action of operator.}
 \label{fig:analyticity}
\end{figure}
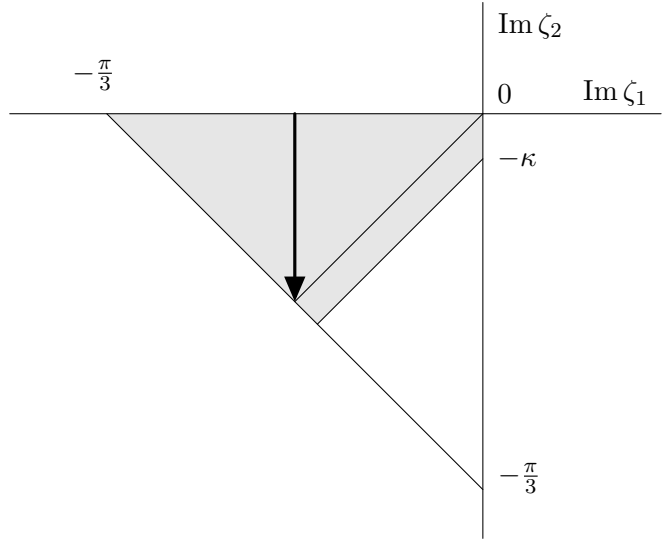
\end{Lemma}
\begin{proof}
By assumption, $\Psi$ has an analytic continuation $\theta_1 \to \theta_1 - \frac{\pi i}3$. Furthermore, $\Psi$ is $S$-symmetric, which means that
 \begin{gather*}
 \Psi(\theta_1, \theta_2) = S(\theta_2 - \theta_1)\Psi(\theta_2, \theta_1).
 \end{gather*}
By assumption (S\ref{ax:regularity}), $S$ has only f\/initely many zeros in the physical strip whose distance for the real line is larger than $0 < \k$ $(< \frac\pi 3)$. Let $\{\a_j\}$ be the poles of $S$, $-\frac\pi 3 < \im \a_j < -\k$ (which correspond to the zeros in the physical strip by (S\ref{ax:unitarity}) and~(S\ref{ax:hermitian})).

We take the function $\psi(\theta) := \prod_j\frac{\sinh\frac12(\theta - \a_j)}{\sinh\frac12(\theta - \b_j)}$, where $\im \b_j > \frac{2\pi }3$. We can take $\{\b_j\}$ so that $\psi$ is bounded below in $-\k < \im \theta < \frac\pi 3$. Indeed, it only has to be a f\/inite Blaschke product in the strip $\RR + i(-\frac\pi 3, \frac{2\pi}3)$ (as $S$ has only f\/initely many zeros in the physical strip).

Now let us consider $\psi(\theta_2-\theta_1)\Psi(\theta_1,\theta_2)$. This function belong to $\dom\big(\shiftone\otimes\1\otimes\1\big)$, since both factors are analytic in $\theta_1, -\frac\pi 3 < \im \theta_1 < 0$. Furthermore,
\begin{gather*}
\psi(\theta_2 - \theta_1)\Psi(\theta_1, \theta_2) = \psi(\theta_2 - \theta_1)S(\theta_2 - \theta_1)\Psi(\theta_2, \theta_1)
\end{gather*}
and the poles of $S(\zeta_2 - \zeta_1)$ in $-\frac{\pi }3 < \im(\zeta_2 - \zeta_1) < 0$ are cancelled by zeros of $\psi(\zeta_2- \zeta_1)$. Then, by looking at the right-hand side, this function belongs to $\dom\big(\1\otimes\shiftone\otimes\1\big)$.

Summarizing, $\psi(\theta_2-\theta_1)\Psi(\theta_1,\theta_2) \in \dom\big(\shiftone\otimes\1\otimes\1\big)\cap\dom\big(\1\otimes\shiftone\otimes\1\big)$ which implies (by \cite[Proposition~A.1]{Tanimoto15-1}) that $\psi(\theta_2-\theta_1)\Psi(\theta_1,\theta_2) \in H^2(\RR^2+i\tilde\C, \K)$, where
 \begin{gather*}
 \tilde\C = \left\{(\l_1,\l_2)\,|\, -\frac{\pi}3 < \l_1 < 0,\; -\frac{\pi}3 < \l_2 < 0,\;
 -\frac{\pi}3 < \l_1 + \l_2 \right\}.
 \end{gather*}
This region is represented by the large triangle in Fig.~\ref{fig:analyticity}. This can be considered as an $L^2$-version of the Malgrange--Zerner theorem (cf.~\cite{Epstein66}).

As $\psi(\theta_2 - \theta_1)$ is analytic and bounded below in $-\k < \im(\theta_2 - \theta_1) < \frac\pi 3$, the function $\Psi(\theta_1,\theta_2) = \psi(\theta_2- \theta_1)\Psi(\theta_1,\theta_2) / \psi(\theta_2-\theta_1)$ is analytic and uniformly $L^2$-bounded in the claimed region. In other words, $\Psi \in H^2(\RR^2+i\C, \K)$.
\end{proof}

\begin{Lemma}\label{lm:symmetric-cont2}
With the same assumption as in Lemma~{\rm \ref{lm:symmetric-cont}}, $S\big(\theta_1 - \theta_2 + \frac{\pi i}3\big)\Psi(\theta_1, \theta_2)$ is in $H^2(\RR^2 + i\C_\epsilon, \K)$ for any $\epsilon > 0$, where
 \begin{gather*}
 \C_\epsilon = \left\{(\l_1,\l_2)\,|\, -\frac{\pi}3 < \l_1 < 0,\; \l_2 < 0, \;\frac{\pi}3 < \l_1 + \l_2 < -\epsilon,\;
 -\k \le \l_2 - \l_1 \right\}.
 \end{gather*}
Especially, $S\big(\theta_1-\theta_2 + \frac{\pi i}6\big)\Psi\big(\theta_1- \frac{\pi i}6, \theta_2\big) \in \dom\big(\1 \otimes\halfshiftone\otimes\1\big)$.
\end{Lemma}
\begin{proof}
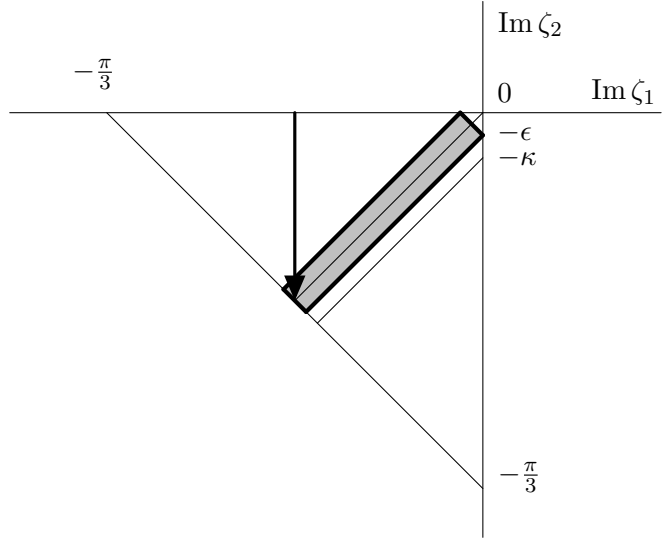
\begin{figure}[ht]
 \centering
\begin{tikzpicture}[line cap=round,line join=round,>=triangle 45,x=1.0cm,y=1.0cm]
\clip(-6.29,-5.65) rectangle (2.37,1.48);
\fill[line width=1.6pt,fill=black,fill opacity=0.25] (-2.65,-2.35) -- (-2.35,-2.65) -- (0,-0.3) -- (-0.3,-0) -- cycle;
\draw [domain=-6.29:2.37] plot(\x,{(-0-0*\x)/-5});
\draw (0,-5.65) -- (0,1.48);
\draw (-5,0)-- (0,-5);
\draw (0,0)-- (-2.5,-2.5);
\draw (-2.2,-2.8)-- (0,-0.6);
\draw (-5.59,0.85) node[anchor=north west] {$-\frac{\pi}3$};
\draw (0.07,-4.51) node[anchor=north west] {$-\frac{\pi}3$};
\draw (1.31,0.58) node[anchor=north west] {$\mathrm{Im}\,\zeta_1$};
\draw (0.07,1.48) node[anchor=north west] {$\mathrm{Im}\,\zeta_2$};
\draw (0.06,0.52) node[anchor=north west] {$0$};
\draw (0.06,-0.33) node[anchor=north west] {$-\kappa$};
\draw (0.06,0) node[anchor=north west] {$-\epsilon$};
\draw [line width=1.6pt] (-2.65,-2.35)-- (-2.35,-2.65);
\draw [line width=1.6pt] (-2.35,-2.65)-- (0,-0.3);
\draw [line width=1.6pt] (0,-0.3)-- (-0.3,0);
\draw [line width=1.6pt] (-0.3,0)-- (-2.65,-2.35);
\draw [->,line width=1.2pt] (-2.5,0) -- (-2.5,-2.5);
\end{tikzpicture}
 \caption{The tilted rectangular which covers the arrowhead.}
 \label{fig:analyticity2}
\end{figure}
 We know from Lemma~\ref{lm:symmetric-cont} that $\Psi \in H^2(\RR^2 + i\C,\K)$. Then we can apply Lemma~\ref{lm:zero} to $S(\theta_1 - \theta_2 + \frac{\pi i}3)\Psi(\theta_1,\theta_2)$ as a function of $\theta_2-\theta_1$ on a tilted rectangular around on the segment between $(0,0)$ and $(-\frac\pi 6,-\frac\pi 6)$ with width $\sqrt 2\epsilon$, which gives the $L^2$-boundedness on the rectangular. The function $S(\theta_1-\theta_2 + \frac{\pi i}3)$ is bounded in $\RR^2 + i\C$ away from the line $\theta_2-\theta_1 = 0$,
 hence we obtain the claimed boundedness.
\end{proof}

\begin{Proposition}\label{pr:analyticity-n}
 Let $\Psi_n$ be in the following domain:
 \begin{gather*}
 \dom\big(\1\otimes\cdots \otimes\1\otimes\underset{k\text{\rm -th}}{\big(\Delta_1^{\frac{\epsilon}{2\pi}} + \Delta_1^{-\frac{\epsilon}{2\pi}}\big)}\otimes\1\otimes\cdots\otimes\1\big)\\
 \qquad{} \bigcap_{l\neq k} \dom\big(\1\otimes\cdots\otimes\1\otimes\underset{l\text{\rm -th}}{\Delta_1^{\frac\epsilon{2\pi}}}\otimes\1\otimes\cdots\otimes\1\big)
 \end{gather*}
for some $0<\epsilon <\kappa$ and assume that $\Psi_n$ has zeros at $\theta_k -\theta_l = 0$ for $k\neq l$. Let $\{f_l\}$, $l = 1,\dots, n$, be $\CC$-valued meromorphic functions on $\RR + i(-\frac{\k}{2\pi},\frac{\k}{2\pi})$ with the only simple pole at $\zeta = 0$ and bounded outside a neighborhood of $\zeta = 0$. Then for $1 \le k \le n$, $\prod\limits_{l\neq k}f_l(\theta_k - \theta_l)\Psi_n(\theta_1,\dots,\theta_n)$ belongs to the same domain.
\end{Proposition}
\begin{proof}
By looking at the variable $\zeta_k$ and $\zeta_l$, we obtain the analyticity in the large triangle $\C$ in Fig.~\ref{fig:analyticity-n}. By changing the variable to $(\zeta_k-\zeta_l,\zeta_k+\zeta_l)$, and considering the small rectangle $\mathcal R_{\epsilon'}$, we can apply Lemma \ref{lm:zero} to conclude that for $\epsilon' > 0$, $f_l(\theta_k - \theta_l)\Psi_n(\theta_1,\dots,\theta_n)$ has a continuation to an element in $H^2(\RR^2+i\mathcal R_{\epsilon'},L^2(\RR^{n-2}))$, as $\Psi_n$ has a zero at $\theta_k-\theta_l = 0$ by assumption.

\begin{figure}[ht]
 \centering
\begin{tikzpicture}[line cap=round,line join=round,>=triangle 45,x=1.0cm,y=1.0cm]
\clip(-4.11,-0.79) rectangle (4.47,4.34);
\fill[fill=black,fill opacity=0.25] (0,0.5) -- (0.5,0) -- (1.75,1.25) -- (1.25,1.75) -- cycle;
\fill[fill=black,fill opacity=0.1] (0,3) -- (-2.5,0.5) -- (2.5,0.5) -- cycle;
\draw [domain=-4.11:4.47] plot(\x,{(-0-0*\x)/-3});
\draw (0,-0.79) -- (0,4.34);
\draw (3.36,0.58) node[anchor=north west] {$\mathrm{Im}\,\zeta_k$};
\draw (0.05,4) node[anchor=north west] {$\mathrm{Im}\,\zeta_l$};
\draw (0,0) node[anchor=north west] {$0$};
\draw (-3,0)-- (0,3);
\draw (0,3)-- (3,0);
\draw [line width=1.6pt] (0,0)-- (1.5,1.5);
\draw (0.5,0)-- (1.75,1.25);
\draw (0,0.5)-- (1.25,1.75);
\draw (1.25,1.75)-- (1.75,1.25);
\draw (0,0.5)-- (0.5,0);
\draw (0.5,0)-- (1.75,1.25);
\draw (1.75,1.25)-- (1.25,1.75);
\draw (1.25,1.75)-- (0,0.5);
\draw (-2.5,0.5)-- (2.5,0.5);
\draw (2.8,0) node[anchor=north west] {$\epsilon$};
\draw (-3.5,0.05) node[anchor=north west] {$-\epsilon$};
\draw (0.2,3.2) node[anchor=north west] {$\epsilon$};
\draw (0.49,0.1) node[anchor=north west] {$\epsilon'$};
\draw (-0.9, 1.1) node[anchor=north west] {$-\epsilon'$};
\end{tikzpicture}
\caption{The domain of analyticity of $\Psi_n$ (the larger triangle $\C$), the pole of $S$ (the thick segment), an auxiliary rectangle $\mathcal R_{\epsilon'}$ (dark) and an intermediate domain $\C_{\epsilon'}$ (the shaded triangle).} \label{fig:analyticity-n}
\end{figure}
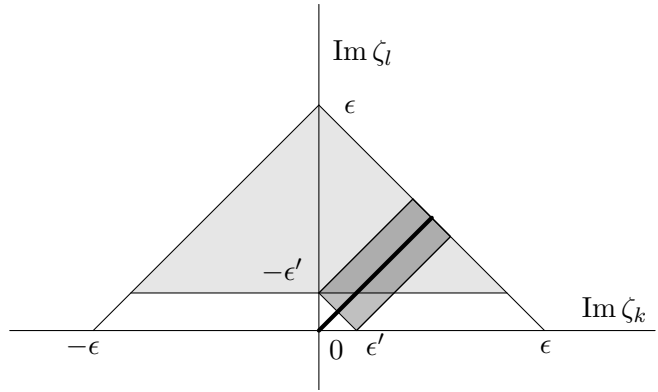

Furthermore, as $f_l$ has no other pole, it continues further to $\RR^2+i\mathcal \C_{\epsilon'}$. By $\K$-valued three line theorem (cf.~\cite[Theorem~12.9]{Rudin87}), we obtain that the $H^2$-norm as an element of $H^2(\RR^2+i\mathcal \C_{\epsilon'},L^2(\RR^{n-2}))$ is determined at the corners of the triangle $\C_{\epsilon'}$ which have f\/inite distance from the line $\theta_k- \theta_l = 0$. As $S$ is bounded at these corners and the edges, we obtain an $H^2$-bound which is uniform in $\epsilon'$. Finally, as $\epsilon' < \epsilon$ is arbitrary, actually we obtain that $f_l(\theta_k - \theta_l)\Psi_n(\theta_1, \dots,\theta_n) \in H^2(\RR^2 + i\C, L^2(\RR^{n-2}))$. This is equivalent to the vector $f_l(\theta_k - \theta_l)\Psi_n(\theta_1,\dots,\theta_n)$ being in the same domain, as the analyticity in the other variables are not af\/fected.

Repeating this $n-1$ times, we obtain the claim.
\end{proof}

\subsubsection*{Acknowledgements}
I owe some ideas of Lemmas~\ref{lm:sum-closed} and~\ref{lm:cleave} to Henning Bostelmann and the counterexample in Section~\ref{cef} to Ludvig D.~Faddeev. I am grateful to Marcel Bischof\/f, Henning Bostelmann, Detlev Buchholz, Daniela Cadamuro, Sebastiano Carpi, Wojciech Dybalski, Luca Giorgetti, Gandalf Lechner, Roberto Longo, Karl-Henning Rehren and Bert Schroer for their interesting discussions and encouraging comments. I appreciate the careful reading and useful suggestions by the referee.
I am supported by Grant-in-Aid for JSPS fellows 25-205.

\pdfbookmark[1]{References}{ref}
\LastPageEnding

\end{document}